%% file: arxiv.tex
\documentclass[10pt,journal,compsoc]{IEEEtran}

\usepackage[T1]{fontenc}
\usepackage[utf8]{inputenc}
\usepackage{lmodern}
\usepackage{amsmath,amssymb}
\usepackage{amsthm}
\usepackage{graphicx}
\usepackage{booktabs}
\usepackage{url}

\graphicspath{{figures/}}

\renewcommand{\thetable}{\arabic{table}}

\newtheoremstyle{suppresult}{\smallskipamount}{\smallskipamount}{\normalfont}{\parindent}%
  {\bfseries}{.}{.5em}{\thmname{#1}\thmnote{ (#3)}}
\newtheoremstyle{suppprelim}{\smallskipamount}{\smallskipamount}{\normalfont}{\parindent}%
  {\bfseries}{.}{.5em}{\thmname{#1}\thmnote{ #3}}
\newtheoremstyle{suppremark}{\smallskipamount}{\smallskipamount}{\normalfont}{\parindent}%
  {\itshape}{.}{.5em}{\thmname{#1}\thmnote{ (#3)}}
\newcommand{\suppresultname}{}
\theoremstyle{suppresult}
\newtheorem*{suppnamedresult}{\suppresultname}
\newenvironment{proposition}[1]
  {\renewcommand{\suppresultname}{Proposition #1}\begin{suppnamedresult}}{\end{suppnamedresult}}
\newtheorem*{corollary}{Corollary}
\newtheorem*{statement}{Statement}
\newtheorem*{claim}{Claim}
\newtheorem*{setting}{Setting}
\newtheorem*{assumption}{Assumption}
\newtheorem*{assumptions}{Assumptions}
\theoremstyle{suppprelim}
\newtheorem*{suppprelimresult}{\suppresultname}
\makeatletter
\newenvironment{prelim}[1]
  {\renewcommand{\suppresultname}{(Pre-#1)}\def\@currentlabel{Pre-#1}\begin{suppprelimresult}}%
  {\end{suppprelimresult}}
\makeatother
\theoremstyle{suppremark}
\newtheorem*{remark}{Remark}

\begin{document}

\title{DOMIC: Provably Calibrated Detection of Dependence-Structure Change Points via Density-Operator Mutual Information}

\author{Jiwon~Kang and Yun~Am~Seo
\thanks{J. Kang is with the Department of Data Science, Jeju National University,
Jeju 63243, Republic of Korea (ORCID 0009-0000-0675-0572).}%
\thanks{Y. A. Seo is with the Department of Data Science, Jeju National University,
Jeju 63243, Republic of Korea, and with the NAVI Hyper-Tropicalization Research
Institute, Jeju National University, Jeju 63243, Republic of Korea
(ORCID 0000-0001-9283-4376).}%
\thanks{Corresponding author: Yun Am Seo (e-mail: seoya@jejunu.ac.kr).}}

\maketitle

\begin{abstract}
\input{abstract}
\end{abstract}


\input{body}

\clearpage
\onecolumn
\renewcommand{\thesection}{\Alph{section}}
\renewcommand{\thesubsection}{\thesection.\arabic{subsection}}
\renewcommand{\thesectiondis}{Supplement~\thesection.}
\renewcommand{\thesubsectiondis}{\thesubsection}
\setcounter{section}{0}
\begin{center}{\Large Supplementary Material}\end{center}
\noindent\emph{Section, equation, table, figure, and citation numbers of the form ``Section 4.3'', ``(8)'', or ``[28]'' refer to the main article above; sections of the form A.x and B.x, and figures and tables labelled B.x, are internal to this Supplementary Material.}

\input{appendixA}

\setcounter{figure}{0}
\setcounter{table}{0}
\renewcommand{\thefigure}{B.\arabic{figure}}
\renewcommand{\thetable}{B.\arabic{table}}
\input{appendixB}

\end{document}

%% file: abstract.tex
Dependence can change between variable blocks without changing correlation; serial dependence complicates calibration. We propose DOMIC, a rank-based detector using density-operator mutual information (DOMI). Unit-norm random features make sample-level partial traces exact, so prefix sums compute segment statistics at per-split cost independent of segment length; a Gram form uses the exact kernel. Permutation tests are finite-sample exact under pair or block exchangeability for windows and schemes fixed in advance; validity after data-driven scheme selection is unproved. For i.i.d. pairs with feature dimension, frequency draw and bandwidth fixed, we prove a weighted chi-square limit at independence for the rank-based random-feature statistic. An additive entropy cost with Holevo gains drives segmentation. DOMI attains higher localised power than rank, copula, kernel and distance baselines in most non-Gaussian and correlation-free settings; the Gram form leads or ties in all 16 non-Gaussian settings where any statistic exceeds 0.10. On a three-break benchmark, Holevo partitioning returns exactly three breaks in 47.5\% of replicates, against at most 19\% for calibrated baselines. Block-calibrated primary scans reject in all three pre-specified pairs of a US financial panel, with strongest stage-two evidence for stocks and Treasury yields. One of 27 weather candidates passes the multiplicity-corrected re-test; candidate selection is unadjusted.

%% file: body.tex
\section*{1. Introduction}
\addcontentsline{toc}{section}{1. Introduction}

Change-point analysis is mature for changes in the mean or variance of a time series. Changes in the \emph{dependence structure} between components of a multivariate series concern which variables move together, and in what manner. Such changes matter in finance (co-movements strengthening in crises), in climate science (seasonal shifts in temperature--humidity coupling) and in physiology (vital signs decoupling before adverse events). They are harder to detect: dependence may change without any change in correlation, for example through a coupling of amplitudes or tails, and in applications the margins may drift at the same time, so a useful dependence-change statistic must be \emph{specific}: it should not reject under purely marginal changes. This paper proposes DOMIC (density-operator mutual information for change points), an offline, rank-based detector of dependence changes that combines a permutation scan, finite-sample exact for exchangeability groups fixed in advance, with segment statistics computed exactly from prefix sums and an additive cost for multiple breaks.

DOMIC is built on the \textbf{density operator} of a segment: a positive semi-definite, trace-one matrix \(\rho\) obtained by trace-normalising the second moment of a feature embedding of the data. It needs no density estimation or matrix inversion, remains defined when the feature moment is singular, and carries a calculus (von Neumann entropy, Holevo information, partial traces) that reduces to the Shannon calculus for commuting operators (Lemma 1). The density operator also admits a dual Gram-matrix representation: \(\rho\) has the same non-zero spectrum as the trace-normalised Gram matrix of the features, and with the exact kernel its order-one entropies are the matrix-based entropies of {[}14,15{]} (Section 4.8). DOMIC takes the complementary primal route. With unit-norm features the \textbf{partial trace}, which maps a joint state \(\rho_{XY}\) to its marginal \(\rho_X\), is exact at the sample level, so the density-operator mutual information (DOMI) \(I(X{:}Y)=S(\rho_X)+S(\rho_Y)-S(\rho_{XY})\) of every segment is read from prefix sums of small moment matrices, and the limit law developed here is stated for that statistic. Kernel mean-embedding procedures, which dominate nonparametric change-point detection {[}28{]}, use different statistics. All computation is classical; no quantum hardware is involved.

Concretely, we transform each of two variable blocks to ranks
(pseudo-observations), embed the ranks by unit-norm random Fourier
features, and form per-sample pure product states whose segment average
is a separable joint density operator. Ranks make the statistic invariant to strictly increasing marginal transformations over the whole window (Lemma 4(ii)). Figure 1 summarises the procedure.

\begin{figure*}[!t]
\centering
\includegraphics[width=0.85\textwidth,height=0.42\textheight,keepaspectratio]{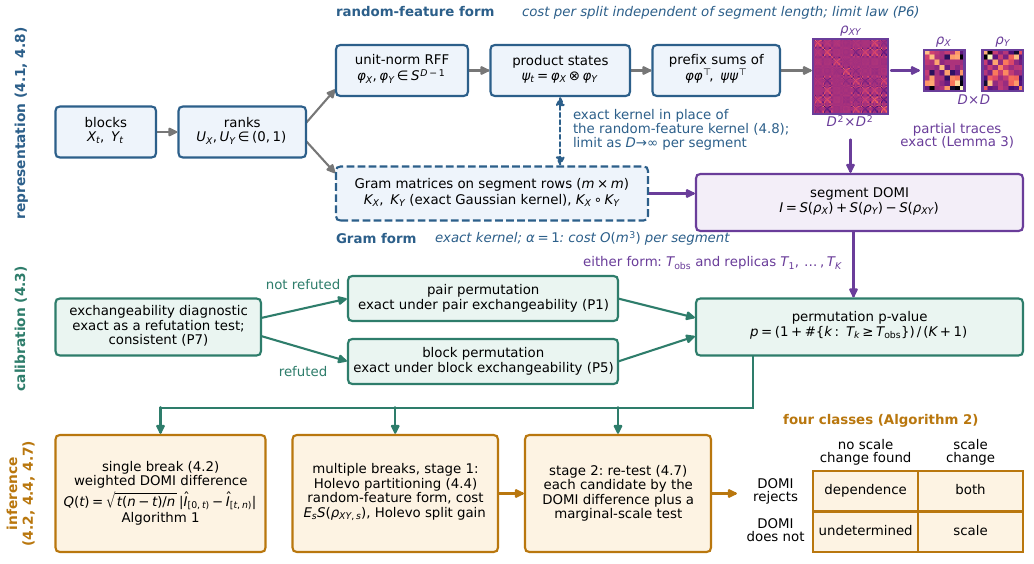}
\caption{Overview of DOMIC in three bands. Top: ranked blocks give the segment DOMI in two computational forms, the random-feature form (solid) and the Gram form (dashed). Middle: an exchangeability diagnostic selects pair or block permutation, which calibrates either form by a permutation p-value; exactness holds for a window and scheme fixed in advance and is not established after the diagnostic's selection; stage-two p-values are nominal. Bottom: the single-break test, and segmentation by Holevo partitioning followed by a re-test that sorts the candidates into four classes.}
\end{figure*}

Two recent segmentation methods are closest in aim: BMCTC {[}37{]}, a moving-cut total-correlation method, segments a moving-window matrix-based total correlation by the Bernaola-Galván rule, and a copula-entropy method {[}39{]} tests equality of whole segment distributions against a fixed threshold; neither provides the finite-sample calibration of a scan established here, and DOMIC, unlike the second, targets dependence changes alone.

DOMIC fits no model of the data-generating process: it is a \emph{statistic} computed on a segment, with a calibration that attaches a p-value and a cost that drives a segmentation, and its comparators are the model-free, retrospective procedures of Section 2, not autoregressive or generalised autoregressive conditional heteroskedasticity (GARCH) models; model-based monitors are compared in Supplementary Section B.4. Because the temporal dynamics are not modelled, the null hypothesis is one of exchangeability, which serial dependence violates; Section 4.3 quantifies the cost and gives a block-permutation calibration that restores the level to near-nominal under short-range serial dependence.

\textbf{(C1) A per-split cost independent of segment length.} Because unit-norm features make the partial trace exact at the sample level (Lemma 3), every segment DOMI of the random-feature form is read from prefix sums of \(D^2\times D^2\) and \(D\times D\) moment matrices, at a cost per candidate split that does not depend on the segment length (Section 4.6). The Gram form, which replaces the random-feature kernel by the exact Gaussian kernel and is the matrix-based mutual information of {[}15{]} on ranks, is more powerful at moderate length and costs \(O(m^3)\) per segment (Section 4.8).

\textbf{(C2) A limit law for the rank statistic at a fixed bandwidth.} For i.i.d. pairs at the independence boundary, with \(D\), the frequency draw and the bandwidth fixed, the random-feature form on ranks taken once over the window has a weighted-\(\chi^2\) limit at rate \(m\) at each fixed split, with a corresponding process limit over the split fractions (P6, P6-R). Simulating that limit calibrates, without permutations, a scan whose null is independence throughout (Section 4.6); the limit is proved at a fixed bandwidth and checked empirically at the deployed one, and it does not calibrate a test between two dependent regimes.

\textbf{(C3) Calibration and segmentation.} Under the null of no change we assume exchangeability, of the pairs of a window or of blocks of consecutive observations. For a window and scheme fixed in advance, pair and block permutation calibrate the scan exactly in finite samples (P1, P5), a general property of permutation tests {[}25{]} that holds for any statistic treated symmetrically. In the procedure an exchangeability diagnostic, consistent against stationary ergodic serial dependence visible to its lag statistics (P7), chooses between the two schemes; exactness after this data-driven choice is not established. The energy-weighted von Neumann cost is additive over segments, with split gain \(E_s\) times a Holevo information, so penalised optimal partitioning applies directly; we call the resulting algorithm Holevo partitioning. Being a functional of the joint state, this cost also responds to marginal change; re-testing each candidate with the DOMI difference lowers the fraction of replicates with a false detection on a marginal-only design from 0.975 to 0.13 (Section 4.7 and Supplementary Table B.7).

\textbf{(C4) Evidence.} On identical rank inputs and under one calibration, DOMI exceeds kernel, distance-correlation, rank and copula statistics by at least 15 percentage points of localised power on the correlation-free alternatives S2 and S4 wherever it or one of these statistics exceeds 0.10, at the reported calibration draw; across three calibration draws the margin on S4 against HSIC ranges from 4.6 to 40.4 points (Section 6.6). On six non-Gaussian dependence changes with standard normal margins the Gram form is the most powerful statistic of Table 1, alone or tied, in all 16 settings in which any statistic exceeds 0.10, and under the tested marginal scale changes the DOMI difference keeps detection rates of 0.032--0.058 (Section 6.2). No method is uniformly best: the copula test of {[}10{]} detects the copula-shape change S3, BMCTC {[}37{]}, calibrated alike, is more powerful on the sign-mixed change S4, and correlation statistics are more powerful than the random-feature form on monotone correlation. At the stronger level of the three-break benchmark MB, Holevo partitioning returns exactly three breaks in 47.5\% of replicates, against at most 19\% for kernel change-point, e.divisive and matrix-information baselines with penalties or significance levels calibrated alike to a nominal 5\% null detection rate (Section 6.3). On serially dependent records, a season-restricted calibration lowers the random-date rejection rate of the dependence test on hourly weather from 9.7\% to 5.4\% at 0.05 (Section 6.7), and block-calibrated primary scans reject in all three pre-specified pairs of a US financial panel (Section 6.8).

\section*{2. Related work}
\addcontentsline{toc}{section}{2. Related work}

\textbf{Change-point detection.} Classical parametric segmentation is
dominated by likelihood costs with penalised dynamic programming,
notably PELT {[}1{]}, and by CUSUM-type statistics; wild binary
segmentation {[}2{]} addresses multiscale localisation.
Covariance-structure breaks are treated in {[}3{]} and {[}4{]}.
Nonparametric approaches include kernel change-point analysis {[}5,6{]},
energy-distance methods {[}7{]}, MMD scan statistics {[}8{]}, classifier- or ratio-based tests such as KL-CPD {[}9{]}, and autoencoder-based detectors such as MC-TIRE {[}38{]}; see {[}28{]} for a survey. All of these target changes in the joint or marginal
distribution; none is specific to dependence.

\textbf{Dependence and copula change points.} Rank-based approaches to dependence-change detection include sequential empirical-copula Cramér--von Mises tests {[}10,11{]} and tests accommodating marginal changes at \emph{known}
instants {[}12{]}. A recent kernel approach detects changes in the
causal dependence of one variable on another, given a designated set of
confounders, through mean embeddings of the corresponding conditional
copulas {[}13{]} --- the prior work closest to ours in problem
formulation. Their statistic is a first-moment (mean-embedding) object; ours is the von Neumann entropy of the second moment. For a joint-state statistic on raw-data features, replacing the entropy by a Frobenius norm costs power on dependence and rotation alternatives (Supplementary Section B.2); this ablation is not run on DOMI itself. We do not run the method of {[}13{]} as a baseline, because its estimand requires a conditioning set that the bivariate designs of Section 5 do not supply.
The test of {[}10{]} ranks the observations within each segment, which keeps its false-alarm rate near nominal under marginal scale changes; the same windowed statistic loses this property when computed from global pseudo-observations (Section 6.2). The copula-entropy method of {[}39{]} responds to marginal changes too, since its null is equality of segment distributions.

\textbf{Model-based dependence monitors.} A parallel econometric
tradition models the dependence parameter itself as time-varying or regime-switching, as DCC-GARCH does {[}29{]}; the CUSUM test of Wied et al.~{[}30{]} addresses
constancy of the correlation coefficient directly. The time-varying and switching models track a chosen parametric dependence measure rather than test for breaks with error control, the CUSUM test targets the correlation coefficient alone, and the correlation-based implementations evaluated here are structurally insensitive to dependence changes that leave the correlation unchanged (Supplementary Section B.4).

\textbf{Matrix-based and kernel information theory.} Matrix-based Rényi
entropy from normalised Gram matrices {[}14{]}, its multivariate
mutual-information extension {[}15{]}, and the kernel von Neumann
information theory of Bach {[}16{]} define information measures on kernel second moments. DOMI belongs to this class. By the duality of Section 4.8,
DOMI computed with the exact Gaussian kernel is the matrix-based mutual
information of {[}15{]} at order \(\alpha=1\), evaluated on ranks, and
the random-feature form of Section 4.1 approximates it. Qian et al.
{[}37{]} apply matrix-based Rényi total correlation to abrupt changes in
the dependence among hydrological series: a moving-window sequence of
the measure is segmented by the Bernaola-Galván algorithm, whose
significance threshold is derived for independent observations. A
moving-window sequence is autocorrelated by construction, so the
threshold is not guaranteed to control the false-alarm rate of the scan;
in our reimplementation it rejects in every null replicate
(Supplementary Section B.13). That work does not examine changes confined to the margins and evaluates a full Gram matrix in every window; to our knowledge no earlier change-point method for this class controls the false-alarm rate in finite samples. On the quantum side,
``quantum change-point'' detection {[}17{]} concerns sequences of
physical quantum states, and quantum-kernel estimation of mutual information {[}18{]} targets static independence testing.

\section*{3. Density-operator statistics for change
points}
\addcontentsline{toc}{section}{3. Density-operator statistics for change
points}

\subsection*{3.1 Segment density
operators}
\addcontentsline{toc}{subsection}{3.1 Segment density operators}

Let \(z_1,\dots,z_n\in\mathbb{R}^p\) be a feature sequence (specified in Section 4). For a segment
\(s=[a,b)\) define the second moment, the energy and the \textbf{density
operator}

\begin{equation}\begin{split}M_s&=|s|^{-1}\sum_{t\in s}z_tz_t^{\top},\\ E_s&=\sum_{t\in s}\lVert z_t\rVert^2=|s|\cdot\mathrm{tr}\,M_s,\\ \rho_s&=M_s/\mathrm{tr}\,M_s.\end{split}\end{equation}

Every \(\rho_s\) is positive semi-definite with unit trace and rank at
most \(\min(|s|,p)\); no inversion or regularisation is involved, so the
construction is well-defined in the singular regime \(|s|<p\). Prefix
sums of \(z_tz_t^{\top}\) give any segment's \(M_s\) in \(O(p^2)\) time.

If \(s=L\sqcup R\)
then the energies \(E_s\) and the unnormalised moments \(|s|M_s\) are
additive, and

\begin{equation}\rho_s = q_L\rho_L + q_R\rho_R,\qquad q_i = E_i/E_s,\end{equation}
the \emph{mixture identity}.

Let \(S(\rho)=-\mathrm{tr}\,\rho\log\rho\) be the von Neumann entropy.

\textbf{Lemma 1 (reduction; Supplementary Section A.1).} If \([\rho_1,\rho_2]=0\)
with a common eigenbasis and eigenvalue vectors \(\lambda_1,\lambda_2\),
then the quantum Jensen--Shannon divergence {[}21{]}, squared Bures
distance, and quantum relative entropy of \((\rho_1,\rho_2)\) equal the
classical Jensen--Shannon divergence, twice the squared Hellinger
distance, and the Kullback--Leibler divergence of
\((\lambda_1,\lambda_2)\), respectively.

Here \(\lambda_1,\lambda_2\) are the coordinates of \(\rho_1\) and \(\rho_2\) in the shared basis, paired index by index, not the independently sorted spectra \(\lambda^{\downarrow}\) of Lemma 2; behaviour beyond the classical case arises from non-commutativity or from a mismatch in that ordering.

\textbf{Lemma 2 (spectral-alignment decomposition; Supplementary Section A.2).} For any density
operators \(\rho_1,\rho_2\),
\(\mathrm{QJSD}(\rho_1,\rho_2)=\mathrm{JSD}(\lambda^{\downarrow}(\rho_1),\lambda^{\downarrow}(\rho_2))+\Delta_{\mathrm{nc}}(\rho_1,\rho_2)\)
with \(\Delta_{\mathrm{nc}}\ge 0\), where \(\lambda^{\downarrow}\)
denotes the sorted spectrum; \(\Delta_{\mathrm{nc}}=0\) if and only if
the two operators commute with similarly ordered spectra.

\(\Delta_{\mathrm{nc}}\) registers a change in \emph{which} directions carry the second moment, not in its sorted spectrum. In the ablation of Supplementary Section B.2 the spectral term alone is powerless on rotation and markedly weaker on dependence without correlation; the gain from \(\Delta_{\mathrm{nc}}\) is not uniform across alternative types.

\subsection*{3.2 A segment-separable entropy cost and the Holevo split
gain}
\addcontentsline{toc}{subsection}{3.2 A segment-separable entropy cost
and the Holevo split gain}

Define the segment cost \(C(s)=E_s\,S(\rho_s)\). The segmentation
objective is the sum of \(C\) over the segments of a partition (the
property optimal partitioning requires), and by (2) and concavity of
\(S\),

\begin{equation}\begin{split}C(s)-C(L)-C(R)&=E_s\,\chi,\\ \chi &= S(q_L\rho_L+q_R\rho_R)-q_LS(\rho_L)\\&\quad-q_RS(\rho_R)\ \ge\ 0,\end{split}\end{equation}
with equality iff \(\rho_L=\rho_R\). The split gain is \(E_s\) times a \textbf{Holevo information} {[}20{]} --- the quantum bound on extractable classical information about the segment label; for segments of equal energy it is the quantum Jensen--Shannon divergence of \(\rho_L\) and \(\rho_R\) (Section 3.4). \(C\) can therefore serve as the segment cost of penalised optimal partitioning, playing the role that \(|s|\cdot\log\det\hat\Sigma_s\) plays for Gaussian likelihood, but bounded, inversion-free, and defined for singular moments.

\subsection*{3.3 Plug-in bias and its
removal}
\addcontentsline{toc}{subsection}{3.3 Plug-in bias and its removal}

\(S(\hat\rho_s)\) is biased downward on short segments, so raw costs of different lengths are not comparable and a raw DOMI plug-in is biased upward under independence. The segmentation cost is corrected by subtracting a mean over permuted copies (Section 4.4); the single-break test is not corrected (Section 4.2).

\subsection*{3.4 Why entropy rather than
Frobenius}
\addcontentsline{toc}{subsection}{3.4 Why entropy rather than Frobenius}

The natural ``classical'' statistic on the same second moment is
\(\lVert M_L-M_R\rVert_F\), equivalent to an MMD with a squared kernel.
A local expansion identifies the directions in which the entropy
functional is more sensitive: for a perturbation \(\delta\) of \(\rho\),

\begin{equation}\mathrm{QJSD}(\rho,\rho+\delta)=\tfrac18\sum_{ij}|\delta_{ij}|^2/L(\lambda_i,\lambda_j)+o(\lVert\delta\rVert^2)\end{equation}
in the eigenbasis of \(\rho\), where \(L\) is the logarithmic mean of
the eigenvalue pair (Supplementary Section A.9); the three functionals
weight the direction \((i,j)\) by the reciprocal logarithmic mean
(QJSD), the reciprocal arithmetic mean (Bures/quantum Fisher) and a constant (Frobenius; Supplementary Figure B.3).

Since \(L(a,b)\le(a+b)/2\), the entropic weighting amplifies highly asymmetric eigenvalue pairs, the low-eigenvalue directions, more than either alternative, though only logarithmically in the eigenvalue ratio; the expansion is local and assumes \(\rho\succ 0\) (Supplementary Section A.9). The expansion motivates the entropy functional without showing that the targeted alternatives perturb those directions, so the case is empirical: in the joint-state comparison of Supplementary Section B.2 it is more powerful on dependence and rotation alternatives, though not on every type.

\subsection*{3.5 Statistical guarantees}
\addcontentsline{toc}{subsection}{3.5 Statistical guarantees}

Seven results, P1--P7, establish the guarantees of the calibration and of the statistic (the \(D\to\infty\) half of P4 only conditionally on a spectral-decay hypothesis); P5 and P7 are stated in Section 4.3 with the calibration they concern. Proofs are in Supplementary Sections A.5--A.8 and A.11--A.13, with standard operator-theoretic tools from {[}26,27{]}; Supplementary Section B.3 gives finite-sample illustrations and simulation studies of properties these proofs do not establish. \textbf{(P1, finite-sample exactness.)} If, under the no-change null, the pairs \((X_t,Y_t)\) of a window fixed in advance are exchangeable, then for
any statistic treating the observed series and \(K\) jointly permuted
replicas symmetrically --- a studentised maximum whose per-\(t\) moments
are computed from all \(K{+}1\) curves qualifies --- the permutation
p-value \((1+\#\{T_k\ge T_{\mathrm{obs}}\})/(K+1)\) satisfies
\(P(p\le\alpha)\le\alpha\) exactly {[}25{]} (Supplementary Section A.5). The pair-permutation tests reported here use that symmetric form; the block-permutation tests rest on P5.

\textbf{(P2, asymptotic normality; Supplementary Section A.6.)} For a stationary \(\alpha\)-mixing series with fixed \(D\), bounded features and a full-rank population moment, the delta method yields \(\sqrt{m}\)-normality of the segment DOMI plug-in; a process-level version (P2$'$) is a theorem for i.i.d. segments and supplies the tightness used in P3. Both are stated under the oracle marginal transform, so their variance need not be that of the rank statistic we compute. At the independence null the influence function vanishes identically, \(\sigma^2=0\), and the correct normalisation is \(m\), not \(\sqrt{m}\); the limit is then the weighted-\(\chi^2\) law of Proposition P6.

\textbf{(P6, degenerate limit law at the independence boundary; oracle
margins.)} Let the pairs be i.i.d. with independent components, keep
\(D\) and the unit-norm feature maps fixed, apply them to the true
marginal transforms, and assume \(\rho_X\) and \(\rho_Y\) nonsingular.
Then the first-order term of the entropy expansion of \(\hat I\)
vanishes \emph{identically} --- a sample-path identity, because Lemma 3
makes \(\mathrm{Tr}_Y\hat\rho_{XY}=\hat\rho_X\) exact and
\(\log(\rho_X\otimes\rho_Y)\) splits as
\(\log\rho_X\otimes I+I\otimes\log\rho_Y\) --- and \(m\,\hat I_m\)
converges in distribution to
\(\tfrac12\big[q_{\rho_X\otimes\rho_Y}(G)-q_{\rho_X}(\mathrm{Tr}_Y G)-q_{\rho_Y}(\mathrm{Tr}_X G)\big]\),
where \(G\) is the Gaussian limit of
\(\sqrt m(\hat\rho_{XY}-\rho_X\otimes\rho_Y)\) and \(q_\rho\) is the
quadratic form with the reciprocal-logarithmic-mean kernel of Section
3.4: a weighted sum of independent \(\chi^2_1\) variables with
non-negative weights, at rate \(m\). Those weights sum to the bias
constant of Lemma 4(iv), whose strict positivity at independence is
computed in closed form at every configuration used here rather than
proved in general (Supplementary Sections A.10 and A.12). A corollary in Supplementary Section
A.12 extends the limit to the process level over split fractions ---
jointly over the two segments, hence to the unstudentised grid scan of
Section 4.6 --- by the same cancellation applied per segment together
with a Donsker argument. Proposition P6-R in the same appendix then
removes the oracle-margin assumption at the boundary: the conclusion, pointwise and at the process level, holds for the statistic as computed at a fixed bandwidth, with ranks taken once over the whole window as in Section 4.6. The rank corrections change the covariance of the Gaussian limit
{[}10,11,34{]} but enter only in the product form (marginal fluctuation \(\otimes\) marginal state), which the limiting quadratic form
annihilates, so the limit law is that of P6. Section 4.6 uses this limit to calibrate the scan without permutations.

\textbf{(P3, detection and localisation consistency.)} Under a fixed
dependence change the criterion process converges uniformly to a
deterministic limit that is bounded away from zero at the true break
fraction, so power tends to one for any threshold sequence \(c_n\to 0\)
with \(\sqrt{n}c_n\to\infty\). Consistency of the break \emph{location} requires in addition an identifiability condition (A3), a unique maximiser of the weighted limit criterion, which we verify numerically for the S2 family (Supplementary Section A.7). Both statements
concern the unstudentised criterion; consistency for the studentised
estimator actually deployed, and localisation \emph{rates}, remain open.

\textbf{(P4, identifiability; Supplementary Section A.8.)} Population independence
implies \(I=0\) exactly, and the converse fails at finite \(D\):
Supplementary Section A.8 constructs a dependent copula density whose
\(D^2\times D^2\) moment equals the independent product's exactly, so
that \(I=0\) although the pair is not independent, for any unit-norm feature map with finitely many components; for a copula density fixed in advance, and frequencies
drawn from any law with a Lebesgue density, \(I>0\) almost surely
already at \(D=4\) per block, by an argument through the Paley--Wiener
theorem and the kernel characterisation of independence {[}22,36{]},
carried through for scalar blocks --- the case of every application in
Section 6, though not of the vector-valued blocks of Supplementary
Section B.6. Whether \(\lim_D I\) exists is settled only under a uniform
spectral-decay hypothesis we do not establish. The detection result and
the limit are proved for the paired feature map, whose norm is constant
by construction; the deployed map divides pointwise by a non-constant norm, so the detection proof does not apply to it directly, and its conclusion is checked numerically on the evaluated alternatives (Supplementary Section A.8). Detection requires only that the two regimes have distinct population \(I\), which Section 6.4 verifies for the two scenario families where the question is delicate. A change between regimes of equal \(I\), such as a Gaussian correlation reversing its sign, is not located by the scan (Supplementary Section B.21).

\section*{4. Density-operator mutual information for dependence
changes}
\addcontentsline{toc}{section}{4. Density-operator mutual information for
dependence changes}

\subsection*{4.1 Construction}
\addcontentsline{toc}{subsection}{4.1 Construction}

Let \((X_t,Y_t)\), \(t=1,\dots,n\), be two \emph{variable blocks} of a
multivariate series (\emph{block} also denotes a run of time points in Section 4.3); the experiments of the main text use scalar blocks, and Supplementary Section B.6 measures power and specificity when the blocks are vector-valued at a fixed feature dimension. Compute
componentwise ranks over the full observation window, scaled to
\((0,1)\) (pseudo-observations), and embed each block with random
Fourier features (RFF) {[}19{]} normalised to unit Euclidean norm,

\begin{equation}\varphi_{X,t}=\varphi(U_{X,t})/\lVert\varphi(U_{X,t})\rVert\in S^{D-1},\end{equation}
and likewise \(\varphi_{Y,t}\). Each
time point contributes a pure product state
\(\psi_t=\varphi_{X,t}\otimes\varphi_{Y,t}\), and a segment defines the
separable joint state and its marginals,

\begin{equation}\rho_{XY}=\mathrm{mean}_t\,\psi_t\psi_t^{\top}\in\mathbb{R}^{D^2\times D^2},\end{equation}
\(\rho_X=\mathrm{mean}_t\,\varphi_{X,t}\varphi_{X,t}^{\top}\) and \(\rho_Y\) likewise.

\textbf{Lemma 3 (exact partial traces; Supplementary Section A.4).} If both blocks'
features are unit-norm,
\(\lVert\varphi_{X,t}\rVert=\lVert\varphi_{Y,t}\rVert=1\) for all \(t\),
then \(\mathrm{Tr}_Y\,\rho_{XY}=\rho_X\) and
\(\mathrm{Tr}_X\,\rho_{XY}=\rho_Y\) (the first identity alone needs only
the second norm condition). \emph{Proof:}
\(\mathrm{Tr}_Y(uu^{\top}\otimes vv^{\top})=\lVert v\rVert^2\,uu^{\top}\);
average over \(t\). $\blacksquare$

The \textbf{segment DOMI} is

\begin{equation}\begin{split}I(X{:}Y)&=S(\rho_X)+S(\rho_Y)-S(\rho_{XY})\\&=S(\rho_{XY}\Vert\rho_X\otimes\rho_Y)\ \ge\ 0.\end{split}\end{equation}

At the population level independence gives
\(\rho_{XY}=\rho_X\otimes\rho_Y\), hence \(I=0\); \(I\) responds to any
dependence expressible in the feature space, not only correlation. (Sample quantities carry hats, \(\hat\rho\), \(\hat I\), where population counterparts appear.)

\textbf{Lemma 4 (elementary properties; Supplementary Section A.10).}
(i) \emph{Range:} \(0\le \hat I\le\min(S(\rho_X),S(\rho_Y))\le\log D\),
the upper bound following from separability of \(\rho_{XY}\) together
with strong subadditivity. (ii) \emph{Invariance:} \(\hat I\) is exactly
unchanged under strictly increasing transformations of either block and
under orthogonal changes of the feature basis; tied values are ranked in time order, which such transformations preserve. The Gram form also ignores the sign of an untied block. (iii) \emph{Consistency:}
\(\hat I\to I\) almost surely under assumptions (A0)--(A1) of Supplementary Section
A.6. (iv) \emph{Bias:} for fixed \(D\) and i.i.d. pairs under (A0) and (A2),
\(\mathbb{E}[\hat I]-I=c/m+o(1/m)\), with \(c\) given in closed form by
the entropy Hessian of Supplementary Section A.0 (Pre-7).

The invariance of (ii) is exact for a transformation applied to the \emph{whole window}; a marginal change part-way through the series is only partly absorbed by the pooled ranks, and Section 6.2 measures the residual. By (i), \(\hat I\) saturates as the dependence approaches a deterministic relation, whereas classical mutual information diverges; for the rank embedding the operative ceiling \(\min(S(\rho_X),S(\rho_Y))\) lies between about 0.8 and 1.3 nats for \(D\) from 4 to 64 and does not grow with \(D\) (Supplementary Section A.10), which compresses differences between strongly dependent regimes (Section 6.4). In (iv), \(c\) was positive at every independence configuration evaluated, and \(c\ge 0\) at the boundary is a theorem (Supplementary Section A.10(iv), via A.12); the plug-in is thus biased \emph{upward} near independence, and Section 4.2 gives a correction for short segments. Under near-deterministic dependence \(c\) changes sign (Supplementary Section A.10). At \(D=8\) under independence the closed form evaluates to \(c\approx3.07\) and the measured \(m\,\mathbb{E}[\hat I]\) is 3.12, a finite-sample comparison.

The kernel bandwidth of the random features is set by the median heuristic on the rank inputs, a rule fixed across data sets and scenarios, so that no reported power is the best of a bandwidth search. Supplementary Section B.7 measures what the rule costs: halving the bandwidth would raise power on every alternative and lower the false-alarm rate, a quarter is better still on the copula-shape alternative and collapses the correlation ones, and the deployed setting is best on none of the alternatives evaluated. Every comparison is made at a common bandwidth; Supplementary Section B.7 also sweeps each statistic over the same grid. The feature dimension is \(D=8\) by default (Supplementary Section B.1 reports the sensitivity to \(D\)).

\subsection*{4.2 Change-point statistics}
\addcontentsline{toc}{subsection}{4.2 Change-point statistics}

For a candidate split \(t\) we use the weighted DOMI difference

\begin{equation}Q(t)=\sqrt{t(n-t)/n}\,\bigl|\hat I_{[0,t)}-\hat I_{[t,n)}\bigr|\end{equation}
(global form; a sliding two-window form is analogous), and, for changes
in the joint state, the Holevo statistic \(n\chi(t)\) of Section 3.2
applied to \(\rho_{XY}\). The plug-in bias of \(\hat I\) can be corrected by subtracting, per segment, the mean DOMI of copies in which the Y rows are permuted --- a correction designed for the independence case that over-corrects under strong dependence (Supplementary Section A.10). The test of Section 4.3 uses the uncorrected \(Q(t)\), since the observed statistic and its replicas carry the same bias.

\subsection*{4.3 Fully data-driven calibration by pair
permutation}
\addcontentsline{toc}{subsection}{4.3 Fully data-driven calibration by
pair permutation}

We calibrate from the observed series alone: generate \(K\) replicas by
\emph{jointly permuting the time indices of the pairs} \((X_t,Y_t)\) --- under a no-change null stated as exchangeability of the pairs, the replicas preserve the dependence level while destroying any change
structure. Studentise \(Q(t)\) pointwise by the mean and standard
deviation taken over all \(K{+}1\) curves, the observed one included,
and reject when the permutation p-value

\begin{equation}p=\frac{1+\#\{k\ge 1:\ T_k\ge T_{\mathrm{obs}}\}}{K+1}\end{equation}
is at most \(\alpha\). Proposition P1 then applies verbatim. What is permuted fixes the null. Permuting only the Y rows generates an \emph{independence} null, which misstates the variability of the statistic when the series is dependent throughout, as a no-change null for a dependence change allows: at the deployed configuration (\(D=8\), \(K=99\), 200 replicates, level 0.05) the Y-only scheme rejects in 96.5\% of replicates under a constant Gaussian-copula null and in 26.5\% under a constant nonlinear-dependence null, against 7.0\% and 6.0\% for pair permutation; under an independent null both are valid (6.0\% and 5.0\%). We use the p-value form (9) throughout; the rule that compares \(T_{\mathrm{obs}}\) with the empirical \((1-\alpha)\) quantile of the replica maxima drops the \(+1\) and is anticonservative at finite \(K\) (Supplementary Section B.3).

\textbf{Serial dependence, and calibration by block permutation.} The
exactness of P1 rests on exchangeability of the pairs, which serial
dependence violates: a stationary series with autocorrelated levels or
clustered volatility is not exchangeable, and the pair permutation then
compares the observed curve against under-dispersed replica curves. Partition the series into consecutive \emph{time} blocks of length \(b\) (runs of observations, not the variable blocks \(X\) and \(Y\)) and permute the block \emph{order}, jointly in \(X\) and \(Y\). Every observation is used exactly once, so this is a genuine permutation, and Proposition P5 (Supplementary Section A.11) gives the finite-sample exactness of P1 over the block-permutation group, under exchangeability of the blocks. A serially dependent series satisfies that only approximately, with an error governed by the dependence carried across a block boundary; exchangeability also requires identically distributed blocks, which a seasonal cycle in variability breaks at any block length comparable with the cycle, however weak the coupling across boundaries. Both failures can be tested: under block exchangeability the permutation distribution of any block-level summary is exact, so a portmanteau statistic on the block means, or on the log block variances, recomputed over permutations of the block order returns an exact p-value even when a window holds few blocks. In Supplementary Table B.5, with first-order autoregressive (AR(1)) margins, block permutation brings the level to within one standard error of nominal at both \(b=20\) and \(b=50\), against 24.5\% for pair permutation; longer blocks carry less dependence across a boundary at the price of a coarser permutation group, while the candidate grid is unaffected (Supplementary Sections A.11 and B.17). We therefore use pair permutation for the synthetic studies, whose null is exchangeable by construction, and block permutation for serially dependent series such as those of Sections 6.7 and 6.8.

Block permutation improves the level but leaves little power at moderate signal: with AR(1) margins and a correlation-free change appearing at the midpoint, it rejects in 7.0\% of replicates at \(a=0.7\), against its own null rate of 6.5--8.5\% in that run (\(K=99\)), and in 19--22\% at \(a=0.9\), a likely consequence of the effective sample size for the mean, which for an AR(1) with \(\phi=0.6\) is 150 at \(n=600\) (Supplementary Section B.17).

\textbf{An exchangeability diagnostic, and the choice of calibration.}
Exchangeability itself can be tested by the same finite-sample logic: the group argument behind P1 places no restriction on the functional (Supplementary Section A.5, Remark on testing exchangeability), so a statistic aimed at serial structure, calibrated by joint pair permutation, is an exact-level test \emph{of} pair exchangeability. Non-rejection does not establish exchangeability, and no test can (same Remark).

We use five statistics computed from the centred rank sequences \(z_t=u_t-\tfrac{1}{2}\) of each margin: Ljung--Box portmanteaux at \(L\) lags on \(z_t\) and on \(z_t^2\) of each margin (the level and scale mechanisms that break the pair calibration), and a two-sided portmanteau on the lagged cross-correlations, each calibrated by pair permutation at \(K=199\) and combined by Bonferroni at level 0.05. If no statistic rejects, we use pair permutation; otherwise block permutation with \(b=\lceil 5\hat\tau_{\mathrm{int}}\rceil\), where \(\hat\tau_{\mathrm{int}}\) is the largest integrated autocorrelation time of the four sequences \(z_t\) and \(z_t^2\), so that the dependence carried across one block boundary is at most of order \(e^{-5}\) under autoregressive decay. For a fixed block length the block test is exact under pair exchangeability as well, the block-permutation group being a subgroup of the full permutation group (Supplementary Section A.11, Remark). The protocol, however, chooses the scheme and the block length from the same data, which Propositions P1 and P5 do not cover, so its level is not proved. On the three null designs of Supplementary Table B.5 it rejects in 4.3--5.5\% of replicates, but in 9.0\% (standard error 2.0\%) with AR(1) margins and independent innovations, against 6.5\% for either fixed block length in the same \(K=199\) run (Supplementary Section B.17). The pair branch rests on a pre-test, and a violation the diagnostic cannot detect distorts the level by an amount it cannot bound.

The protocol selects block permutation in 100\% of AR(1) and 96.5\% of GARCH replicates and in 1\% of serially independent ones. The constant 5 is a floor that keeps the permutation group rich; it is not a guarantee of level, although for an AR(1) with \(\phi=0.6\) it gives \(b=20\), where the level is 5.0\% in Supplementary Table B.5 and 6.0\% in a separate run with \(K=199\) (Supplementary Section B.17). The diagnostic is consistent against every strictly stationary ergodic alternative whose serial structure is visible to one of its statistics at some lag \(k\le L\), by the finite-population permutation moments of Wald and Wolfowitz {[}33{]} (Proposition P7, Supplementary Section A.13).

\subsection*{4.4 Multiple breaks: Holevo partitioning with bias-corrected
costs}
\addcontentsline{toc}{subsection}{4.4 Multiple breaks: Holevo partitioning with
bias-corrected costs}

For multiple change points, apply optimal partitioning with cost
\(C(s)=E_s\,S(\hat\rho_{XY,s})\) on a candidate grid, \emph{after}
subtracting from each segment cost its mean under a small number of
permutations of the whole series (we use 3--5), drawn from the pair or block scheme that the protocol of Section 4.3 selects. For a fixed permutation scheme under its exchangeability null, this centres every split gain at zero in expectation, making a single penalty
\(\beta\) meaningful across segment lengths; \(\beta\) itself is chosen
as the smallest value at which at most a fraction \(\alpha\) of permuted
copies of the observed series produce any change point, with copies drawn from the same permutation scheme as the cost correction. Without the correction we found no stable \(\beta\): the length bias of the raw entropy drives the segmentation from severe over- to global under-segmentation. The partitioning is exact and unpruned: the centred gains can be negative, and we do not establish for them the cost inequality that PELT pruning {[}1{]} requires; its cost is quadratic in the number of grid points (Section 4.6).

\subsection*{4.5 Combining directional and omnibus
statistics}
\addcontentsline{toc}{subsection}{4.5 Combining directional and omnibus
statistics}

DOMI spreads its power over every dependence direction expressible in the \(D^2\)-dimensional feature moment, so on a purely monotone-correlation alternative the random-feature form is generally less powerful than Spearman (scenario S1). Because P1 holds for \emph{any} statistic that treats the observed series and its pair-permuted replicas symmetrically, the combined statistic \(T=\max(T_{\mathrm{DOMI}},T_{\mathrm{Spearman}})\) of the two studentised maxima, computed from the \emph{same} replicas, is finite-sample valid exactly as the individual tests are. At strong signal it attains power close to the better of the two in both regimes, and it loses 14 percentage points (0.23 against 0.37) when the nonlinear signal is weak (Supplementary Section B.18).

\subsection*{4.6 The procedure in summary, and its
cost}
\addcontentsline{toc}{subsection}{4.6 The procedure in summary, and its
cost}

Algorithms 1 and 2 state the single-break test and the two-stage procedure in execution order.

\begin{figure}[!t]
\footnotesize
\hrule\vspace{2pt}
\noindent\textbf{Algorithm 1} Single-break test for a dependence change\\[2pt]
\hrule\vspace{3pt}
\noindent\textbf{Input:} variable blocks $X,Y$ observed at $t=1,\dots,n$; feature dimension $D$; replicas $K$; level $\alpha$; time-block length $b$ ($b=1$ for pair permutation)\\
\textbf{Output:} permutation p-value $p$; break estimate $\hat\tau$\\
1: rank each component over the whole window, scale to $(0,1)$; draw $D$ Fourier frequencies per variable block, bandwidth by the median heuristic (ties and the bandwidth subsample: Supplementary Section A.5)\\
2: form unit-norm features $\varphi_{X,t},\varphi_{Y,t}$, product states $\psi_t=\varphi_{X,t}\otimes\varphi_{Y,t}$, and prefix sums of $\varphi_X\varphi_X^{\top}$, $\varphi_Y\varphi_Y^{\top}$, $\psi_t\psi_t^{\top}$\\
3: evaluate $Q_0(t)$ on the candidate grid by (8), every segment quantity a difference of prefix sums\\
4: \textbf{for} $k=1$ \textbf{to} $K$ \textbf{do} permute the pairs jointly in time blocks of length $b$; recompute $Q_k(t)$ \textbf{end for}\\
5: studentise: $T_k=\max_t\,(Q_k(t)-\bar Q(t))/\mathrm{sd}(Q(t))$, moments taken over all $K{+}1$ curves\\
6: $p\leftarrow(1+\#\{k\ge 1: T_k\ge T_{\mathrm{obs}}=T_0\})/(K+1)$ by (9); $\hat\tau\leftarrow\arg\max_t$ of the studentised observed curve\\
7: \textbf{return} $p$, $\hat\tau$\\[2pt]
\hrule\vspace{8pt}
\hrule\vspace{2pt}
\noindent\textbf{Algorithm 2} Segment and re-test\\[2pt]
\hrule\vspace{3pt}
\noindent\textbf{Input:} variable blocks $X,Y$; grid step; replicas $K$; levels $\alpha$ (stage one) and $q$ (stage two)\\
\textbf{Output:} classified change points\\
1: run the exchangeability diagnostic of Section 4.3; \textbf{if} it rejects \textbf{then} $b\leftarrow\lceil 5\hat\tau_{\mathrm{int}}\rceil$ \textbf{else} $b\leftarrow 1$; for seasonal records, permute within seasons (Section 6.7)\\
2: build the cost matrix $C(s)=E_sS(\hat\rho_{XY,s})$ on a coarse grid\\
3: subtract from each entry its mean over three to five permuted copies of the series, using the centring of Section 4.4\\
4: $\beta\leftarrow$ the smallest penalty at which at most a fraction $\alpha$ of permuted copies yield any break; run optimal partitioning\\
5: \textbf{for each} candidate break \textbf{do} re-test it with Algorithm 1 on a window centred on it, and run the marginal-scale test of Section 4.7 alongside \textbf{end for}\\
6: apply a multiplicity adjustment to the dependence p-values across candidates at level $q$; classify each candidate as a dependence change with no detected scale change, a change in both, a scale change only, or undetermined (the p-values are nominal for the selected windows)\\
7: \textbf{return} the classified candidates\\[2pt]
\hrule
\end{figure}

A segment DOMI needs
two \(D\times D\) eigendecompositions and one \(D^2\times D^2\)
eigendecomposition, so evaluating the whole difference curve costs
\(O(n_{\mathrm{grid}}D^6)\) after an \(O(nD^4)\) pass for the prefix
sums; the permutation calibration multiplies this by \(K{+}1\), and the
segmentation costs \(O(n_{\mathrm{grid}}^2D^6)\) for the cost matrix and
\(O(n_{\mathrm{grid}}^2)\) for the partitioning; at \(n=9600\) a test with \(K=99\) takes 341 s over the dense grid and 5.9 s over 161 splits (Supplementary Section B.9).

\textbf{A permutation-free calibration at the independence boundary.}
When the null is independence throughout and the pairs are i.i.d., as in a test for dependence emerging from an independent regime, Proposition P6-R gives an asymptotic calibration of the scan that needs no permutations: the scan converges to a weighted functional of independent Brownian motions, with weights computed from the whole-series feature moments (Supplementary Section A.12; proved with \(D\), the frequency draw and the bandwidth fixed, checked empirically at the deployed bandwidth). Simulating that functional with 2000 draws gives critical values inside
the exact Monte Carlo confidence interval from \(n=600\) to 38,400 and
levels of 0.051--0.055 for the unstudentised scan, and of 0.051--0.059
when the scan is studentised by the moments of the limit process. The
studentised version is the more powerful of the two limit-calibrated scans; on an S2-type emergence at \(a=0.7\) and \(n=600\) it rejects in 64\% of replicates, at an empirical level of 0.058, against 48\% for the permutation-calibrated scan. The saving over \(K=99\)
permutations is eight- to elevenfold (Supplementary Section B.16). The
calibration covers the independence boundary only; tests of a change between two dependent regimes, and the scan studentised by permutation moments, remain calibrated by permutation. This calibration targets independence throughout the window, so a rejection alone does not establish a change in the dependence.

\subsection*{4.7 Segmenting and re-testing: a two-stage
procedure}
\addcontentsline{toc}{subsection}{4.7 Segmenting and re-testing: a
two-stage procedure}

The cost \(C(s)=E_s S(\hat\rho_{XY,s})\) of Section 4.4 is a functional of the \emph{joint} state, so it responds to a change in either component, marginal changes included (Section 6.2). The DOMI-difference test holds its level under M1, but a segmentation built on it recovers the three breaks of scenario MB far less often (Section 6.3).

\textbf{Stage one} therefore segments with Holevo partitioning, for the power to find a multi-break structure; \textbf{stage two} re-tests each candidate break with the DOMI difference under the permutation calibration of Section 4.3 (block form for serially dependent series, permuted within seasons for seasonal records), for the specificity to flag a change in the coupling.

A candidate that fails stage two is \emph{undetermined}, not refuted:
the coupling may have changed too little to be detected. Likewise, ``no detected marginal change'' is a non-rejection of the marginal test, not a demonstration that the margins held. To separate the two possibilities we run alongside a marginal test of the same form and block structure, a block-permutation test of \(|\log\hat\sigma_L-\log\hat\sigma_R|\), weighted as in Section 4.2, on each variable separately. It holds its level under serial dependence (5.0\% of replicates on an AR(1) margin with \(\phi=0.6\) and no scale change) and rejects in every replicate against an eightfold scale change; it tests the scale of each margin, not the equality of the marginal distributions. The two tests classify each candidate: a significant
dependence change with no detected scale change; a change in both; a
scale change without a significant dependence change; and, when neither
test rejects, undetermined.

Stage two tests a window that stage one selects from the same data.
Propositions P1 and P5 treat the window as fixed, and the selection is
not invariant under the permutations that stage two applies, so the
guarantee does not carry over: the stage-two p-values are nominal and
carry no correction for the selection. Two designs remove the selection effect. Pre-specifying the windows does so at some cost in power, and suits applications where the candidate locations are not themselves of interest. Splitting the series into interleaved blocks, selecting on one set and testing on the other, does so approximately when blocks separated by a gap are
close to independent and each half still holds enough blocks for the
permutation distribution to reach the level. An exploratory resampling check of the whole two-stage procedure, which reuses the observed stage-one calibration, is not established as valid, and its null is the absence of any change, marginal changes included (Supplementary Section B.14). The calibration of stage two can be checked on the record by its rejection rate at dates drawn at random (Section 6.7).

\subsection*{4.8 Two computational forms}
\addcontentsline{toc}{subsection}{4.8 Two computational forms}

\begin{table}[!t]
\centering
\caption{Single-break comparison: localised power at a false-alarm rate of 0.05 (500 replicates; break at the centre, \(\tau=300\); calibration and tolerance of Section 5). Upper block: S1--S4 and the detection rate under the marginal-only control M1 (FA; nominal 0.05). Lower block: G1--G6. DOMI is the random-feature form at \(D=8\), Gram its Gram form, HSIC the Hilbert--Schmidt independence criterion, dCor the distance correlation, Spear.\ the Spearman correlation, CvM the copula test of [10]; one variant per statistic and block (Section 5), raw for CvM and Spearman. Bold: largest value in a row at two decimals (ties bolded); \(\tau\) in S3 and G2--G3 is Kendall's \(\tau\). Further variants and joint-distribution references: Supplementary Tables B.4 and B.6.}
\footnotesize
\setlength{\tabcolsep}{3pt}
\begin{tabular}{l cccccc}
\toprule
Setting & DOMI & Gram & HSIC & dCor & Spear. & CvM \\
\midrule
S1, $r{=}0.35$ & 0.60 & \textbf{0.82} & 0.57 & 0.60 & 0.65 & 0.72 \\
S2, $a{=}0.7$ & 0.35 & \textbf{0.58} & 0.06 & 0.01 & 0.00 & 0.03 \\
S2, $a{=}0.9$ & 0.87 & \textbf{0.98} & 0.66 & 0.07 & 0.00 & 0.05 \\
S3, $\tau{=}0.5$ & 0.01 & 0.01 & 0.02 & 0.01 & 0.00 & \textbf{0.34} \\
S4, $r{=}0.7$ & 0.09 & \textbf{0.12} & 0.01 & 0.01 & 0.00 & 0.02 \\
S4, $r{=}0.85$ & \textbf{0.55} & \textbf{0.55} & 0.39 & 0.24 & 0.00 & 0.01 \\
M1, $s{=}2$ (FA) & 0.05 & 0.04 & 0.07 & 0.04 & 0.03 & 0.05 \\
\midrule
G1, $\nu{=}8$ & \textbf{0.01} & \textbf{0.01} & 0.00 & 0.00 & 0.00 & \textbf{0.01} \\
G1, $\nu{=}4$ & 0.03 & \textbf{0.05} & 0.00 & 0.00 & 0.00 & 0.01 \\
G1, $\nu{=}2$ & 0.33 & \textbf{0.44} & 0.03 & 0.01 & 0.00 & 0.01 \\
G2, $\tau{=}0.2$ & 0.72 & \textbf{0.79} & 0.45 & 0.48 & 0.51 & 0.60 \\
G2, $\tau{=}0.35$ & 0.97 & \textbf{1.00} & 0.65 & 0.68 & 0.88 & 0.91 \\
G2, $\tau{=}0.5$ & \textbf{1.00} & \textbf{1.00} & 0.76 & 0.84 & 0.99 & 0.99 \\
G3, $\tau{=}0.2$ & 0.72 & \textbf{0.86} & 0.65 & 0.56 & 0.55 & 0.63 \\
G3, $\tau{=}0.35$ & 0.97 & \textbf{1.00} & 0.88 & 0.78 & 0.92 & 0.94 \\
G3, $\tau{=}0.5$ & 0.99 & \textbf{1.00} & 0.96 & 0.87 & 0.98 & 0.99 \\
G4, $a{=}0.5$ & 0.97 & \textbf{0.98} & 0.77 & 0.52 & 0.01 & 0.36 \\
G4, $a{=}0.7$ & \textbf{1.00} & \textbf{1.00} & 0.91 & 0.86 & 0.01 & 0.84 \\
G4, $a{=}0.9$ & \textbf{1.00} & \textbf{1.00} & \textbf{1.00} & 0.99 & 0.01 & 0.99 \\
G5, $a{=}0.5$ & 0.94 & \textbf{0.98} & 0.77 & 0.49 & 0.00 & 0.22 \\
G5, $a{=}0.7$ & \textbf{1.00} & \textbf{1.00} & 0.97 & 0.92 & 0.01 & 0.80 \\
G5, $a{=}0.9$ & \textbf{1.00} & \textbf{1.00} & 0.99 & 0.99 & 0.01 & 0.97 \\
G6, $b{=}0.3$ & 0.49 & \textbf{0.65} & 0.16 & 0.18 & 0.00 & 0.10 \\
G6, $b{=}0.5$ & 0.86 & \textbf{0.91} & 0.60 & 0.73 & 0.01 & 0.59 \\
G6, $b{=}0.8$ & 0.94 & \textbf{0.97} & 0.74 & 0.95 & 0.01 & 0.95 \\
\bottomrule
\end{tabular}
\end{table}

\begin{table}[!t]
\centering
\caption{Cost of the two computational forms: wall-clock seconds for one full difference curve over the candidate grid, single core, scenario S2 at $a=0.7$. Asterisk: median per-candidate time times the grid size; --: not run (Supplementary Section B.9).}
\footnotesize
\setlength{\tabcolsep}{2pt}
\begin{tabular}{l cccccc}
\toprule
Form & $n{=}600$ & 1200 & 2400 & 4800 & 9600 & 19200 \\
\midrule
random features, $D{=}8$ & 0.1 & 0.3 & 0.7 & 1.5 & 3.0 & 5.9 \\
random features, $D{=}16$ & 2.0 & 4.4 & 9.1 & 18.2 & 37.5 & 76.0 \\
Gram form & 10.8 & 233$^{*}$ & 3,809$^{*}$ & 78,510$^{*}$ & -- & -- \\
\bottomrule
\end{tabular}
\end{table}

The DOMI of a segment can be computed on either side of a spectral
duality. Stack the unit-norm features of a segment of length \(m\) as
the rows of \(\Phi_X\) and \(\Phi_Y\), and let
\(K_X=\Phi_X\Phi_X^{\top}\) and \(K_Y=\Phi_Y\Phi_Y^{\top}\) be the
\(m\times m\) Gram matrices, which have unit diagonal. Then
\(\rho_X=\Phi_X^{\top}\Phi_X/m\) and \(K_X/m\) have the same non-zero
eigenvalues, and so do \(\rho_{XY}\) and \((K_X\circ K_Y)/m\), where
\(\circ\) is the entrywise product, because \(K_X\circ K_Y\) is the Gram
matrix of the product states \(\psi_t\). The entropies of Section 4.1 can therefore be computed exactly from these Gram matrices. Replacing the random-feature kernel in these Gram matrices by
the Gaussian kernel it approximates, at the same bandwidth, gives a
different statistic, the \textbf{Gram form} of DOMI, which is the
matrix-based mutual information of {[}15{]} at \(\alpha=1\) on rank
inputs. On any fixed segment the random-feature form converges to the
Gram form as \(D\to\infty\), almost surely over the frequency draws: the
normalised random-feature kernel converges to the Gaussian kernel at
each of the finitely many pairs of the segment, and the entropy is continuous (Supplementary Section A.8).
Unqualified, DOMI refers to the random-feature form.

Unlike the random-feature form, whose cost per split does not depend on the segment length (Section 4.6), the Gram form requires eigendecompositions of \(m\times m\) matrices per segment, at cost \(O(m^3)\), and Section 6 evaluates it on every row of every segment; subsampling rows to cap that cost is not a neutral approximation, since a fresh subsample for each candidate split perturbs the difference curve and degrades localisation (Supplementary Section B.10). At \(n=4800\) one difference curve takes 1.5 s in the random-feature form at \(D=8\) and an estimated 22 hours in the Gram form (Table 2). The order-2 Gram variant is the exception: two-dimensional prefix sums give every segment in constant time, but they occupy three \(n\times n\) arrays, about 110 GiB in double precision at the \(n\approx 70{,}000\) of Section 6.7, whereas the random-feature form stores \(O(nD^4)\) (Supplementary Section B.9).

Propositions P1 and P5 require only that the observed series and its
replicas be treated symmetrically, so they hold for the Gram form
without change, as does the rank invariance of Lemma 4(ii). The limit
laws P2 and P6 are proved at fixed \(D\) and therefore describe the
random-feature form, which Holevo partitioning also uses because prefix sums make its \(O(n_{\mathrm{grid}}^2)\) segments practical. The Gram form is available for the single-break test and for stage two of Section 4.7, where the permutation calibration needs no limit law.

\section*{5. Experimental design}
\addcontentsline{toc}{section}{5. Experimental design}

Unless noted otherwise, each series has \(n=600\) and a single break at \(\tau=300\). In S1, S3 and S4 the margins are standard normal in every regime; in S2 the change also alters the shape of the \(Y\) margin, keeping its mean and variance. The scenarios (code labels D1--D4 for S1--S4) are: \textbf{S1} Gaussian
correlation \(0\to r\); \textbf{S2} independence $\to$ uncorrelated
dependence \(Y=\sqrt{1-a^2}\,\varepsilon + a|X|\varepsilon'\);
\textbf{S3} copula-shape change Gaussian $\to$ Clayton at equal Kendall's
\(\tau\); \textbf{S4} sign-mixed dependence at mixture level \(r\)
(overall correlation zero) $\to$ independence; \textbf{M1} specificity
control: marginal variance change only, \(X\perp\!\!\!\perp Y\)
throughout; \textbf{MB} (multi-break) \(n=1800\) with independent /
\(|X|\)-dependent (as in S2) / independent / correlated segments at
\(\tau=450,900,1350\). Dependence-specific baselines on identical rank
inputs: the difference of the Hilbert--Schmidt independence criterion
(HSIC), computed from the same RFF with the cross-covariance Frobenius
norm as the mean-embedding counterpart of DOMI; the distance correlation
difference {[}24{]}; the Spearman difference; and the Cramér--von Mises
statistic of the sequential empirical-copula test {[}10{]}, with ranks
recomputed within each segment as that test prescribes. The HSIC entry ablates the functional and is not the standard full-Gram method, which Supplementary Section B.7 adds with a bandwidth sweep of both statistics.
Joint-distribution references, not specific to dependence: on raw data, MMD {[}23{]} at three bandwidths and a squared-kernel variant and a bivariate Gaussian likelihood ratio; on the rank features, the joint-state Holevo statistic of Section 4.2. Neural detectors such as KL-CPD {[}9{]} target generic distributional change without marginal invariance and are outside the main comparison (MC-TIRE {[}38{]}: Supplementary Section B.20). Power protocol: thresholds set to a
false-alarm rate of 0.05 separately for each method, and separately for
each level, from 500 null replicates of the no-change process that
matches that level --- the alternative's pre-change regime held
throughout, which for S3 and S4 carries dependence and therefore differs
from level to level --- of which the first half supply the pointwise studentising moments and the second half the threshold quantile. For each baseline we report, per table block, whichever of its two variants, the raw weighted global difference or the per-\(t\) studentised one, has the larger mean power over the block. The variant is chosen on the evaluation replicates; for DOMI and its Gram form it is the studentised one that Algorithm 1 fixes in advance, so the choice can only favour the baselines. A rejection counts towards localised power only if \(|\hat\tau-\tau|\le 30\); the fully data-driven
procedure is evaluated separately (Sections 6.5, 6.7, and 6.8).

The Gram form of Section 4.8 is run on the same replicates under the same calibration, with the
Gaussian kernel at the bandwidth the random features approximate, at
\(\alpha=1\) and, as the order-2 matrix-based Rényi variant of
{[}14,15{]}, at \(\alpha=2\). In the non-Gaussian suite G1--G6, \(X\)
and \(Y\) are independent standard normal before \(\tau\) and, after
it, follow a \(t\) copula with zero correlation parameter (G1, degrees of
freedom \(\nu=8,4,2\)); a Gumbel (G2) or Clayton (G3) copula at Kendall's
\(\tau=0.2,0.35,0.5\); a U-shaped dependence
\(Y^*=\sqrt{1-a^2}\,\varepsilon+a(X^2-1)/\sqrt2\) (G4); a periodic
dependence \(Y^*=\sqrt{1-a^2}\,\varepsilon+a\,c\,(\cos\pi X-\mathbb{E}\cos\pi X)\)
with \(c\) setting unit variance (G5; \(a=0.5,0.7,0.9\) in G4 and G5);
or a volatility coupling \(Y^*=\varepsilon\exp(bX-b^2)\) (G6,
\(b=0.3,0.5,0.8\)), with \(\varepsilon\) standard normal and independent
of \(X\). G1, G4, G5 and G6 have zero correlation. In G4--G6, \(Y^*\) is
mapped to a standard normal margin through its distribution function,
estimated once from a fixed Monte Carlo sample of 400,000 draws, so in
all six scenarios both margins are standard normal throughout and only
the dependence changes. The pre-change regime is independence, so the null does not vary with the level.

\section*{6. Results}
\addcontentsline{toc}{section}{6. Results}

\subsection*{6.1 Single break}
\addcontentsline{toc}{subsection}{6.1 Single break}

\begin{figure*}[!t]
\centering
\includegraphics[width=0.80\textwidth,height=0.42\textheight,keepaspectratio]{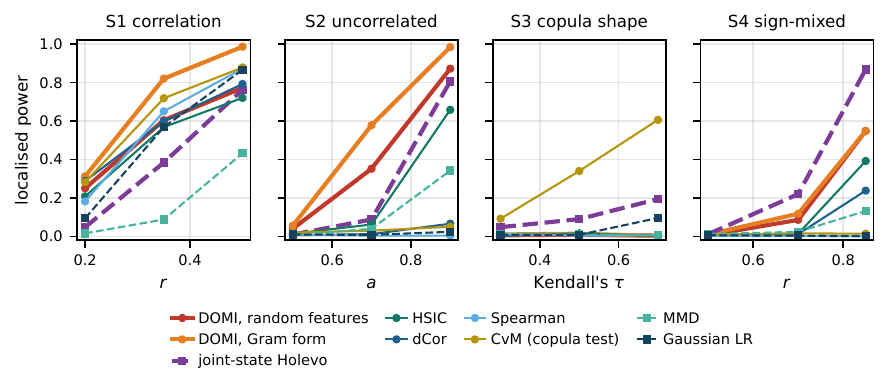}
\caption{Localised power on S1--S4 against the change magnitude (500 replicates; calibration of Section 5), DOMI in both forms (random features at \(D=8\)). The dashed joint-state Holevo, MMD and Gaussian likelihood-ratio curves are joint-distribution references, not dependence-specific: under M1 at \(s=2\) each rejects in every replicate (Supplementary Table B.6). BMCTC is not shown (Supplementary Section B.13).}
\end{figure*}

On the appearance of uncorrelated dependence (S2), DOMI exceeds the best classical dependence statistic by 29 and 21 percentage points at \(a=0.7\) and 0.9 (Table 1). On the disappearance of dependence (S4) at \(r=0.85\) the gain is 15.6 percentage points on the unrounded rates. At \(r=0.7\) localised power is low for every statistic of Table 1 and no separation is claimed; the matrix-information baseline of {[}37{]}, calibrated alike, is more powerful on S4 at both levels (0.23--0.57 here and 0.80--0.94 at \(r=0.85\); Supplementary Section B.13). Single-break scans of the MC-TIRE curves, calibrated alike, attain a localised power of at most 0.022 on the eleven dependence settings of Supplementary Table B.15. On linear correlation (S1) DOMI trails the best classical dependence statistic by 4 to 11 points at each level, distance correlation at \(r=0.2\) and the copula test at \(r=0.35\) and 0.5 (Figure 2). Localisation follows the ordering of power: on S2, and on S4 at its stronger level, DOMI localises more often than the classical statistics, which can detect the change without locating it (HSIC detects 0.90 of the S2 changes at \(a=0.9\) and 0.70 of the S4 changes at \(r=0.85\), Supplementary Table B.10; the median error of DOMI is 24 samples at \(r=0.85\), against 85 at \(r=0.7\), which exceeds the tolerance of 30), while on S1 Spearman localises slightly more accurately (median error 17 versus 22.5 samples at \(r=0.35\)).
The Gram form of DOMI (Section 4.8) attains 0.31 and 0.99 on S1 at \(r=0.2\) and 0.5, against 0.28 and 0.88 for the copula test and 0.32 and 0.91 for the correlation CUSUM of Supplementary Section B.4; its detection rate under M1 is 0.04.

On the non-Gaussian suite (lower block of Table 1), the Gram form attains the largest estimated localised power, alone or tied, in all 16 settings in which some statistic exceeds 0.10, and the random-feature
form attains the largest printed value in five of them. The
random-feature form trails the Gram form by at most 16.6 percentage
points (G6 at \(b=0.3\)) and is at least as powerful as every classical
statistic in 15 of the 16 settings; in the sixteenth, G6 at \(b=0.8\),
the copula test attains 0.952 and distance correlation 0.946 against
0.944. On the correlation-free scenarios G1, G4, G5 and G6 Spearman attains at most 0.012 at every level, and the copula test attains at
most 0.36 at the weakest level of each. At the weakest level of each
scenario at which some statistic exceeds 0.10, the Gram form exceeds the
best classical statistic by 19 to 47 percentage points.

Without the localisation requirement the forms of DOMI remain ahead, or tied at 1.00, where the dependence carries no correlation (Supplementary Table B.10), whereas on S1 and on the Gumbel and Clayton copulas HSIC and distance correlation also detect in at least 93\% of replicates, so the differences Table 1 shows there are differences in localisation. With the break at a quarter or three
quarters of the series, both forms of DOMI still detect in at least
88\% of replicates, but when the longer segment is the dependent one the
studentised maximiser moves towards the centre and fewer than 1\% of estimates fall within the tolerance; HSIC and distance correlation, scanned alike, share this pattern (Supplementary Table B.12). The failure concerns location, not detection, and in the applications location is the task of stage one (Section 4.7). Table 1 is calibrated by Monte Carlo draws from the known null. Under the
pair-permutation calibration the procedure itself uses
(\(K=99\), 200 replicates),
every statistic rejects in 3.0--8.0\% of replicates on three no-change designs, and both forms of DOMI retain higher localised power than the baselines on the three correlation-free alternatives, the Gram form highest. The random-feature form attains rejection and localised rates of 0.435 and 0.245 on S2 at \(a=0.7\) against at most 0.115 and 0.035 for the baselines,
1.000 and 0.690 on S4 at \(r=0.85\) against 0.980 and 0.565 for HSIC,
and 0.995 and 0.675 on G6 at \(b=0.5\) against at most 0.915 (distance correlation) and 0.580 (the copula test). Spearman and the copula test lead on S1, and the copula test on S3 (Supplementary Table B.11).

\subsection*{6.2 Specificity under marginal scale change (M1)}
\addcontentsline{toc}{subsection}{6.2 Specificity under marginal scale change (M1)}

A statistic is \emph{dependence-specific} here if a change confined to
the margins leaves its detection rate near the nominal level; under M1 the DOMI-difference statistic of Section 4.2 meets it, whereas the joint-state Holevo statistic of the same section, a functional of the joint operator \(\rho_{XY}\), does not. With only a marginal variance change, the
DOMI-difference statistic attains detection rates of 0.032--0.058 across
all change magnitudes (nominal 0.05); the joint-state Holevo statistic
reaches 0.51 at \(s=1.3\) and 1.00 at \(s\ge 1.6\). The segmentation cost of Section 4.4 is the same joint-state functional. Running Holevo partitioning on this
design, with the penalty calibrated to a 4\% spurious-detection rate on
the no-change null, returns a break in 45--100\% of replicates. Specificity thus belongs to the DOMI-difference statistic, not to the segmentation cost (Section 4.7). A control that changes the \(Y\) margin exactly as S2 does while keeping \(X\) and \(Y\) independent gives DOMI detection rates of 0.056--0.074 (Supplementary Section B.15), so the high power on S2 is not explained by the marginal-shape change alone; the joint-state Holevo statistic reaches 0.29 there at \(a=0.9\).

The copula test of {[}10{]}, which ranks within each segment, attains
detection rates of 0.04--0.05 across the change magnitudes; computed
from global pseudo-observations, the same statistic reaches 0.07--0.09
at \(s=1.3\) and 0.81--0.87 at \(s=2\), because pooled ranks make the segment copulas marginal-sensitive. Spearman and HSIC also stay near nominal (Table 1), so specificity is shared by all the rank-based dependence statistics of Table 1; what separates DOMI is its power on S2 and S4, where the copula test attains at most 0.05.

M1 keeps the two blocks independent and so cannot show the partial absorption of Section 4.1 under dependence.
Superimposing marginal shifts on an unchanged dependence of \(r=0.7\)
raises the false-alarm rate of the DOMI contrast to 0.085 (an eightfold
scale change) and 0.068 (a \(\tanh\) squash) against a nominal 0.05,
compared with 0.068 and 0.052 in the independent case. The specificity
is therefore approximate, and it is a measured property: neither P1
nor the rank invariance of Lemma 4(ii) implies it, since a marginal change part-way through the window violates exchangeability. On weather margins every statistic compared rejects under a winter-to-summer shift, and a December-to-February shift raises the false-alarm rate of DOMI to 0.110 but not that of the copula test of {[}10{]} (0.035; Supplementary Section B.19).

\subsection*{6.3 Multiple breaks}
\addcontentsline{toc}{subsection}{6.3 Multiple breaks}

\begin{table}[!t]
\centering
\caption{Three-break benchmark MB (200 replicates). Null: achieved null detection rate after calibration to a nominal 5\%. At $a=0.9$: probability of exactly three breaks, of all three within 90 time points of the true breaks, and mean adjusted Rand index (ARI); last column: exactly three at $a=0.8$. --: not computed. More in Supplementary Table B.7.}
\footnotesize
\setlength{\tabcolsep}{3pt}
\begin{tabular}{l ccccc}
\toprule
Method & Null & $P(\hat K{=}3)$ & All three & ARI & $a{=}0.8$ \\
\midrule
Holevo partitioning & 0.000 & 0.475 & 0.450 & 0.62 & 0.010 \\
Rank-Gaussian PELT & 0.020 & 0.010 & 0.000 & 0.32 & 0.010 \\
Kernel binary seg. & 0.050 & 0.115 & 0.050 & 0.48 & 0.015 \\
KCP, raw values & 0.050 & 0.175 & 0.115 & 0.44 & 0.025 \\
KCP, ranks & 0.050 & 0.030 & 0.015 & 0.33 & 0.010 \\
e.divisive, raw values & 0.055 & 0.015 & 0.000 & 0.07 & 0.005 \\
e.divisive, ranks & 0.060 & 0.000 & 0.000 & 0.12 & 0.000 \\
BMCTC, $\alpha{=}1.01$ & 0.050 & 0.190 & 0.115 & 0.27 & 0.055 \\
DOMI binary seg. & 0.040 & 0.040 & -- & -- & -- \\
\bottomrule
\end{tabular}
\label{tab:mb}
\end{table}

Table~\ref{tab:mb} compares the segmentation methods on the three-break benchmark MB (an example segmentation is in Supplementary Figure B.4). With penalties or significance levels calibrated to a nominal 5\% null detection rate for every method, Holevo partitioning at \(a=0.9\) returns exactly three breaks in 47.5\% of replicates and places all three within 90 time points of the true breaks in 45\%; no baseline exceeds 19\% and 11.5\%, respectively. At \(a=0.8\) no method of Table~\ref{tab:mb} exceeds 5.5\%. For Holevo partitioning this probability rises with segment length: at the standard level it reaches 99.5\% at \(n=3600\) with the break fractions held fixed (Supplementary Section B.3). A dependence-specific alternative to the segmentation cost exists: recursive binary segmentation driven by the DOMI difference, with each split tested by permutation at the 5\% level. It is as specific as the statistic it is built on --- false-alarm rates of 0.04--0.05 on the marginal-change design of Section 6.2, where Holevo partitioning reaches 1.00 --- but recovers the exact number of breaks in 4\% of replicates against 47.5\%, finding on average one of the three. Neither construction provides both specificity and multi-break power (Section 4.7).

\subsection*{6.4 Copula shape at fixed Kendall's \(\tau\) (S3)}
\addcontentsline{toc}{subsection}{6.4 Copula shape at fixed Kendall's
\(\tau\) (S3)}

Of the dependence-specific statistics only the copula test detects
this change, at 0.09, 0.34 and 0.61 for \(\tau=0.3\), 0.5 and 0.7; DOMI
attains 0.008, HSIC 0.016 and Spearman 0.004 at \(\tau=0.5\). For DOMI the cause is an unfavourable signal-to-noise ratio, not a vanishing signal. At \(\tau=0.5\) the population DOMI moves
from 0.1936 to 0.2064, a contrast of 0.0128; in scenario S2 at
\(a=0.7\), where power is 0.35, the contrast is 0.0148. The contrasts are of similar size, but holding \(\tau\) fixed forces \emph{both} regimes to carry substantial dependence, and the
sampling standard deviation of the difference of two segment plug-ins at
\(m=300\) is 0.0347 here against 0.0078 in S2, a factor of 4.4. The
standardised contrast is therefore 0.37 against 1.90, an ordering
consistent with the observed powers. The kernel statistics share the difficulty at the deployed bandwidth, the largest HSIC detection rate over the \(\tau\) sweep being 0.148 at \(\tau=0.5\), unaccompanied by localisation; the copula test, which compares the segment copulas directly, does not.
Supplementary Section B.7 shows that a quarter of the median bandwidth
lifts DOMI to 0.19 here, against 0.008 at the deployed setting and 0.03 for the best tuned kernel competitor, consistent with the variance argument.

\subsection*{6.5 Data-driven calibration}
\addcontentsline{toc}{subsection}{6.5 Data-driven calibration}

Each configuration here uses 100 null and 100 alternative replicates (standard error up to 0.05, about 0.02 at the false-alarm rates reported). With pair
permutation at \(K=99\), empirical false-alarm rates for the DOMI
statistic are 0.03--0.06 across all configurations (Supplementary Table B.8), and
0.02--0.09 across DOMI, HSIC and the copula test. Power differs from the Monte Carlo-calibrated values in both directions: on S2 it is lower (0.20/0.67 against 0.35/0.87), on S4 at \(r=0.7\) far higher (0.53 against 0.09). A likely cause is that the permutation replicas, built from the observed series, track the mixture of the two regimes, whereas the Monte Carlo threshold must cover the more variable pre-change regime held throughout; where the regimes differ sharply, as in S4, the permutation calibration, which the applications use, is the more powerful one.

\subsection*{6.6 Variability over feature and calibration draws}
\addcontentsline{toc}{subsection}{6.6 Variability over feature and calibration draws}

Every Monte Carlo-calibrated power above is computed from one draw of the random features and one draw of the null replicates that set the threshold. Supplementary Section B.5 shows that both draws move the powers, and the margins between statistics, well beyond the binomial error: on S4 at \(r=0.85\) the margin of DOMI over HSIC is 4.6, 15.6 and 40.4 percentage points under three calibration draws. The comparative claims of this paper are therefore quoted at the reported calibration draw.

\subsection*{6.7 Application: coupling regimes in surface
weather}
\addcontentsline{toc}{subsection}{6.7 Application: coupling regimes in
surface weather}

We analyse hourly Automated Synoptic Observing System (ASOS) observations of the Korea Meteorological Administration at twelve major stations with essentially complete hourly records, chosen for geographic spread before any analysis, over 2018--2025, about 70,100 hours at each station. The observations are reduced to anomalies against a climatology indexed by hour of day and day of year and smoothed over \(\pm 7\) days within each hour, which removes the diurnal and annual cycles from the level and lets the diurnal amplitude vary with season. Serial dependence differs sharply across variables: on the \(\pm 1\) year windows of stage two the effective sample sizes \(n(1-\hat\rho_1)/(1+\hat\rho_1)\) are 90--241 for temperature (median 142), 257--510 for humidity and 833--2895 for wind (full-record values in Supplementary Section B.8). Pressure is excluded: its lag-one autocorrelation of about 0.997
leaves effective sample sizes of only 104--136 over the full record.

We analyse three variable pairs at each station. Stage one segments with Holevo partitioning on
a two-week candidate grid, the penalty calibrated so that at most 5\% of copies permuted in twelve-week blocks produce any break; this returns 27 candidates
across the 36 station--pair analyses. Stage two re-tests each candidate
with the DOMI-difference statistic on a \(\pm 1\) year window by block
permutation of six-week blocks at \(K=9999\), and runs the
marginal-scale test of Section 4.7 on each variable with the same
blocks. The diagnostic of Section 4.3 rejects pair exchangeability in
all 36 analyses, every permutation p-value at its minimum of \(1/200\),
with recommended block lengths of 4.0--22.2 weeks (median 6.6). The
twelve-week blocks of stage one meet the recommendation in 33 of the 36
analyses, the six-week blocks of stage two in eight.

Block exchangeability is violated in these records mainly through the variances rather than the means. On the twelve-week blocks of the whole record
the exact block-permutation test rejects for the log block variance in
32 of the 36 station--variable series, against 9 for the block mean: the anomaly
transform removes the seasonal cycle from the level of each series but
not from its amplitude. For stage two, at 720 dates drawn at random, 20 per station--pair analysis and at least a
year from either end of the record, the unrestricted block permutation
returns a dependence p-value at or below 0.05 at 9.7\% of the dates, and the smaller of the two marginal p-values is at or below 0.05 at 23.2\% of them.
Permuting blocks only within the four calendar seasons brings the dependence test to 5.4\% at 0.05, and the marginal test to 11.4\%, against 9.75\% for the smaller of two independent p-values at level 0.05. Stage two therefore permutes within
seasons, with every other setting unchanged. These frequencies are a diagnostic of calibration sensitivity, not verified null rejection rates (Supplementary Section B.8).

Under the season-restricted stage two, five candidates reach
\(p\le 0.05\). One passes the Benjamini--Hochberg adjustment across the 27, at \(q\le 0.05\), and none passes the Benjamini--Yekutieli correction at \(q\le 0.05\), which allows arbitrary dependence among the candidates but assumes valid p-values, whereas these are nominal. It is Suwon humidity--wind, 2019-08-26 (Figure 3; dependence \(p=0.0005\); Benjamini--Hochberg \(q=0.0135\), Benjamini--Yekutieli \(q=0.0525\); marginal \(p=0.016\) and 0.042). The next candidate, Jeonju humidity--wind, 2019-10-21, has dependence \(p=0.0179\) (\(q=0.20\)). Tied values are ranked in time order; with ties broken at random, as the equivariance argument of Supplementary Section A.5 requires, the p-values of the five candidates at \(p\le 0.05\) and the candidate counts are unchanged (Supplementary Section B.8). Under the classification rule of Section 4.7 the Suwon candidate is a change in both dependence and scale; the marginal test is still somewhat liberal, so this label may overstate the marginal side, and in simulation a concurrent marginal shift inflates the dependence test (Supplementary Section B.19). Four candidates show a scale change without a significant dependence change, which a single-stage reading would have reported as dependence changes, and 22 are undetermined.

Neither of the two analyses aimed at the selection establishes a dependence change at the candidate dates. Selecting on
alternate six-week blocks and testing on the others leaves about two
blocks per season in each half-window, so the season-restricted test has
6 to 48 distinct permutations; 13 of the 33 test windows cannot reach
\(p\le 0.05\) at all, and no candidate is significant in either
direction. An exploratory resampling check of the whole procedure, which reuses the observed stage-one calibration, has as its null the absence of any change: on
a design with a marginal change only it rejects in 57.5\% of replicates
(Supplementary Section B.14). On these records the version of that check built on the stage-two statistic gives an exploratory \(p=0.105\), and stage one returns candidates in 23 of the 36 analyses against at most 8 in any permuted record; neither result isolates the dependence. The Suwon date is a candidate for domain analysis; attributing a physical mechanism to it is outside the scope of
this paper.

\begin{figure}[!t]
\centering
\includegraphics[width=\columnwidth]{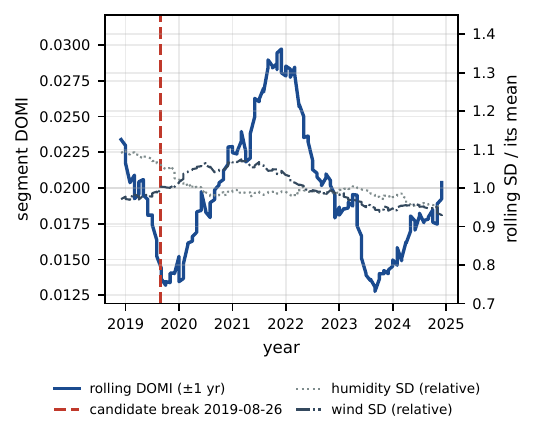}
\caption{The strongest candidate of the hourly weather analysis, Suwon humidity--wind at 2019-08-26: a rolling DOMI on a $\pm 1$ year window, with the marginal standard deviations of humidity and wind, each relative to its mean over the record, as dotted and dash-dotted curves. The rolling curve is a smoothed view whereas the test compares two segments, so a break need not coincide with an extremum of the curve.}
\end{figure}

\subsection*{6.8 Application: financial dependence
regimes}
\addcontentsline{toc}{subsection}{6.8 Application: financial dependence
regimes}

\begin{table}[!t]
\centering
\caption{Pre-specified US financial panel (Section 6.8; daily, 2011-08-23 to 2026-08-20). Stage one: whole-record single-break test (DOMI for gold, the combined statistic otherwise) under block permutation at block length $\hat b$ (days), Benjamini--Hochberg $q$ over the three pairs. Stage two: dependence re-test on the largest centred window of whole blocks (Blocks), Benjamini--Hochberg $q$ over the three pairs; Draws: feature draws with $p\le0.05$.}
\footnotesize
\setlength{\tabcolsep}{2pt}
\begin{tabular}{l ccc cccc}
\toprule
 & \multicolumn{4}{c}{Stage one} & \multicolumn{3}{c}{Stage two} \\
\cmidrule(lr){2-5}\cmidrule(lr){6-8}
Pair & $\hat b$ & $p$ & $q$ & Break & Blocks & $q$ & Draws \\
\midrule
Stocks--gold & 134 & 0.008 & 0.012 & 2013-03 & 4 & 0.098 & 0/20 \\
Stocks--Treasury & 173 & 0.020 & 0.020 & 2020-12 & 16 & 0.012 & 20/20 \\
Stocks--dollar & 133 & 0.005 & 0.012 & 2013-04 & 6 & 0.098 & 5/20 \\
\bottomrule
\end{tabular}
\label{tab:fin}
\end{table}

We fixed an analysis protocol, recorded with a hash and a time stamp before two of the series were obtained (Supplementary Section B.22): daily S\&P 500 returns against gold (GLD; primary statistic DOMI), against 10-year Treasury yield changes and against the dollar index (both with the combined statistic of Section 4.5), one whole-record single-break test per pair under block permutation at the diagnostic's block length (\(K=9999\)), and Benjamini--Hochberg at \(q=0.10\) over the three pairs; the Treasury pair had been examined exploratorily before. All three primary tests reject (Table~\ref{tab:fin}). For stocks and gold DOMI rejects while the whole-record Spearman change-point scan does not (\(p=0.23\)), and the DOMI p-value stays within 0.004--0.015 over 20 further permutation draws. In stage two, the stock--Treasury break is a dependence change without a detected scale change on 16 blocks, with \(p\le 0.05\) in all 20 feature draws; the gold and dollar breaks pass as dependence changes at the pre-set threshold on four and six blocks and would not pass at the \(q=0.05\) used for the weather candidates.

These results also illustrate why the calibration diagnostic matters: volatility clustering leads the diagnostic of Section 4.3 to reject pair exchangeability for every daily and weekly series here, and on weekly S\&P 500 and yield changes pair calibration would have returned four breaks where block calibration at \(\hat b=44\) weeks returns one (2020-02-27; Supplementary Figure B.5). Stage two classifies this break, and the monthly breaks of the stock--bond and KOSPI--USD/KRW pairs, as scale changes or undetermined (Supplementary Section B.8).

\section*{7. Discussion}
\addcontentsline{toc}{section}{7. Discussion}

Four properties hold together in the density-operator framework and, on the evidence of Section 6, in none of the tested baselines: (i) substantial
power on nonlinear, correlation-free dependence changes, where
correlation-based and rank-correlation methods are structurally insensitive; (ii) a
false-alarm rate near nominal under the M1 scale changes for the testing statistic (as for the copula test of {[}10{]}; not under the tested seasonal weather-margin shifts, Supplementary Section B.19), a property the two-stage procedure of Section 4.7 transfers in part to the segmentation, by re-testing and not by construction (false detections on a marginal-only design fall from 0.975 to 0.13; Supplementary Table B.7); (iii) an additive segmentation cost with a direct
information-theoretic interpretation, whose segmentation recovers the
multi-break benchmark more often than the kernel and matrix-information
baselines; and (iv) calibration from the observed series alone, exact in finite samples for windows and schemes fixed in advance under exchangeability. No method in the comparison, ours included, is uniformly best (Section 6.1); when the nature of the change is unknown, Section 4.5 combines DOMI with a correlation-targeted statistic.

The Gram form is at least as powerful as the random-feature form in every setting of Table 1, at the cost shown in Table 2; raising \(D\) from 8 to 16 recovers part of the difference on S2 at \(a=0.7\), all of it at \(a=0.9\) and none on S4 (Supplementary Table B.3), since \(D\) frequencies per variable represent the Gaussian kernel only approximately.

Two open problems are methodological: a selection-adjusted test for the
two-stage procedure that is specific to dependence --- the exploratory whole-procedure check of Supplementary Section B.14 reuses the observed stage-one calibration and responds to marginal change --- and a localiser that stays accurate when
the break is far from the centre. The theoretical ones are a limit law for the Gram form,
which requires the feature dimension to grow with the segment length;
localisation rates; the
mixing-case extension of the process-level Proposition P2$'$ of Supplementary Section
A.6, which needs an invariance principle for bounded strongly mixing
sequences {[}35{]}; a limit theory for the permutation-studentised scan and consistency of Holevo partitioning in the number of segments (numerical evidence for both in Supplementary Section B.3); and a detection-delay
lower bound via the quantum Fisher information, for which Section 3.4 is
the first step. A self-normalising studentisation in the spirit of {[}31,32{]} could add asymptotic validity under mixing to the exactness of P1 under exchangeability, trading the block-length choice for a bandwidth-free normalisation. At a fixed feature dimension, power and specificity both degrade as the variable-block dimension grows (Supplementary Section B.6), which suggests raising \(D\) with the block dimension, at a cost of \(O(D^6)\). The S3 limitation, one of estimation variance at strong dependence (Section 6.4), points to variance reduction, not a larger feature dimension.

\section*{8. Conclusion}
\addcontentsline{toc}{section}{8. Conclusion}

DOMIC treats a segment of a pair of variable blocks as a density
operator and monitors the density-operator mutual information of rank
features for change. The measure belongs to the matrix-based family of
information measures; this paper develops the inference and computation that make it a change-point detector, with finite-sample error control for windows and permutation schemes fixed in advance under the corresponding exchangeability: an exact permutation
calibration and its block form for serially dependent series, an
exchangeability diagnostic that chooses between them, a boundary limit law for the random-feature statistic with \(D\), the frequency draw and the bandwidth fixed, a bias correction, and a
segment-separable cost that drives penalised optimal partitioning. Exact partial traces reduce the random-feature form to prefix sums of small moment matrices; the Gram form computes the measure with the exact kernel under the same calibration.

On the correlation-free alternatives S2 and S4, DOMI gains at least 15 percentage points of localised power over rank-based and kernel statistics computed from the same inputs wherever it or one of them exceeds 0.10, at the reported calibration draw (4.6--40.4 points against HSIC on S4 across three draws, Section 6.6), and its Gram form matches or exceeds every classical statistic in all 16 non-Gaussian settings in which any statistic exceeds 0.10. Correlation-targeted statistics are more powerful than the random-feature form on monotone correlation, and the copula test of {[}10{]} detects copula-shape changes at fixed Kendall's \(\tau\) that DOMI does not detect at the deployed bandwidth. Under season-restricted block calibration, one of 27 candidates passes the multiplicity-corrected re-test on hourly weather records; its p-value is not adjusted for its selection, and neither the split-sample analysis nor the exploratory whole-procedure check establishes it. Under a hash-recorded protocol for US financial series, block-calibrated primary scans reject in all three pre-specified pairs, with the strongest stage-two evidence for stocks and Treasury yields. Extending the test beyond a pair of variable blocks, to joint dependence among three or more, is left to future work.

\section*{Code and data availability}
\addcontentsline{toc}{section}{Code and data availability}

All experiments use fixed random seeds. The implementation, the scripts
that produced every reported number and figure, and the outputs those
scripts wrote are publicly available at
\url{https://github.com/yunam-seo/domic-changepoint}. The synthetic experiments and the numbers cited in the proofs run from the release alone; the real-data analyses of Section 6 and the Supplementary Material also need the openly accessible third-party observations below, which are not redistributed. Hourly surface
observations for the twelve stations over 2018--2025 were obtained from the ASOS service of the Korea Meteorological Administration API Hub (apihub.kma.go.kr; free registration). Daily
closing levels over 2011--2026 are from Yahoo Finance: the S\&P 500
index (\textasciicircum{}GSPC), the CBOE 10-year Treasury yield index (\textasciicircum{}TNX), the KOSPI composite index (\textasciicircum{}KS11), the USD/KRW rate (KRW=X), the SPDR Gold Shares ETF (GLD) and the US Dollar Index (\mbox{DX-Y.NYB}).

\section*{Acknowledgment}
\addcontentsline{toc}{section}{Acknowledgment}

This research was supported by the Global -- Learning \& Academic
research institution for Master's$\cdot$PhD students, and Postdocs (G-LAMP)
Program of the National Research Foundation of Korea (NRF) grant funded
by the Ministry of Education (No.~RS-2026-25610465).

%% file: appendixA.tex
\section{Complete proofs}
\setcounter{equation}{0}\renewcommand{\theequation}{S.\arabic{equation}}
\setcounter{subsection}{-1}

Every result the main text invokes is proved here under the assumptions
stated with it. Where those assumptions differ from the deployed statistic,
the section says so: P2, P3 and Lemma 4(iii)--(iv) assume oracle margins,
and P4(iii) is proved for the paired feature map. Two statements are
\textbf{not} proved and are labelled where they occur: the mixing-case
extension of the process-level Proposition P2$'$ in \ref{sec:A6} (its i.i.d.
form is proved there), and the \emph{strict} positivity of \(c\) at
independence in \ref{sec:A10}(iv) (non-negativity is proved there via
\ref{sec:A12}, and computed counterexamples show that the sign is not
universal away from independence). Part (iv) of \ref{sec:A8} is proved
conditionally on a uniform spectral-decay hypothesis (H) that is not
established.

Where a statement rests on a computed number rather than on its proof alone, the number carries a tag of the form [NV-k], and the script that produces it is in the code release. These calculations illustrate finite-sample behaviour or document evaluated configurations; the mathematical claims rest on the proofs under their stated assumptions. Tags [NV-5] to [NV-7] mark simulation studies of properties the proofs do not establish (Section~\ref{sec:B3}). The permutation propositions P1, P5 and P7 carry no tag; their finite-sample behaviour in the simulated settings is illustrated in Section~\ref{sec:B3}, in Section 4.3 of the main text and in Section~\ref{sec:B17}.

Notation:
matrices are real symmetric unless stated otherwise; \(\rho,\sigma\)
denote density operators (positive semi-definite, unit trace);
\(\lambda^{\downarrow}(A)\) is the eigenvalue vector of \(A\) in
decreasing order; \(H(p)=-\sum_i p_i\log p_i\) (natural logarithm,
\(0\log 0:=0\)); \(S(\rho)=-\mathrm{tr}\,\rho\log\rho\).

\subsection{Preliminaries used throughout}\label{sec:A0}

Preliminaries (\ref{pre:2}), (\ref{pre:3}), (\ref{pre:5}) and (\ref{pre:6}) are standard and are stated without proof: see [27, Ch.~III] for the Ky Fan maximum principle and Ky Fan's majorisation for the eigenvalues of a sum, [27, Ch.~V] for the derivative of a trace function, and [26, Ch.~11] for Klein's inequality with its equality case. Preliminary (\ref{pre:4}), whose equality case Lemma 2 uses, and Preliminary (\ref{pre:7}), on which \ref{sec:A9}, \ref{sec:A10} and \ref{sec:A12} rest, are proved.

\begin{prelim}{1}[Spectral calculus]\label{pre:1}
For symmetric \(A\) with spectral
decomposition \(A=\sum_i\lambda_i u_iu_i^{\top}\) and \(f\) defined on
the spectrum, \(f(A):=\sum_i f(\lambda_i)\,u_iu_i^{\top}\). Then
\(\mathrm{tr}\,f(A)=\sum_i f(\lambda_i)\); in particular \(S(\rho)=H(\lambda(\rho))\).
\end{prelim}

\begin{prelim}{2}[Ky Fan maximum principle]\label{pre:2}
For symmetric
\(A\in\mathbb{R}^{d\times d}\) and \(1\le k\le d\),
\begin{equation}\sum_{i=1}^{k}\lambda_i^{\downarrow}(A)=\max\bigl\{\mathrm{tr}(U^{\top}AU): U\in\mathbb{R}^{d\times k},\ U^{\top}U=I_k\bigr\}.\end{equation}
\end{prelim}

\begin{prelim}{3}[Lidskii/Ky Fan majorisation]\label{pre:3}
For symmetric \(A,B\):
\(\lambda(A+B)\prec\lambda^{\downarrow}(A)+\lambda^{\downarrow}(B)\),
i.e.~for every \(k\),
\begin{equation}\sum_{i\le k}\lambda_i^{\downarrow}(A+B)\ \le\ \sum_{i\le k}\lambda_i^{\downarrow}(A)+\sum_{i\le k}\lambda_i^{\downarrow}(B),\end{equation}
with equality of the full sums at \(k=d\) (both sides equal
\(\mathrm{tr}\,A+\mathrm{tr}\,B\)).
\end{prelim}

\begin{prelim}{4}[Schur concavity of \(H\)]\label{pre:4}
If \(x\prec y\) (both probability vectors) then \(H(x)\ge H(y)\), with equality only when \(x\) is a permutation of \(y\).
\end{prelim}
\begin{proof}
By the
Hardy--Littlewood--Pólya theorem, \(x\prec y\) iff \(x=Dy\) for some
doubly stochastic \(D\). \(H\) is concave and symmetric; by Birkhoff's
theorem, \(D=\sum_{\tau}w_{\tau}P_{\tau}\) (a convex combination of
permutation matrices), so
\begin{equation}H(x)=H\Bigl(\sum_{\tau}w_{\tau}P_{\tau}y\Bigr)\ \ge\ \sum_{\tau}w_{\tau}H(P_{\tau}y)=H(y).\end{equation} Since \(H\) is strictly concave, equality forces \(P_\tau y\) to coincide for every \(\tau\) with \(w_\tau>0\), so \(H(x)=H(y)\) only when \(x\) is a permutation of \(y\).
\end{proof}

\begin{prelim}{5}[Klein's inequality]\label{pre:5}
\(S(\rho\Vert\sigma):=\mathrm{tr}\,\rho(\log\rho-\log\sigma)\ \ge\ 0\)
when \(\mathrm{supp}\,\rho\subseteq\mathrm{supp}\,\sigma\), with
equality iff \(\rho=\sigma\).
\end{prelim}

\begin{prelim}{6}[Derivative of a trace function]\label{pre:6}
Let \(f\) be analytic
on an open subset of \(\mathbb{C}\) containing the spectrum of the
symmetric matrix \(A\). Then \(g(t)=\mathrm{tr}\,f(A+tH)\) is
differentiable with \begin{equation}g'(0)=\mathrm{tr}\bigl(f'(A)\,H\bigr).\end{equation}
\end{prelim}

\begin{prelim}{7}[Second derivative of matrix entropy]\label{pre:7}
With
\(\rho\succ 0\) and \(\delta\) symmetric, in the eigenbasis of \(\rho\)
(eigenvalues \(\lambda_i\)):
\begin{equation}\frac{d^2}{dt^2}S(\rho+t\delta)\Big|_{t=0}=-\sum_{i,j}|\delta_{ij}|^2\,c(\lambda_i,\lambda_j),\qquad c(a,b)=\frac{\log a-\log b}{a-b},\quad c(a,a)=\frac1a.\end{equation}
Here \(c(a,b)\) is the reciprocal of the \textbf{logarithmic mean}
\(L(a,b)=(a-b)/(\log a-\log b)\).
\end{prelim}
\begin{proof}
\(\frac{d}{dt}S(\rho+t\delta)=-\mathrm{tr}\bigl[(\log(\rho+t\delta)+I)\,\delta\bigr]\)
by (\ref{pre:6}). Differentiating again requires
\(\frac{d}{dt}\log(\rho+t\delta)\); use the integral representation
\begin{equation}\log X=\int_0^{\infty}\Bigl[\frac{1}{u+1}I-(uI+X)^{-1}\Bigr]\,du,\end{equation}
whence
\(\frac{d}{dt}\log(\rho+t\delta)\big|_{0}=\int_0^{\infty}(uI+\rho)^{-1}\delta\,(uI+\rho)^{-1}du\).
Therefore
\begin{equation}\begin{aligned}\frac{d^2}{dt^2}S&=-\mathrm{tr}\Bigl[\delta\int_0^{\infty}(uI+\rho)^{-1}\delta\,(uI+\rho)^{-1}du\Bigr]\\
&=-\sum_{i,j}|\delta_{ij}|^2\int_0^{\infty}\frac{du}{(u+\lambda_i)(u+\lambda_j)}=-\sum_{i,j}|\delta_{ij}|^2\,c(\lambda_i,\lambda_j),\end{aligned}\end{equation}
computing the scalar integral
\(\int_0^{\infty}\frac{du}{(u+a)(u+b)}=\frac{\log a-\log b}{a-b}\) for
\(a\ne b\), and \(=1/a\) for \(a=b\).
\end{proof}

\subsection{Lemma 1 (commuting reduction)}\label{sec:A1}

\begin{statement}
Let \(\rho_1,\rho_2\) be density operators with
\(\rho_1\rho_2=\rho_2\rho_1\). Then there is an orthonormal basis
\(\{u_k\}\) with \(\rho_i=\sum_k\lambda_{ik}u_ku_k^{\top}\), and

\begin{enumerate}
\renewcommand{\labelenumi}{(\roman{enumi})}
\item
  \(\mathrm{QJSD}(\rho_1,\rho_2):=S\bigl(\tfrac{\rho_1+\rho_2}{2}\bigr)-\tfrac12S(\rho_1)-\tfrac12S(\rho_2)=\mathrm{JSD}(\lambda_1,\lambda_2)\);
\item
  \(F(\rho_1,\rho_2):=\mathrm{tr}\sqrt{\sqrt{\rho_1}\,\rho_2\sqrt{\rho_1}}=\sum_k\sqrt{\lambda_{1k}\lambda_{2k}}\),
  hence the squared Bures distance \(D_B^2=2(1-F)\) equals twice the
  squared Hellinger distance between \(\lambda_1\) and \(\lambda_2\),
  under the convention
  \(H^2(p,q)=\tfrac12\sum_k(\sqrt{p_k}-\sqrt{q_k})^2\);
\item
  \(S(\rho_1\Vert\rho_2)=\mathrm{KL}(\lambda_1\Vert\lambda_2)\) (when
  \(\mathrm{supp}\,\lambda_1\subseteq\mathrm{supp}\,\lambda_2\)). Here
  \(\lambda_i=(\lambda_{ik})_k\) is the coordinate vector of \(\rho_i\)
  in the common basis \(\{u_k\}\), indexed across \(i\) by the same
  \(k\); it is not the independently sorted spectrum
  \(\lambda^{\downarrow}(\rho_i)\) of the notation, and the two agree
  only when the eigenvalues are similarly ordered along \(\{u_k\}\) ---
  the extra condition Lemma 2 states for its own reduction.
\end{enumerate}
\end{statement}

\begin{proof}
Commuting symmetric matrices are simultaneously
diagonalisable (standard: eigenspaces of \(\rho_1\) are invariant under
\(\rho_2\)). In the common basis every quantity is computed by (\ref{pre:1})
coordinate-wise: \(\tfrac{\rho_1+\rho_2}{2}\) has eigenvalues
\(\tfrac{\lambda_{1k}+\lambda_{2k}}{2}\), so (i) is the definition of
classical JSD. For (ii),
\(\sqrt{\rho_1}\,\rho_2\sqrt{\rho_1}=\sum_k\lambda_{1k}\lambda_{2k}u_ku_k^{\top}\),
whose square root has eigenvalues \(\sqrt{\lambda_{1k}\lambda_{2k}}\).
For (iii), \(\log\rho_i=\sum_k\log\lambda_{ik}\,u_ku_k^{\top}\) and the
trace reduces to the classical sum.
\end{proof}

\subsection{Lemma 2 (spectral-alignment decomposition)}\label{sec:A2}

\begin{statement}
For any density operators \(\rho_1,\rho_2\):
\begin{equation}\mathrm{QJSD}(\rho_1,\rho_2)=\mathrm{JSD}\bigl(\lambda^{\downarrow}(\rho_1),\lambda^{\downarrow}(\rho_2)\bigr)+\Delta_{\mathrm{nc}},\qquad \Delta_{\mathrm{nc}}\ge 0.\end{equation}
Moreover \(\Delta_{\mathrm{nc}}=0\) if and only if \(\rho_1,\rho_2\)
commute and their eigenvalues are similarly ordered along a common
basis; the proof below gives the sufficient direction and the Remark the necessary one.
\end{statement}
\begin{proof}
Since
\(S(\rho_i)=H(\lambda^{\downarrow}(\rho_i))\) by (\ref{pre:1}),
\begin{equation}\Delta_{\mathrm{nc}}=\mathrm{QJSD}-\mathrm{JSD}=S\bigl(\tfrac{\rho_1+\rho_2}{2}\bigr)-H\bigl(\tfrac{\lambda^{\downarrow}_1+\lambda^{\downarrow}_2}{2}\bigr).\end{equation}
By (\ref{pre:3}) with \(A=\rho_1/2\), \(B=\rho_2/2\):
\(\lambda\bigl(\tfrac{\rho_1+\rho_2}{2}\bigr)\prec\tfrac{\lambda^{\downarrow}_1+\lambda^{\downarrow}_2}{2}\).
By (\ref{pre:4}) and (\ref{pre:1}),
\begin{equation}S\bigl(\tfrac{\rho_1+\rho_2}{2}\bigr)=H\Bigl(\lambda\bigl(\tfrac{\rho_1+\rho_2}{2}\bigr)\Bigr)\ \ge\ H\bigl(\tfrac{\lambda^{\downarrow}_1+\lambda^{\downarrow}_2}{2}\bigr).\end{equation}
Hence \(\Delta_{\mathrm{nc}}\ge 0\). If the operators commute with
similarly ordered eigenvalues,
\(\lambda\bigl(\tfrac{\rho_1+\rho_2}{2}\bigr)=\tfrac{\lambda^{\downarrow}_1+\lambda^{\downarrow}_2}{2}\)
exactly and \(\Delta_{\mathrm{nc}}=0\).
\end{proof}
\begin{remark}[the necessary direction]
Since \(H\) is \emph{strictly} Schur-concave,
\(\Delta_{\mathrm{nc}}=0\) forces
\(\lambda^{\downarrow}\bigl(\tfrac{\rho_1+\rho_2}{2}\bigr)=\tfrac{\lambda^{\downarrow}_1+\lambda^{\downarrow}_2}{2}\),
the equality case of (\ref{pre:3}); by (\ref{pre:2}) that requires one
\(k\)-dimensional subspace to attain the maximum against \(\rho_1\) and
against \(\rho_2\) simultaneously, for every \(k\), hence a common
eigenbasis ordering both spectra decreasingly. Commutation on its own is
not sufficient: commuting operators with oppositely ordered spectra give
\(\Delta_{\mathrm{nc}}>0\). Nothing below uses the converse.
\end{remark}

\subsection{Energy-weighted cost and Holevo gain}\label{sec:A3}

Throughout, every part is assumed to have strictly positive energy, so
that all density operators below are defined and the length proportions
are non-degenerate. Let the features be \(z_t\in\mathbb{R}^p\), a
segment \(s\) with \(m=|s|\), \(M_s=m^{-1}\sum_{t\in s}z_tz_t^{\top}\),
\(E_s=m\,\mathrm{tr}\,M_s\), \(\rho_s=M_s/\mathrm{tr}\,M_s\), and cost
\(C(s)=E_s\,S(\rho_s)\).

\textbf{(a) Mixture identity.} If \(s=L\sqcup R\) (disjoint,
\(m=m_L+m_R\)), then \(M_s=(m_LM_L+m_RM_R)/m\), so
\(\mathrm{tr}\,M_s=(E_L+E_R)/m\) and
\begin{equation}\rho_s=\frac{M_s}{\mathrm{tr}\,M_s}=\frac{m_LM_L+m_RM_R}{E_L+E_R}=q_L\rho_L+q_R\rho_R,\qquad q_i=\frac{E_i}{E_L+E_R},\end{equation}
because \(m_iM_i=E_i\rho_i\) (indeed \(M_i=(E_i/m_i)\rho_i\)).\qed

\textbf{(b) Split gain = Holevo information.}
\begin{equation}C(s)-C(L)-C(R)=E_s\bigl[S(q_L\rho_L+q_R\rho_R)-q_LS(\rho_L)-q_RS(\rho_R)\bigr]=:E_s\,\chi.\end{equation}
\begin{proof}
Direct substitution using \(E_s=E_L+E_R\) and (a); the
bracket is the Holevo information of the ensemble \(\{(q_i,\rho_i)\}\). By (c) the gain of this raw cost is non-negative; the bias-corrected cost of Section 4.4 subtracts a permutation mean from each segment cost, so its gains can be negative, and Section 4.4 partitions without pruning.
\end{proof}

\textbf{(c) \(\chi\ge 0\) with equality iff \(\rho_L=\rho_R\).} The
identity
\begin{equation}\chi=\sum_i q_i\,S(\rho_i\Vert\bar\rho),\qquad \bar\rho=\sum_i q_i\rho_i,\end{equation}
holds whenever \(\mathrm{supp}\,\rho_i\subseteq\mathrm{supp}\,\bar\rho\)
(true since \(\bar\rho\succeq q_i\rho_i\)):
\begin{equation}\sum_i q_iS(\rho_i\Vert\bar\rho)=\sum_i q_i\,\mathrm{tr}\,\rho_i\log\rho_i-\mathrm{tr}\Bigl(\sum_i q_i\rho_i\Bigr)\log\bar\rho=-\sum_i q_iS(\rho_i)+S(\bar\rho)=\chi.\end{equation}
Klein's inequality (\ref{pre:5}) gives \(\chi\ge 0\), and \(\chi=0\) forces
\(\rho_i=\bar\rho\) for every \(i\) with \(q_i>0\).\qed

\subsection{Lemma 3 (exact sample-level partial trace)}\label{sec:A4}

\begin{statement}
Let \(\varphi_{X,t},\varphi_{Y,t}\in\mathbb{R}^{D}\)
with \(\lVert\varphi_{X,t}\rVert=\lVert\varphi_{Y,t}\rVert=1\) for all
\(t\). (The identity \(\mathrm{Tr}_Y\rho_{XY}=\rho_X\) needs only
\(\lVert\varphi_{Y,t}\rVert\) to be constant in \(t\), trace
normalisation absorbing the constant; the unit trace of \(\rho_{XY}\),
the symmetric identity \(\mathrm{Tr}_X\rho_{XY}=\rho_Y\), and the
entropy bounds of \ref{sec:A10} use both norms.) Set
\(\psi_t=\varphi_{X,t}\otimes\varphi_{Y,t}\) and
\(\rho_{XY}=m^{-1}\sum_t\psi_t\psi_t^{\top}\). Then
\begin{equation}\mathrm{Tr}_Y\,\rho_{XY}=m^{-1}\sum_t\varphi_{X,t}\varphi_{X,t}^{\top}=:\rho_X.\end{equation}
\end{statement}
\begin{proof}
The partial trace is the linear map determined by
\(\mathrm{Tr}_Y(A\otimes B)=(\mathrm{tr}\,B)\,A\) on elementary tensors.
For each \(t\),
\(\psi_t\psi_t^{\top}=(\varphi_{X,t}\varphi_{X,t}^{\top})\otimes(\varphi_{Y,t}\varphi_{Y,t}^{\top})\)
--- verified entrywise: \((a\otimes b)(a\otimes b)^{\top}\) has entries
\(a_ia_kb_jb_l\) at position \(((i,j),(k,l))\), which is exactly
\((aa^{\top})\otimes(bb^{\top})\). Hence
\begin{equation}\mathrm{Tr}_Y(\psi_t\psi_t^{\top})=\mathrm{tr}(\varphi_Y\varphi_Y^{\top})\,\varphi_X\varphi_X^{\top}=\lVert\varphi_{Y,t}\rVert^2\varphi_{X,t}\varphi_{X,t}^{\top}=\varphi_{X,t}\varphi_{X,t}^{\top}.\end{equation}
Average over \(t\). (Unit trace of \(\rho_{XY}\): \(\mathrm{tr}\,\psi\psi^{\top}=\lVert\psi\rVert^2=\lVert\varphi_X\rVert^2\lVert\varphi_Y\rVert^2=1\).)
\end{proof}

\begin{corollary}[independence \(\Rightarrow\) \(I=0\) at population level]
If \(X\perp\!\!\!\perp Y\) then
\begin{equation}\mathbb{E}[\psi\psi^{\top}]=\mathbb{E}\bigl[(\varphi_X\varphi_X^{\top})\otimes(\varphi_Y\varphi_Y^{\top})\bigr]=\mathbb{E}[\varphi_X\varphi_X^{\top}]\otimes\mathbb{E}[\varphi_Y\varphi_Y^{\top}]=M_X\otimes M_Y,\end{equation}
using independence in the middle equality. The spectrum of
\(A\otimes B\) is \(\{a_ib_j\}\), so
\begin{equation}S(A\otimes B)=-\sum_{ij}a_ib_j(\log a_i+\log b_j)=S(A)+S(B)\end{equation} for unit
traces. Hence \(I=S(M_X)+S(M_Y)-S(M_X\otimes M_Y)=0\).\qed
\end{corollary}

\subsection{Proposition P1 (finite-sample exactness of the pair-permutation p-value)}\label{sec:A5}

\begin{setting}
Data \(W=(w_1,\dots,w_n)\), \(w_t=(x_t,y_t)\). Null
hypothesis \(H_0\): the law of \(W\) is exchangeable,
i.e.~\(W\circ\sigma:=(w_{\sigma(1)},\dots,w_{\sigma(n)})\) has the same
distribution as \(W\) for every fixed permutation \(\sigma\). Draw
\(\sigma_1,\dots,\sigma_K\) i.i.d. uniform on the symmetric group
\(S_n\), independent of \(W\), and set \(W^{(0)}=W\),
\(W^{(k)}=W\circ\sigma_k\). Let the statistics be produced by a map
\begin{equation}T_k=g\bigl(W^{(k)};\ \{W^{(0)},\dots,W^{(K)}\}\bigr),\end{equation} where \(g\)
may depend on the whole collection only through its \emph{unordered
multiset} (our studentised maximum: per-\(t\) mean and standard
deviation computed from all \(K{+}1\) curves). Define
\begin{equation}p=\frac{1+\#\{k\ge 1: T_k\ge T_0\}}{K+1}.\end{equation}
\end{setting}
\begin{claim}
Under
\(H_0\), \(P(p\le\alpha)\le\alpha\) for every \(\alpha\in(0,1)\);
moreover \(P(p\le\alpha)=\lfloor\alpha(K{+}1)\rfloor/(K{+}1)\) if the
\(T_k\) are a.s. distinct.
\end{claim}
\begin{proof}
\textbf{Step 1 (re-indexing).} Fix \(j\in\{1,\dots,K\}\) and swap the roles of coordinates \(0\) and \(j\): replace \(W\) by \(W\circ\sigma_j\), \(\sigma_j\) by \(\sigma_j^{-1}\), and each other \(\sigma_k\) by \(\sigma_j^{-1}\circ\sigma_k\). With the convention \((W\circ\sigma)\circ\pi=W\circ(\sigma\circ\pi)\) this leaves every record \(W^{(k)}\), \(k\ne 0,j\), unchanged and interchanges \(W^{(0)}\) and \(W^{(j)}\). It also leaves the joint law of the record and the permutation draws unchanged: conditionally on \(\sigma_j=s\), \(W\circ s\) has the law of \(W\) under \(H_0\), and left translation by \(s^{-1}\) preserves the uniform law on \(S_n\), so the new draws are again i.i.d. uniform and independent of the new record. Since \(g\) sees the collection only as an unordered multiset, the law of \((T_0,\dots,T_K)\) is invariant under the transposition \((0\,j)\); these transpositions generate \(S_{K+1}\), so \((T_0,\dots,T_K)\) is exchangeable.

\textbf{Step 2 (rank argument, with ties).} Put \(R_j=\#\{k: T_k\ge T_j\}\), so \(p=R_0/(K+1)\); ties are counted in full. For any fixed real vector and any integer \(r\ge 0\), \(\sum_j\mathbf 1\{R_j\le r\}\le r\), because \(R_j\le r\) holds exactly for those \(j\) in the union of the top tie-groups whose cumulative size is at most \(r\). By Step 1 every term has the same expectation, so \(P(R_0\le r)\le r/(K+1)\), and \(r=\lfloor\alpha(K{+}1)\rfloor\) gives \(P(p\le\alpha)\le\alpha\). If the \(T_k\) are a.s. distinct, the first inequality is an equality.
\end{proof}
\begin{remark}
The proof covers the symmetric studentisation used throughout, in which the studentising moments are computed from all \(K{+}1\) curves; the leave-one-out variant is not covered and is not used.
\end{remark}

\begin{remark}[testing exchangeability itself]
The proof above places
no restriction on the statistic: the map \(g\) is arbitrary, so the
permutation p-value of \emph{any} functional of the sample is exactly
valid under exchangeability. Taking the functional to be a measure of
serial structure --- an autocorrelation portmanteau of the centred
ranks, say --- therefore yields a finite-sample exact test \emph{of}
exchangeability; the diagnostic protocol of Section 4.3 applies this
with five such statistics and controls its family-wise level by the
union bound. The test is a refutation instrument only, and the
certification direction is impossible rather than merely unproven. Let
\(\varphi\) be any test with \(E[\varphi]\le\alpha\) under every
exchangeable law, let \(x\) be a sample path with at least two distinct
coordinates, and let \(O_x=\{x\circ\sigma:\sigma\in S_n\}\) be its
permutation orbit. The uniform distribution on \(O_x\) is exchangeable,
so the level constraint gives
\(|O_x|^{-1}\sum_{z\in O_x}\varphi(z)\le\alpha\). Each point mass
\(\delta_z\) with \(z\in O_x\) is a non-exchangeable law, and the power
of \(\varphi\) against it is \(\varphi(z)\); the constraint says exactly
that the average power over these alternatives is at most \(\alpha\) on
every orbit --- a non-randomised test rejects on at most an \(\alpha\)
fraction of each orbit, so the test has zero power against at least a
\(1-\alpha\) fraction of the point-mass alternatives. Hence no procedure
can certify exchangeability from a single realisation, and non-rejection carries no positive content.
\end{remark}

\begin{remark}[randomness drawn independently of the data]
The statistic may use randomness \(V\) drawn independently of the data: the index subsample from which the median-heuristic bandwidth is taken, or a uniform jitter that breaks tied values at random. Apply the proof above, or that of \ref{sec:A11}, to the augmented record \((w_t,V_t)_t\), permuted jointly. Under the null it is exchangeable, because \(V\) is independent of the data and invariant in law under permutation, so the p-value of any statistic of the augmented record is exact; the guarantee averages over \(V\), as it does over the permutation draws. Computing the features once and permuting them is that test when the features of a permuted record are the permuted features. This holds for the bandwidth subsample, whose rank values move with the permutation, and for random tie-breaking; it fails for tied values ranked in time order, where the observed rank vector increases with time within each tie and a permuted copy in general does not. The weather stage two of Section 6.7 computes the features once, ranks ties in time order, and is therefore not covered as implemented; Section~\ref{sec:B8} repeats it with ties broken at random.
\end{remark}

\subsection{Proposition P2 (asymptotic normality of the segment DOMI plug-in)}\label{sec:A6}

\begin{assumptions}
(A0) \emph{Oracle probability integral transform.}
The features are formed from the true marginal transforms,
\(\varphi_{X,t}=\varphi(F_X(x_t))\) and likewise for \(Y\), rather than
from empirical ranks. (A1) Within the segment,
\(\psi_t=\varphi_{X,t}\otimes\varphi_{Y,t}\) are i.i.d. (or strictly
stationary, ergodic and \(\alpha\)-mixing at a rate sufficient for the
multivariate central limit theorem,
\(\sum_k\alpha(k)^{\delta/(2+\delta)}<\infty\) for some \(\delta>0\),
the accompanying moment condition being automatic from the unit norm of
(A2$'$); the CLT then holds with the long-run variance). (A2) The
population moments \(M_X=\mathbb{E}[\varphi_X\varphi_X^{\top}]\),
\(M_Y\), \(M_{XY}=\mathbb{E}[\psi\psi^{\top}]\) are positive definite.
(A2$'$) The features have unit norm,
\(\lVert\varphi_{X,t}\rVert=\lVert\varphi_{Y,t}\rVert=1\); this gives
boundedness, \(\lVert\psi\psi^{\top}\rVert\le 1\); the unit norm is used
again in Step 2, for the cancellation of the identity terms, and in the boundary remark below.
\end{assumptions}

\begin{statement}
With
\(\hat I=S(\hat M_X)+S(\hat M_Y)-S(\hat M_{XY})\) (unit-norm features
give \(\mathrm{tr}\,\hat M=1\) exactly, so no trace normalisation is
needed) and \(I\) its population value,
\begin{equation}\sqrt{m}\,(\hat I-I)\ \overset{d}{\longrightarrow}\ N(0,\sigma^2),\qquad \sigma^2=\sum_{k\in\mathbb{Z}}\mathrm{Cov}\bigl(h_c(\psi_0),h_c(\psi_k)\bigr),\end{equation}
which reduces to \(\mathrm{Var}[h_c(\psi_1)]\) in the i.i.d. case, where
\(h_c=h-I\) is the centred influence function and
\begin{equation}h(\psi)=-\langle\log M_X,\varphi_X\varphi_X^{\top}\rangle-\langle\log M_Y,\varphi_Y\varphi_Y^{\top}\rangle+\langle\log M_{XY},\psi\psi^{\top}\rangle,\end{equation}
with \(\langle A,B\rangle=\mathrm{tr}\,AB\). Since
\(\mathbb{E}[\psi\psi^{\top}]=M_{XY}\) and \(\langle\log M,M\rangle=-S(M)\),
\(\mathbb{E}\,h(\psi)=S(M_X)+S(M_Y)-S(M_{XY})=I\), so \(h\) itself is not
centred unless \(I=0\), and every sum below is a sum of \(h_c\). The limit
is a non-degenerate normal law provided \(\sigma^2>0\); if \(\sigma^2=0\)
the display holds with the point mass at zero as its limit.
\end{statement}
\begin{proof}
\textbf{Step 1 (CLT for moments).} We give the i.i.d.
case; under the mixing alternative in (A1) the same step holds by a CLT
for stationary \(\alpha\)-mixing sequences with summable coefficients,
with the long-run covariance in place of \(\Sigma_M\).
\(\hat M-M=m^{-1}\sum_t(\psi_t\psi_t^{\top}-M)\) is a mean of i.i.d.
bounded random matrices; the multivariate CLT gives
\(\sqrt{m}\,\mathrm{vec}(\hat M-M)\to N(0,\Sigma_M)\) jointly for the
three moment matrices (they are deterministic linear images of the same
per-sample tensor, since
\(\varphi_X\varphi_X^{\top}=\mathrm{Tr}_Y(\psi\psi^{\top})\) by \ref{sec:A4}, so
joint normality is automatic).

\textbf{Step 2 (differentiability).} On
the open set \(\{M\succ 0\}\), \(M\mapsto S(M)=-\mathrm{tr}\,M\log M\)
is Fréchet differentiable with derivative
\(H\mapsto-\mathrm{tr}\bigl((\log M+I)H\bigr)\) by (\ref{pre:6})
(\(f(x)=-x\log x\) is analytic on \((0,\infty)\), which contains the
spectrum by (A2); eventually \(\hat M\succ 0\) a.s.). In the directional
derivative of \(\hat I\) the identity terms cancel, because
\(\mathrm{tr}\bigl(I\,(\varphi\varphi^{\top}-M)\bigr)=0\); using the
exact linearity \(\hat M_X=\mathrm{Tr}_Y(\hat M_{XY})\) of \ref{sec:A4}, the
three directional terms collapse to
\(DI[\hat M-M]=m^{-1}\sum_t h_c(\psi_t)\): each observation contributes
\(h(\psi_t)\), and the three population terms
\(\langle\log M_X,M_X\rangle+\langle\log M_Y,M_Y\rangle-\langle\log M_{XY},M_{XY}\rangle=-I\)
supply the centring.

\textbf{Step 3 (delta method).}
\(\sqrt{m}(\hat I-I)=m^{-1/2}\sum_t h_c(\psi_t)+\sqrt{m}\,R\), the first
term converging to \(N(0,\sigma^2)\) by the central limit theorem. The
remainder satisfies \(|R|\le C\lVert\hat M-M\rVert^2\) on the event
\(\lVert\hat M-M\rVert\le\varepsilon\) (the second derivative is bounded
on a compact neighbourhood of the population moments inside the
positive-definite cone, by (\ref{pre:7}) with all eigenvalues bounded below),
and \(\sqrt{m}\lVert\hat M-M\rVert^2=O_p(m^{-1/2})\to 0\); Slutsky's
lemma then completes the proof.
\end{proof}

\emph{Boundary remark: at the independence null the limit of P2 is the point mass at zero.} If \(M_{XY}=M_X\otimes M_Y\), which by Klein's inequality (\ref{pre:5}) is equivalent to \(I=0\), then \(\log M_{XY}=\log M_X\otimes I+I\otimes\log M_Y\), and for unit-norm features
\begin{equation}\begin{aligned}\langle\log M_{XY},\psi\psi^{\top}\rangle&=\langle\log M_X,\varphi_X\varphi_X^{\top}\rangle\lVert\varphi_Y\rVert^2+\lVert\varphi_X\rVert^2\langle\log M_Y,\varphi_Y\varphi_Y^{\top}\rangle\\
&=\langle\log M_X,\varphi_X\varphi_X^{\top}\rangle+\langle\log M_Y,\varphi_Y\varphi_Y^{\top}\rangle,\end{aligned}\end{equation}
so \(h\equiv 0\): the first-order term vanishes and \(\sigma^2=0\) at every law with \(I=0\), including the dependent moment-matched laws of \ref{sec:A8}(ii). At independence the correct normalisation is \(m\) and the limit is the weighted-\(\chi^2\) law of P6 (Section~\ref{sec:A12}); at other laws with \(I=0\) the \(m\) scale is again the relevant one, but the limit is not covered by P6, whose setting assumes independent components. Section~\ref{sec:B3} shows the two scalings numerically.

\emph{On (A0).} With global empirical ranks the summands are, for i.i.d. pairs, exchangeable but not independent, and the estimated margins add a Hájek-projection term to the limit variance, so \(\sigma^2\) above is not claimed to be the variance of the rank-based statistic. The deployed calibration rests on P1, which uses only exchangeability of the pairs.

\begin{proposition}{P2$'$}[process-level limit away from the boundary; i.i.d. segments, oracle margins]
Under (A0), (A2), (A2$'$), with
\(\{\psi_t\}\) i.i.d. and \(0<\varepsilon<\tfrac12\), and with
\(\sigma^2=\mathrm{Var}[h_c(\psi_1)]\ge0\)
(for \(\sigma^2=0\) the limit below is the zero process):
writing \(\hat I_L(\lambda)\), \(\hat I_R(\lambda)\) for the two segment
plug-ins at split fraction \(\lambda\), jointly in
\(\ell^{\infty}[\varepsilon,1-\varepsilon]\),
\begin{equation}\sqrt n\,\big(\hat I_L(\lambda)-I,\ \hat I_R(\lambda)-I\big)\ \Rightarrow\ \Big(\frac{\sigma\,B(\lambda)}{\lambda},\ \frac{\sigma\,(B(1)-B(\lambda))}{1-\lambda}\Big),\end{equation}
with \(B\) a standard Brownian motion. Consequently
\(\sqrt n\,(\hat I_L-\hat I_R)\) converges to the continuous Gaussian
process
\(\sigma\big[B(\lambda)/\lambda-(B(1)-B(\lambda))/(1-\lambda)\big]\),
and for any continuous weight \(w\),
\(\sup_{\lambda\in[\varepsilon,1-\varepsilon]}\sqrt n\,w(\lambda)\,|\hat I_L(\lambda)-\hat I_R(\lambda)|=O_p(1)\).
\end{proposition}

\begin{proof}
Donsker's invariance principle for the i.i.d. bounded
centred scalars \(h_c(\psi_t)=h(\psi_t)-I\) gives
\(n^{-1/2}\sum_{t\le\lambda n}h_c(\psi_t)\allowbreak\Rightarrow\sigma B(\lambda)\)
with a continuous limit, jointly with the matrix partial-sum process as
in the Corollary of \ref{sec:A12}, and
\(\sup_{\lambda\ge\varepsilon}\|\hat M_{[0,\lambda n)}-M\|=O_p(n^{-1/2})\).
Step 2 of P2 --- the collapse of the three directional terms to the
segment average of the centred influence function \(h_c\) --- is an
algebraic identity that holds
per segment, and Step 3's remainder bound holds uniformly on the event
that every segment moment lies in a fixed positive-definite
neighbourhood of \(M\) (probability tending to one):
\(\sup_\lambda\sqrt n\,|R(\lambda)|\le C\sqrt n\,\sup_\lambda\|\hat M_L(\lambda)-M\|^2=O_p(n^{-1/2})\).
Hence
\(\sqrt n(\hat I_L(\lambda)-I)=\lambda_n^{-1}\,n^{-1/2}\sum_{t\le\lambda n}h_c(\psi_t)+o_p(1)\)
uniformly, where \(\lambda_n=\lfloor\lambda n\rfloor/n\); the right
segment is handled symmetrically via the complementary partial sum, and
the continuous mapping theorem gives the display and the supremum bound.
\end{proof}

\begin{remark}[mixing, ranks, scope]
Under the mixing alternative of (A1) an invariance principle for bounded strongly mixing sequences [35] would replace Donsker's theorem; that extension, and a version for empirical ranks, are not proved here, and \ref{sec:A7} uses only the i.i.d. case.
\end{remark}

\subsection{Proposition P3 (detection consistency; localisation under an identifiability condition)}\label{sec:A7}

\begin{assumptions}
A single change at \(\tau=\lfloor n\theta\rfloor\)
with \(\theta\in(\varepsilon,1-\varepsilon)\) interior to the search
range: \(\{\psi_t\}_{t\le\tau}\) i.i.d. with moment \(M^{(1)}\),
\(\{\psi_t\}_{t>\tau}\) i.i.d. with moment \(M^{(2)}\), both positive
definite, and the population DOMI values differ, \(I(M^{(1)})\ne I(M^{(2)})\),
where for a joint moment \(M\) we write
\(I(M)=S(\mathrm{Tr}_YM)+S(\mathrm{Tr}_XM)-S(M)\). We write \(H_0\) for
the sub-hypothesis \(M^{(1)}=M^{(2)}\). As in \ref{sec:A6} this presumes the
oracle transform (A0); with \emph{global} ranks the pre- and post-change
feature laws are coupled, because the rank transform is computed across
the change point. In S1, S3, S4 and G1--G6 the margins are held fixed across regimes, so the global rank transform converges to a single fixed map and the assumption is asymptotically recovered; in S2, MB and M1 a margin changes and it is not. We state it as an assumption.
\end{assumptions}

\textbf{(a) Uniform convergence of the criterion.} For
\(s\in[\varepsilon,1-\varepsilon]\), the left-segment moment satisfies
\begin{equation}\hat M_{[0,\lfloor ns\rfloor)}\ \longrightarrow\ M_L(s):=\frac{\min(s,\theta)\,M^{(1)}+(s-\theta)_+\,M^{(2)}}{s}\quad\text{a.s., uniformly in } s\end{equation}
(SLLN plus monotone bracketing of partial sums over a finite
\(\varepsilon\)-grid and Lipschitz dependence on \(s\)), and
symmetrically
\begin{equation}\hat M_{[\lfloor ns\rfloor,n)}\ \longrightarrow\ M_R(s):=\frac{(\theta-s)_+M^{(1)}+\bigl(1-\max(s,\theta)\bigr)M^{(2)}}{1-s}\quad\text{a.s., uniformly in }s.\end{equation}
All limits stay in a compact positive-definite set, on which
\(S(\cdot)\) is uniformly continuous, so
\begin{equation}Q_n(s):=\bigl|\hat I_L(\lfloor ns\rfloor)-\hat I_R(\lfloor ns\rfloor)\bigr|\ \longrightarrow\ q(s):=\bigl|I(M_L(s))-I(M_R(s))\bigr|\quad\text{a.s., uniformly.}\end{equation}
\textbf{(b) Detection.} \(q(\theta)=|I(M^{(1)})-I(M^{(2)})|>0\), while
under \(H_0\) the same argument gives \(Q_n\to 0\) uniformly; so any
threshold sequence \(c_n\to 0\) with \(\sqrt{n}\,c_n\to\infty\)
separates the two, and power tends to one. The second half of that
statement --- that under \(H_0\) the \emph{supremum}
\(\sup_{s\in[\varepsilon,1-\varepsilon]}\sqrt{n}\,Q_n(s)\) is
\(O_p(1)\), rather than merely \(\sqrt{n}Q_n(s)=O_p(1)\) for each fixed
\(s\) --- is a tightness statement. Under \(H_0\) both regimes have the same moment \(M\). Kolmogorov's maximal inequality, applied within each regime entrywise to the centred moment matrices, bounds the left- and right-segment moment deviations by \(O_p(n^{-1/2})\) uniformly over the trimmed search range; a segment that contains the break is the sum of one partial sum from each regime, each bounded within its own regime. A Taylor expansion around the positive-definite moment \(M\), with a remainder that is \(O_p(n^{-1})\) uniformly, then gives \(\sup_s\sqrt n\,Q_n(s)=O_p(1)\), without requiring the two regimes to share their law. When the series is i.i.d. throughout, P2$'$ also gives the limit process, and the remarks on \(\sigma^2(M)\) below refer to that case.
When \(\sigma^2(M)>0\) the supremum of the limiting continuous Gaussian
process is finite almost surely; when \(\sigma^2(M)=0\) the limit is the
zero process and \(\sup_s\sqrt n\,Q_n(s)=o_p(1)\). The degenerate case is
not confined to independence: by the boundary remark of \ref{sec:A6} it covers
every \(M\) with \(I(M)=0\), including the dependent moment-matched laws of
\ref{sec:A8}(ii). Whether \(\sigma^2(M)=0\) can occur with \(I(M)>0\) we leave open;
the argument does not need it. At the independence boundary the
Corollary of \ref{sec:A12} sharpens the bound to \(\sup_s n\,Q_n(s)=O_p(1)\). Both
halves of (b) therefore hold uniformly in \(s\). We do not prove that the
P1-calibrated permutation threshold satisfies this condition; its level
is exact by P1, and its behaviour under alternatives is documented
empirically (Section 6.5).

\textbf{(c) Localisation.} Under the
additional identifiability condition
\begin{equation}\text{(A3):}\quad s\mapsto w(s)\,q(s),\ \ w(s)=\sqrt{s(1-s)},\ \text{has a unique maximiser at } s=\theta,\end{equation}
the argmax theorem (uniform convergence of continuous criteria to a
continuous limit with a unique maximiser on a compact set) yields
\(\hat\tau/n\to\theta\) a.s. for the estimator
\(\hat\tau=\arg\max_t w(t/n)Q_n(t)\).\qed

The estimator deployed in the
experiments maximises the \emph{permutation-studentised} curve instead;
since studentisation divides by a data-driven positive function,
consistency does not transfer automatically, and we do not claim it for
the studentised estimator. \emph{On (A3).} This is a population
condition checkable per model, not a consequence of the other
assumptions; we verify it numerically for the S2 family (population
moments from \(2\times 10^5\) samples; unique maximum at \(\theta\)
within grid resolution) [NV-1]. We claim no rate; the empirical
relative error lies between 3.1\% and 5.3\% over
\(n\in\{300,600,1200,2400\}\) under data-driven studentisation, without
a monotone trend (Section~\ref{sec:B3}, Figure~\ref{fig:B2} right), and
sharper localisation theory is left open.

\subsection{Proposition P4 (identifiability and the \(D\to\infty\) connection)}\label{sec:A8}

\begin{enumerate}
\renewcommand{\labelenumi}{(\roman{enumi})}
\setlength{\itemsep}{0pt}\setlength{\parskip}{0pt}
\item
  Independence \(\Rightarrow\) \(I=0\): Corollary in \ref{sec:A4}.
\item
  The
  converse fails at finite \(D\), and a moment-matched dependent law can
  be written down. \(I\) depends on the joint law only through
  \(M_{XY}\) (a \(D^2\times D^2\) moment), so any two joint laws with
  equal \(M_{XY}\) are indistinguishable. Write
  \(K_X(u)=\varphi_X(u)\varphi_X(u)^{\top}\) and let
  \(\mathcal{S}_X\subset L^2([0,1])\) be the span of its entries, so
  \(\dim\mathcal{S}_X\le D^2\); unit norm gives
  \(\mathrm{tr}\,K_X(u)=1\), so the constant function lies in
  \(\mathcal{S}_X\). Since \(L^2([0,1])\) is infinite-dimensional and
  \(\mathcal{S}_X\) is not, the orthogonal complement contains a
  non-zero bounded continuous \(p_X\), normalised to
  \(\lVert p_X\rVert_\infty=1\); and \(\int_0^1p_X=0\), because the
  constants lie in \(\mathcal{S}_X\). Take \(p_Y\) likewise from the
  \(Y\) block. For any \(0<\varepsilon\le 1\) set
  \(c(u,v)=1+\varepsilon\,p_X(u)p_Y(v)\). Then \(c\ge 0\) and
  \(c\not\equiv 1\), both margins are uniform because
  \(\int p_X=\int p_Y=0\), and
  \begin{equation}\bigl[M_{XY}[c]-M_X\otimes M_Y\bigr]_{(ab),(cd)}=\varepsilon\Bigl(\int p_X K_{X,ac}\Bigr)\Bigl(\int p_Y K_{Y,bd}\Bigr)=0,\end{equation}
  so \(M_{XY}[c]=M_X\otimes M_Y\) exactly and \(I=0\) by additivity of
  the von Neumann entropy under tensor products. Hence finite-\(D\) DOMI
  is a \emph{pseudo}-measure of dependence. The construction uses nothing about the
  feature map beyond unit norm and finitely many components, so it
  applies to the deployed map as it stands. This is a statement about
  the population functional only. In the copula-shape scenario of
  Section 6.4 the two regimes have population DOMI \(0.194\) and
  \(0.206\) at \(D=8\), so the change is visible to the finite-\(D\) moment and is lost to estimation variance rather than to non-identifiability.
\item
  \emph{Almost-sure detection at
  \(D=4\) per block, under randomly drawn frequencies.} Let \(c\) be any
  bounded copula density with uniform margins, \(c\not\equiv 1\), and
  build the features in the \emph{paired} form
  \(\varphi(u)=\sqrt{2/D}\,[\cos(\omega_1u),\sin(\omega_1u),\dots]\),
  whose norm is constant without renormalisation. Draw
  \(\omega_1^X,\omega_2^X\) i.i.d. from a distribution absolutely
  continuous with respect to Lebesgue measure, and
  \(\omega_1^Y,\omega_2^Y\) independently likewise. Then \(I>0\) almost
  surely, already at \(D=4\) per block.
\end{enumerate}

\begin{proof}
Put \(\delta=c-1\), extended by zero outside \([0,1]^2\):
bounded, compactly supported and not identically zero. By the
Paley--Wiener theorem its Fourier transform \(\hat\delta\) extends to an
entire function on \(\mathbb{C}^2\), so its restriction to
\(\mathbb{R}^2\) is real-analytic, and it is not identically zero
because the Fourier transform is injective on \(L^1\). The zero set of a
non-trivial real-analytic function on \(\mathbb{R}^2\) has Lebesgue
measure zero. Write \(\xi=\omega_1^X\pm\omega_2^X\) and
\(\eta=\omega_1^Y\pm\omega_2^Y\); each is a convolution of Lebesgue
densities and the two blocks are drawn independently, so \((\xi,\eta)\)
has an absolutely continuous law on \(\mathbb{R}^2\) --- this is where
absolute continuity is needed rather than mere non-atomicity, since a
singular continuous law could put all its mass on a Lebesgue-null set.
Hence \(\hat\delta(\xi,\eta)\neq0\) almost surely. Because \(c\) has
uniform margins, \(M_X\otimes M_Y=\iint K_X(u)\otimes K_Y(v)\,du\,dv\),
so entrywise
\begin{equation}\bigl[M_{XY}[c]-M_X\otimes M_Y\bigr]_{(ab),(cd)}=\iint\delta(u,v)\,K_X(u)_{ac}K_Y(v)_{bd}\,du\,dv,\end{equation}
and in the paired form the product-to-sum identities express \(\cos(\xi u)\) and \(\sin(\xi u)\) as fixed linear combinations of entries of \(K_X\), and likewise \(\cos(\eta v)\) and \(\sin(\eta v)\) of entries of \(K_Y\). The four integrals of
\(\delta\) against \(\{\cos,\sin\}(\xi u)\{\cos,\sin\}(\eta v)\) are
therefore fixed linear combinations of the entries above. If every entry
vanished, all four would vanish and \(\hat\delta(\xi,\eta)=0\),
contradicting the previous step. Hence \(M_{XY}[c]\neq M_X\otimes M_Y\)
almost surely and, by Klein's inequality (\ref{pre:5}), \(I>0\) almost surely.
\end{proof}

\begin{remark}[the relation to (ii)]
The counterexample of (ii) is built against a fixed frequency set, whereas (iii) draws the frequencies independently of a fixed dependence; the two do not conflict.
\end{remark}

\begin{enumerate}
\renewcommand{\labelenumi}{(\roman{enumi})}
\setcounter{enumi}{3}
\setlength{\itemsep}{0pt}\setlength{\parskip}{0pt}
\item
  \emph{Existence of the \(D\to\infty\) limit, conditional on a uniform
  spectral-decay hypothesis.} Let \(\nu\) be the frequency law,
  \(k(t)=\mathbb{E}_\nu[\cos(\omega t)]\), and \(T_k\) the integral
  operator of the Mercer kernel \(k(u-u')\) on \(L^2([0,1])\), of trace
  one since \(k(0)=1\); for \(r=D/2\) frequencies let
  \(K_r(u,u')=r^{-1}\sum_{i\le r}\cos(\omega_i(u-u'))\). \textbf{Step 1
  (exact eigenvalue identity).} Let \(A\) map a coefficient vector to
  the corresponding trigonometric polynomial. Then
  \(A^{\ast}A=r\,M_X(r)\) while \(AA^{\ast}\) is the integral operator
  of \(rK_r\), and the non-zero spectra of \(A^{\ast}A\) and
  \(AA^{\ast}\) coincide with multiplicity, so \(S(M_X(r))=S(T_{K_r})\);
  the same argument applied to the tensor features gives the joint case,
  with the copula density entering through the reference measure
  \(c(u,v)\,du\,dv\) rather than through the kernel. \textbf{Step 2
  (almost-sure eigenvalue convergence).} The uniform law of large
  numbers for random features gives \(K_r\to k\) uniformly almost
  surely, hence convergence in operator norm; for compact self-adjoint
  operators \(|\lambda_j(A)-\lambda_j(B)|\le\lVert A-B\rVert\), so each
  eigenvalue converges, and since every operator here is positive with
  trace exactly one, a discrete Scheffé argument upgrades this to
  \(\sum_j|\lambda_j(T_{K_r})-\lambda_j(T_k)|\to0\) almost surely.
  \textbf{Step 3 (entropy convergence).} \(\ell^1\) convergence of the
  spectra does not by itself give entropy convergence: on an infinite
  alphabet the entropy is only lower semicontinuous, and mass can
  migrate into an ever longer, thinner tail. We therefore assume
  \textbf{(H)}: there are \(C<\infty\) and \(\rho\in(0,1)\), independent
  of \(r\), with \(\lambda_j(T_{K_r})\le C\rho^{j}\) for all \(r,j\),
  and likewise for the \(Y\)-block and joint operators. Under (H) the
  tail sum of \(h(x)=-x\log x\), which is increasing on \((0,e^{-1}]\),
  is uniformly small beyond a fixed index while the head converges by
  Step 2; hence \(S(T_{K_r})\to S(T_k)\) almost surely.
  \textbf{Conclusion.} Under (H),
  \(\lim_{D\to\infty}I_D(c)=S(T_{k_X})+S(T_{k_Y})-S(T_{XY,\infty})\)
  almost surely, where \(T_{XY,\infty}\) acts on
  \(L^2([0,1]^2,c\,du\,dv)\) with kernel \(k_X(u-u')k_Y(v-v')\). Under
  independence the reference measure factorises,
  \(T_{XY,\infty}=T_{k_X}\otimes T_{k_Y}\), and additivity of the
  entropy under tensor products returns \(\lim_D I_D=0\), in agreement
  with the finite-\(D\) identity of the Corollary in \ref{sec:A4}.\qed
\end{enumerate}

\emph{Status of (H).} It is an additional sufficient assumption, not established here for the deployed Gaussian construction.
Proving (H), or a weaker uniform-integrability condition sufficient for
entropy convergence, is left open, as is a quantitative lower bound
relating the finite-\(D\) functional to the HSIC norm.

\begin{remark}[which feature map (iii) and (iv) cover]
Parts (iii) and (iv) are proved for the paired map, whose norm is constant. The deployed map divides pointwise by a norm that is not constant --- for the deployed draws at \(D=8\) and the median-heuristic bandwidth of rank inputs it ranges from about \(0.91\) to \(1.25\) on a 20{,}000-point grid over \([0,1]\), a grid evaluation rather than a certified bound [NV-3] --- so those proofs do not transfer, whereas (ii) and Lemma 3 hold for both. For the deployed map the conclusion of (iii) is checked numerically: on the Frank copula at \(D=4\) the largest entry of \(M_{XY}-M_X\otimes M_Y\) is \(1.5\times10^{-2}\), against a value of order \(10^{-16}\) under independence [NV-2].
\end{remark}

\begin{remark}[sample-level convergence to the Gram form]
Fix a
segment of \(m\) pseudo-observations and the bandwidth \(\gamma\), and
let the frequency--phase pairs \((w_i,b_i)\) be an i.i.d. sequence with
\(w_i\sim N(0,2\gamma)\) and \(b_i\sim U(0,2\pi)\). For the deployed map,
\(f(u)^{\top}f(u')=D^{-1}\sum_{i\le D}2\cos(w_iu+b_i)\cos(w_iu'+b_i)\to\mathbb E\cos(w(u-u'))=e^{-\gamma(u-u')^2}\)
almost surely, at each of the finitely many pairs of the segment, by the
strong law of large numbers; in particular \(\lVert f(u)\rVert^2\to1\).
So the normalised Gram matrices \(K_X/m\), \(K_Y/m\) and
\((K_X\circ K_Y)/m\) converge entrywise to their Gaussian-kernel
counterparts, and continuity of the von Neumann entropy on \(m\times m\)
density matrices gives convergence of the segment DOMI to its Gram form.
This concerns a fixed sample and needs neither (H) nor the paired form;
it says nothing about the population functional, which is the subject of
(iv).
\end{remark}

\begin{remark}[block dimension]
Parts (iii) and (iv) are proved for scalar blocks, the case of every application in Section 6. The argument of (ii) needs only that the entries of \(K_X\) span a finite-dimensional subspace containing the constants and holds for any block dimension, so the vector-valued blocks of Section~\ref{sec:B6} lie outside (iii) and (iv) only.
\end{remark}

\emph{Status of \ref{sec:A8}.} Parts (i)--(iii) are proved, (iii) for the paired feature map, and part (iv) conditionally on (H). For the deployed map, what the main text requires of (iii), that no feature dimension in the deployed range be blind to a dependent alternative, is measured in Section~\ref{sec:B3} [NV-5].

\subsection{Local geometry (second-order expansion used in Section 3.4)}\label{sec:A9}

Write \(\delta=\varepsilon\Delta\) with \(\lVert\Delta\rVert=1\) and
\(\varepsilon\downarrow 0\), all asymptotics below being in
\(\varepsilon\), and assume \(\rho\succ 0\) with
\(\varepsilon<\lambda_{\min}(\rho)\) so that \(\rho+t\delta\succ 0\) for
every \(t\in[0,1]\). Write \(\phi(t)=S(\rho+t\delta)\) for symmetric
\(\delta\) with \(\mathrm{tr}\,\delta=0\), so that
\(\mathrm{QJSD}(\rho,\rho+\delta)=\phi(\tfrac12)-\tfrac12\phi(0)-\tfrac12\phi(1)\).
Taylor expansion with second-order remainder,
\(\phi(t)=\phi(0)+t\phi'(0)+\tfrac12t^2\phi''(0)+o(t^2\lVert\delta\rVert^2)\),
gives
\begin{equation}\mathrm{QJSD}=\Bigl[\phi(0)+\tfrac12\phi'(0)+\tfrac18\phi''(0)\Bigr]-\tfrac12\phi(0)-\tfrac12\Bigl[\phi(0)+\phi'(0)+\tfrac12\phi''(0)\Bigr]+o(\lVert\delta\rVert^2)=-\tfrac18\phi''(0)+o(\lVert\delta\rVert^2),\end{equation}
and by (\ref{pre:7}),
\begin{equation}\mathrm{QJSD}(\rho,\rho+\delta)=\frac18\sum_{i,j}\frac{|\delta_{ij}|^2}{L(\lambda_i,\lambda_j)}+o(\lVert\delta\rVert^2),\end{equation}
where \(L\) is the logarithmic mean. Since \(L(a,b)\le\frac{a+b}{2}\)
(the arithmetic--logarithmic mean inequality),
\begin{equation}\mathrm{QJSD}\ \ge\ \frac14\sum_{i,j}\frac{|\delta_{ij}|^2}{\lambda_i+\lambda_j}+o(\lVert\delta\rVert^2),\end{equation}
that is, the QJSD form is bounded below by one half of the squared Bures
distance, whose kernel is \(1/(2(\lambda_i+\lambda_j))\) in standard
references. The inequality is a statement about the quadratic forms,
i.e.~it holds modulo the \(o(\lVert\delta\rVert^2)\) remainder. The overall constant is not the object
of the comparison: a divergence and a squared distance are not
commensurable. Section 3.4 uses only the \emph{shape} of the kernel. The
three functionals weight the direction \((i,j)\) by, respectively, the
reciprocal logarithmic mean (QJSD), the reciprocal arithmetic mean
(Bures/quantum Fisher), and a constant (Frobenius). Since
\(L(a,b)\big/\tfrac{a+b}{2}\to 0\) as \(b/a\to 0\), the entropic
weighting amplifies highly asymmetric eigenvalue pairs --- precisely the
low-eigenvalue directions --- more sharply than the arithmetic-mean
form, and both amplify them relative to the flat Frobenius form. The gain is modest: the kernel ratio \(\bigl(1/L(a,b)\bigr)\big/\bigl(2/(a+b)\bigr)\) grows only like \(\tfrac12\log(a/b)\) (Figure~\ref{fig:B3}).

\emph{Hypotheses and scope.} (\ref{pre:7}) requires \(\rho\succ 0\), and the expansion is second-order and therefore local. In the deployed setting \(\rho_{XY}\) is \(D^2\times D^2\) with rank at most the segment length and a rapidly decaying spectrum: at \(D=8\) several of the \(64\) eigenvalues lie below \(10^{-12}\), and on segments shorter than \(64\) some vanish exactly, so that \(\rho\succ0\) fails literally, though the ordering of the three kernels is unaffected; the kernel \(1/L(\lambda_i,\lambda_j)\) then reaches magnitudes of order \(10^{12}\) and beyond, and the radius within which the expansion is accurate collapses accordingly. The empirical comparison is Section~\ref{sec:B2}.

\subsection{Lemma 4 (elementary properties of the segment DOMI)}\label{sec:A10}

\begin{statement}
Let \(\rho_X,\rho_Y,\rho_{XY}\) be the segment
operators of Section 4.1, built from unit-norm features, and write
\(\hat I=S(\rho_X)+S(\rho_Y)-S(\rho_{XY})\) for the segment quantity,
\(I\) for its population counterpart. Then

\begin{enumerate}
\renewcommand{\labelenumi}{(\roman{enumi})}
\item
  \emph{(Range)}
  \(0\le \hat I\le\min\bigl(S(\rho_X),S(\rho_Y)\bigr)\le\log D\), and
  the same bounds hold for \(I\);
\item
  \emph{(Invariance)} \(\hat I\) is unchanged when either block is
  transformed by a strictly increasing map of its components, and when
  the feature map of either block is composed with a fixed orthogonal
  matrix; for the Gram form, also when a scalar block without ties is replaced by its negative, since its ranks are then reflected, \(u\mapsto 1-u\), and both the Gaussian kernel on ranks and the median-heuristic bandwidth depend on \(|u-u'|\) only (for the random-feature form this holds in distribution over the frequency draw);
\item
  \emph{(Consistency)} under (A0)--(A1) of \ref{sec:A6}, \(\hat I\to I\) almost
  surely;
\item
  \emph{(Bias)} under (A0), (A2) and the i.i.d. case of (A1), with \(D\)
  fixed, \(\mathbb{E}[\hat I]-I=c/m+o(1/m)\) with
  \begin{equation}c=\beta(M_{XY})-\beta(M_X)-\beta(M_Y),\qquad \beta(M)=\tfrac12\sum_{i,j}\frac{\mathrm{Var}\bigl[(U^{\top}vv^{\top}U)_{ij}\bigr]}{L(\lambda_i,\lambda_j)}\ \ge\ 0,\end{equation}
  where \(M=\mathbb{E}[vv^{\top}]=U\,\mathrm{diag}(\lambda)\,U^{\top}\)
  is the relevant population moment, \(v\) the corresponding per-sample
  feature vector, and \(L\) the logarithmic mean of (\ref{pre:7}).
\end{enumerate}
\end{statement}

\begin{proof}
\textbf{(i) Lower bound.} By \ref{sec:A4},
\(I=S(\rho_{XY}\Vert\rho_X\otimes\rho_Y)\), and
\(\mathrm{supp}\,\rho_{XY}\subseteq\mathrm{supp}(\rho_X\otimes\rho_Y)\)
because every \(\psi_t=\varphi_{X,t}\otimes\varphi_{Y,t}\) lies in the
tensor product of the marginal supports. Klein's inequality (\ref{pre:5})
gives \(I\ge 0\), with equality iff \(\rho_{XY}=\rho_X\otimes\rho_Y\).

\textbf{Upper bound.} \(\rho_{XY}=m^{-1}\sum_t\psi_t\psi_t^{\top}\) is a
convex combination of \emph{product} pure states, i.e.~a separable
state. Adjoin a classical register \(R\) with orthonormal basis
\(\{|t\rangle\}\) and set
\begin{equation}\sigma_{XYR}=\frac1m\sum_t\bigl(\varphi_{X,t}\varphi_{X,t}^{\top}\bigr)\otimes\bigl(\varphi_{Y,t}\varphi_{Y,t}^{\top}\bigr)\otimes|t\rangle\langle t|,\end{equation}
so that \(\mathrm{Tr}_R\,\sigma_{XYR}=\rho_{XY}\). Because \(\sigma\) is
block diagonal in \(R\) and each block is a \emph{pure} product state,
\(S(\sigma_{XYR})=\log m\) and \(S(\sigma_{YR})=\log m\), whence the
conditional entropy
\(S(X|YR)_{\sigma}=S(\sigma_{XYR})-S(\sigma_{YR})=0\). Strong
subadditivity, in the form that discarding part of the conditioning
system cannot decrease conditional entropy [26], gives
\begin{equation}S(X|Y)_{\rho}\ \ge\ S(X|YR)_{\sigma}=0,\qquad\text{i.e.}\qquad S(\rho_{XY})\ \ge\ S(\rho_Y),\end{equation}
and symmetrically \(S(\rho_{XY})\ge S(\rho_X)\). (The step
\(\sigma_{YR}=\mathrm{Tr}_X\sigma_{XYR}=\frac1m\sum_t\lVert\varphi_{X,t}\rVert^2(\varphi_{Y,t}\varphi_{Y,t}^{\top})\otimes|t\rangle\langle t|\)
uses \(\lVert\varphi_{X,t}\rVert=1\); without it the block weights would
not be uniform and \(S(\sigma_{YR})<\log m\), which would invalidate the
argument. Note also that the register \(R\) makes the blocks mutually
orthogonal, so no orthogonality among the \(\varphi_{Y,t}\) themselves
is required.) Therefore
\(I\le S(\rho_X)+S(\rho_Y)-\max\bigl(S(\rho_X),S(\rho_Y)\bigr)=\min\bigl(S(\rho_X),S(\rho_Y)\bigr)\).
Finally \(\rho_X\) is a \(D\times D\) density operator, so
\(S(\rho_X)\le\log D\) by (\ref{pre:1}) and the maximum of \(H\) on the
simplex. The same argument applies directly to the population state
\(M_{XY}=\mathbb{E}[\psi\psi^{\top}]\), itself a mixture of product pure
states over the population law in place of the empirical one, so the
bounds hold for the population \(I\) without invoking (iii). \emph{Convention.} Matrices here are real symmetric; the
real symmetric case embeds isometrically into the complex Hermitian one
with all entropies unchanged, so the standard statement of strong
subadditivity applies.
\end{proof}

\begin{remark}
The finiteness of the bound is a structural difference
from classical mutual information, which diverges as the dependence
approaches a deterministic relation. The operative ceiling is not
\(\log D\) but \(\min(S(\rho_X),S(\rho_Y))\), and for the rank-plus-RFF
map deployed here the marginal entropy lies between about \(0.8\) and
\(1.4\) nats for \(D\) from 4 to 64 and does not grow with \(D\) (for
one feature draw \(S(\rho_X)\) is \(1.22\), \(1.13\), \(1.00\), \(1.10\)
and \(1.27\) at \(D=4,8,16,32,64\)), so the binding
constraint is roughly \(1.1\) nats rather than \(\log D=2.08\) at
\(D=8\), and raising \(D\) does not relax it. The bound bears on Section
6.4: as dependence strengthens, \(\hat I\) saturates rather than grows,
so a fixed \emph{difference} between two strongly dependent regimes is
compressed while the estimator's variance is not.
\end{remark}

\textbf{(ii)} Componentwise ranks are invariant under strictly
increasing transformations of each component (a strictly increasing map preserves the entire order relation, including equalities, and moves no observation in time, so ranks with ties broken in time order, as implemented, are preserved; with random tie-breaking the statement holds conditionally on the tie-break), so \(U_X,U_Y\) and hence all features and all three
operators are unchanged. For the second claim, replacing \(\varphi_X\)
by \(Q\varphi_X\) with \(Q^{\top}Q=I\) sends
\(\rho_X\mapsto Q\rho_XQ^{\top}\) and
\(\rho_{XY}\mapsto(Q\otimes I)\rho_{XY}(Q\otimes I)^{\top}\); both are
orthogonal conjugations, which preserve spectra, and by (\ref{pre:1}) the von
Neumann entropy depends on the operator only through its spectrum. Hence
every term of \(I\) is unchanged.\qed

\textbf{(iii)} By the law of large numbers applied to the bounded i.i.d.
(or ergodic) summands of \ref{sec:A6} Step 1, \(\hat M\to M\) a.s. for each of
the three moments. The function \(S\) is continuous on the whole compact
set of positive semi-definite matrices of unit trace --- not merely on
the positive-definite cone --- because \(x\mapsto-x\log x\) extends
continuously to \(x=0\); positive definiteness is therefore not needed
here, which matters because \(\hat M_{XY}\) is singular whenever
\(m<D^2\). Composition of a.s. convergence with a continuous map gives
\(\hat I\to I\) a.s.\qed

\textbf{(iv)} Fix one moment \(M\succ 0\) and write \(\Delta=\hat M-M\).
Taking expectations of a Taylor expansion requires more than the
in-probability bound of \ref{sec:A6} Step 3, so we argue on an event. Let
\(A=\{\lambda_{\min}(\hat M)\ge\tfrac12\lambda_{\min}(M)\}\). On
\(A^{c}\) we use only \(0\le S(\hat M)\le\log\dim\) --- valid because
\(\mathrm{tr}\,\hat M=1\) exactly by \ref{sec:A4} --- together with the matrix
Bernstein inequality for bounded i.i.d. summands, which gives
\(P(A^{c})\le 2\dim\cdot e^{-c_0m}\) for some
\(c_0=c_0(\lambda_{\min}(M))>0\); hence
\(\mathbb{E}\bigl[|S(\hat M)-S(M)|\mathbf 1_{A^{c}}\bigr]=o(m^{-k})\)
for every \(k\). On \(A\) the spectrum of \(\hat M\) remains in a
compact subset of \((0,\infty)\), so the third derivative of \(S\) is
bounded there and a third-order Taylor expansion has remainder
\(O(\lVert\Delta\rVert^{3})\) with
\(\mathbb{E}\lVert\Delta\rVert^{3}\le(\mathbb{E}\lVert\Delta\rVert^{4})^{3/4}=O(m^{-3/2})=o(1/m)\),
where \(\mathbb{E}\lVert\Delta\rVert^{4}=O(m^{-2})\) because each entry
\(\Delta_{ij}\) is an average of \(m\) bounded i.i.d. mean-zero scalars,
so \(\mathbb{E}\Delta_{ij}^{4}=O(m^{-2})\), and the number of entries is
fixed. The first- and second-order terms below carry the indicator
\(\mathbf 1_A\); removing it changes them by \(O(P(A^{c}))\), which is
exponentially small and absorbed in the remainder. Combining the two
events,
\begin{equation}\mathbb{E}\bigl[S(\hat M)\bigr]=S(M)+\mathbb{E}\bigl[DS[\Delta]\bigr]+\tfrac12\,\mathbb{E}\bigl[D^2S[\Delta,\Delta]\bigr]+o(1/m).\end{equation}
The first-order term vanishes because \(\mathbb{E}\Delta=0\) and \(DS\)
is linear. By (\ref{pre:7}), in the eigenbasis of \(M\),
\(D^2S[\Delta,\Delta]=-\sum_{i,j}|\Delta_{ij}|^2/L(\lambda_i,\lambda_j)\),
and
\(\mathbb{E}|\Delta_{ij}|^2=\mathrm{Var}\bigl[(U^{\top}vv^{\top}U)_{ij}\bigr]/m\)
for an average of \(m\) i.i.d. terms. Hence
\(\mathbb{E}[S(\hat M)]-S(M)=-\beta(M)/m+o(1/m)\) with \(\beta\ge 0\):
every plug-in entropy is biased \emph{downward}, by an amount governed
by the inverse logarithmic mean of eigenvalue pairs. Applying this to
the three moments, with the signs with which their entropies enter
\(\hat I\), gives (iv).\qed

\emph{Sign of \(c\).} At the boundary \(c\ge 0\) is proved. Under (A0) with i.i.d. pairs, \(\mathbb E\,C=0\) exactly and the expected quadratic term of \(m\hat I\) equals \(\sum_j w_j\) for every \(m\). The remainder is split on the event \(\{\lVert C\rVert\le\lambda_{\min}(\bar\rho)/2\}\) of Step 3 of \ref{sec:A12}: on it the Taylor remainder satisfies \(\mathbb E\,m|R_m|\le c'\,m\,\mathbb E\lVert C\rVert^3=O(m^{-1/2})\); off it the entropies are bounded by \(\log D^2\), and the event fails with exponentially small probability by the matrix Bernstein inequality. Hence \(\mathbb E[m\hat I]\to\sum_j w_j=\mathbb E\,\Phi(G)\ge 0\) by \ref{sec:A12}; strict positivity holds exactly when the boundary limit is non-degenerate. The closed form is positive at every independence configuration evaluated, \(c=3.07\) at \(D=8\), against a measured \(m\,\mathbb E[\hat I]\) of \(3.12\) at \(m=200\) and at \(m=800\) [NV-4]; the agreement is empirical, since part (iv) is asymptotic in \(m\) for fixed \(D\). Away from the boundary the sign is not universal: along a Gaussian-correlation path on rank features at \(D=8\) the closed form gives \(+2.21\) at \(r=0.7\), \(+0.69\) at \(r=0.95\), \(-0.02\) at \(r=0.99\) and \(-0.38\) for \(Y=X\), so the independence-based correction of Section 4.2 over-corrects under strong dependence. \emph{On (A0).} Parts (iii) and (iv) are stated for the oracle features of (A0); under global empirical ranks neither the almost-sure limit of (iii) nor the constant \(c\) of (iv) is established for the statistic as computed.

\subsection{Proposition P5 (finite-sample exactness under block permutation)}\label{sec:A11}

\begin{setting}
Partition \(\{1,\dots,n\}\) into
\(B=\lfloor n/b\rfloor\) consecutive blocks of length \(b\) (a trailing
partial block, if any, is held in place). Let \(\mathcal{G}_b\cong S_B\)
denote the subgroup of permutations of \(\{1,\dots,n\}\) induced by
permuting the \emph{order} of these blocks, acting simultaneously on
\(X\) and \(Y\). Every observation appears exactly once in a
\(\mathcal{G}_b\)-image, so this is a genuine permutation of the sample,
not a resample. Draw \(\sigma_1,\dots,\sigma_K\) i.i.d. uniform on
\(\mathcal{G}_b\), independent of \(W\), and define \(T_k\) and \(p\)
exactly as in \ref{sec:A5}.
\end{setting}

\begin{assumption}[block exchangeability]
Under the no-change
hypothesis, \(W\circ\sigma\overset{d}{=}W\) for every
\(\sigma\in\mathcal{G}_b\). Writing \(Z_k\) for the sub-series of pairs
in block \(k\) and \(Z_{\mathrm{rem}}\) for the trailing indices that
are held in place, this is
\begin{equation}(Z_{\pi(1)},\dots,Z_{\pi(B)},Z_{\mathrm{rem}})\ \overset{d}{=}\ (Z_1,\dots,Z_B,Z_{\mathrm{rem}})\qquad\text{for every }\pi\in S_B.\end{equation}
\end{assumption}

The trailing block must be carried in the display: exchangeability of
\((Z_1,\dots,Z_B)\) \emph{alone} is not enough, because the statistic
also sees \(Z_{\mathrm{rem}}\), and a null in which \(Z_{\mathrm{rem}}\)
is dependent on one of the permuted blocks invalidates the argument. The
issue does not arise when \(b\) divides \(n\), or when the trailing
observations are dropped; in the experiments of Table~\ref{tab:B5} and Section~\ref{sec:B17}, \(n=600\)
with \(b\in\{20,50\}\) leaves no trailing block.

\begin{claim}
\(P(p\le\alpha)\le\alpha\) for every \(\alpha\in(0,1)\),
with equality up to the usual discreteness if the \(T_k\) are a.s.
distinct.
\end{claim}

\begin{proof}
The proof of \ref{sec:A5} uses only two things: that the null law
of the data is invariant under the group from which the replicas are
drawn, and that the replica permutations are i.i.d. uniform on that
group and independent of the data. Both hold here with \(\mathcal{G}_b\)
in place of \(S_n\): the assumption above is precisely the required
invariance, and \(\mathcal{G}_b\) is a group, so Step 1's re-indexing
argument --- including the side of composition --- goes through verbatim
with \(\sigma_j^{-1}\circ\sigma_k\in\mathcal{G}_b\). Step 2 is a
statement about exchangeable real vectors and is untouched.
\end{proof}

\emph{Plausibility of the assumption.} It holds when the blocks are i.i.d., and more generally when they are exchangeable jointly with any remainder held fixed. Short-range or finite-range (\(\kappa\)-dependent) serial dependence does not by itself make adjacent blocks exchangeable, and we give no bound on the departure. A larger \(b\) may reduce that departure but coarsens the permutation group, which has \(B!\) elements for \(B\) blocks; the candidate grid does not depend on \(b\). Table~\ref{tab:B5} and Section~\ref{sec:B17} quantify the trade-off. The group \(\mathcal{G}_b\) is a strict subgroup of \(S_n\), so P5 is weaker than P1, but it does not assume exchangeability of individual time points.

\begin{remark}[exactness under the larger group]
The block
permutations of \(\{1,\dots,n\}\) form a subgroup of the full
permutation group, and full exchangeability of the pairs implies
invariance under every subgroup. P5's calibration is therefore exact
under pair exchangeability as well: choosing the block calibration when
the pair calibration would have sufficed costs power and a coarser null, not level, for a block length fixed in advance. The statement does not cover a scheme or block length chosen from the same data, as the protocol of Section 4.3 does.
\end{remark}

\subsection{Proposition P6 (degenerate limit law at the independence boundary; oracle margins)}\label{sec:A12}

\begin{setting}
Pairs \((X_t,Y_t)\), \(t=1,\dots,m\), i.i.d. with
\(X_t\perp Y_t\). Fix \(D\) and the unit-norm feature maps of Section
4.1, applied to the true marginal transforms (oracle margins), so that
\(\varphi_{X,t},\varphi_{Y,t}\) are i.i.d. unit vectors,
\(\psi_t=\varphi_{X,t}\otimes\varphi_{Y,t}\),
\(\hat\rho_{XY}=m^{-1}\sum_t\psi_t\psi_t^{\top}\) (unit trace exactly),
and \(\hat\rho_X=\mathrm{Tr}_Y\hat\rho_{XY}\),
\(\hat\rho_Y=\mathrm{Tr}_X\hat\rho_{XY}\) (Lemma 3). Population states:
\(\rho_X=E[\varphi_X\varphi_X^{\top}]\), \(\rho_Y\), and by independence
\(\rho_{XY}=\rho_X\otimes\rho_Y=:\bar\rho\), so the population DOMI is
\(S(\rho_X)+S(\rho_Y)-S(\bar\rho)=0\). Assume \(\rho_X\succ0\) and \(\rho_Y\succ0\).
\end{setting}

\begin{claim}
\(m\,\hat I_m\ \Rightarrow\ \tfrac12\big[q_{\bar\rho}(G)-q_{\rho_X}(\mathrm{Tr}_Y G)-q_{\rho_Y}(\mathrm{Tr}_X G)\big]\ =\ \sum_j w_j Z_j^2\)
with \(Z_j\) i.i.d. standard normal and \(w_j\ge0\), where \(G\) is the
matrix Gaussian limit of \(\sqrt m(\hat\rho_{XY}-\bar\rho)\) and
\(q_\rho(E)=\sum_{i,j}k(\lambda_i,\lambda_j)\tilde E_{ij}^2\) in the
eigenbasis of \(\rho\), \(k\) being the reciprocal of the logarithmic
mean,
\(k(\lambda_i,\lambda_j)=(\log\lambda_i-\log\lambda_j)/(\lambda_i-\lambda_j)\),
\(k(\lambda,\lambda)=1/\lambda\) --- the kernel of Section 3.4.
\end{claim}

\begin{proof}
Write \(C=\hat\rho_{XY}-\bar\rho\),
\(A=\hat\rho_X-\rho_X\), \(B=\hat\rho_Y-\rho_Y\); all three are
symmetric and traceless, and \(\mathrm{Tr}_Y C=A\),
\(\mathrm{Tr}_X C=B\) hold as sample identities by Lemma 3 and linearity
of the partial trace.

\textbf{Step 1 (the first-order term vanishes identically).} For
traceless \(E\), \(\mathrm dS(\rho)[E]=-\mathrm{tr}(E\log\rho)\). Since
\(\log\bar\rho=\log\rho_X\otimes I+I\otimes\log\rho_Y\) and
\(\mathrm{tr}\big(C(\log\rho_X\otimes I)\big)=\mathrm{tr}\big((\mathrm{Tr}_Y C)\log\rho_X\big)=\mathrm{tr}(A\log\rho_X)\),
and symmetrically in \(Y\), the first-order contribution to
\(\hat I=S(\hat\rho_X)+S(\hat\rho_Y)-S(\hat\rho_{XY})\) is
\(-\mathrm{tr}(A\log\rho_X)-\mathrm{tr}(B\log\rho_Y)+\mathrm{tr}(C\log\bar\rho)=0\):
an identity of the sample path, not an expectation. This is the
structural reason for the degeneracy noted in P2.

\textbf{Step 2 (central limit theorem).} The summands
\(\psi_t\psi_t^{\top}\) are i.i.d. and bounded (unit trace, entries in
\([-1,1]\)) with mean \(\bar\rho\), so \(\sqrt m\,C\Rightarrow G\), a
centred matrix Gaussian; by linearity of the partial trace,
\(\sqrt m\,A\Rightarrow\mathrm{Tr}_Y G\) and
\(\sqrt m\,B\Rightarrow\mathrm{Tr}_X G\) jointly with \(C\).

\textbf{Step 3 (second-order expansion).} On the event
\(\{\|C\|\le\lambda_{\min}(\bar\rho)/2\}\), whose probability tends to
one, the entropy is three times continuously differentiable along the
segment between population and empirical states, with
\(\mathrm d^2S(\rho)[E,E]=-q_\rho(E)\) by (\ref{pre:7}), whose
resolvent-integral proof assumes nothing about multiplicity --- which
matters here, since \(\bar\rho=\rho_X\otimes\rho_Y\) carries the product
spectrum \(\{\lambda_a\mu_b\}\) and its entries can coincide; Taylor's
theorem and the Step-1 cancellation give
\(\hat I=\tfrac12\big[q_{\bar\rho}(C)-q_{\rho_X}(A)-q_{\rho_Y}(B)\big]+R_m\)
with \(|R_m|\le c\,\|C\|^3=O_p(m^{-3/2})\), \(c\) depending on
\(\lambda_{\min}\). Multiplying by \(m\) and applying the continuous
mapping theorem to the continuous quadratic form yields the claim.

\textbf{Step 4 (form and sign of the weights).} A quadratic form of a
Gaussian vector is distributed as \(\sum_j w_j Z_j^2\) with \(w_j\) the
eigenvalues of the form composed with the covariance of \(G\).
Non-negativity: \(\hat I\ge0\) for every sample, by subadditivity of the
von Neumann entropy together with the exactness of the sample partial
trace (Lemma 3), so the limit is almost surely non-negative, and a
Gaussian quadratic form that is a.s. non-negative has no negative weight.
\end{proof}

\begin{remark}[non-degeneracy of the limit]
The weights are not all
zero exactly when the bias constant \(c=\lim_m m\,\mathbb E[\hat I_m]\)
of Lemma 4(iv) is strictly positive. The fourth-moment bound
\(\mathbb E\lVert\hat M-M\rVert^{4}=O(m^{-2})\) established in the proof
of that lemma makes \(m\hat I_m\) uniformly integrable, so
\(\sum_j w_j=\lim_m\mathbb E[m\hat I_m]=c\), and the weights are
non-negative. Strict positivity of \(c\) at independence is not proved
(Lemma 4(iv), Section~\ref{sec:A10}); it holds in closed form at every configuration computed here. Were \(c\) zero, the limit
above would be the point mass at zero and \(m\) would be the wrong normalisation.
\end{remark}

\begin{remark}[spectrum in practice, and what P6 does not cover]
At the default bandwidth the trailing eigenvalues of \(\rho_X\) decay to numerical zero, and the constant \(c\) of Step 3 scales with an inverse power of \(\lambda_{\min}\), so the full-spectrum Taylor bound is not quantitative at the sample sizes considered; P6 does not justify spectral truncation, and Section~\ref{sec:B16} reports the finite-sample calibration. The process level and empirical ranks are supplied by the Corollary and by P6-R below.
\end{remark}

\begin{corollary}[process-level boundary limit, and the scan statistic]
Under the assumptions of P6, fix
\(0<\varepsilon<\tfrac12\); for a split fraction
\(\lambda\in[\varepsilon,1-\varepsilon]\) let \(\hat I_L(\lambda)\) and
\(\hat I_R(\lambda)\) be the DOMI plug-ins of the two segments, and let
\(\Phi(M):=\tfrac12\big[q_{\bar\rho}(M)-q_{\rho_X}(\mathrm{Tr}_Y M)-q_{\rho_Y}(\mathrm{Tr}_X M)\big]\)
be the (positive semi-definite) limit form of P6. Then, jointly in
\(\ell^\infty[\varepsilon,1-\varepsilon]\),
\begin{equation}\big(n\,\hat I_L(\lambda),\ n\,\hat I_R(\lambda)\big)\ \Rightarrow\ \Big(\lambda^{-2}\,\Phi\big(G(\lambda)\big),\ (1-\lambda)^{-2}\,\Phi\big(G(1)-G(\lambda)\big)\Big),\end{equation}
where \(G\) is the matrix Brownian motion limit of the centred
partial-sum process of \(\psi_t\psi_t^{\top}\). Consequently every
continuous functional of this pair converges --- in particular the
supremum, over a fixed finite candidate grid or over
\([\varepsilon,1-\varepsilon]\) with continuous weights, of the
unstudentised difference statistic; the finite grid actually computed in
Section 4.6 is the special case of evaluation at finitely many \(\lambda\).
\end{corollary}

\begin{proof}
The summands are i.i.d., bounded and centred after
subtracting \(\bar\rho\), so the multivariate Donsker theorem gives
\(W_n(\lambda):=n^{-1/2}\sum_{t\le\lambda n}\big(\mathrm{vec}(\psi_t\psi_t^{\top})-\mathrm{vec}\,\bar\rho\big)\Rightarrow G\)
in the Skorokhod space, with a continuous Gaussian limit and
\(\sup_\lambda\|W_n(\lambda)\|=O_p(1)\). For \(\lambda\ge\varepsilon\)
the left-segment deviation is
\(C_L(\lambda)=n^{-1/2}W_n(\lambda)/\lambda_n\) with
\(\lambda_n=\lfloor\lambda n\rfloor/n\), so
\(\sup_{\lambda\ge\varepsilon}\|C_L(\lambda)\|=O_p(n^{-1/2})\), and
symmetrically on the right with the increment \(W_n(1)-W_n(\lambda)\).
Step 1 of P6 --- the exact first-order cancellation --- holds for every
segment separately, because Lemma 3 and the unit trace are per-segment
sample identities. Step 3's Taylor bound therefore holds uniformly: on
the event
\(\sup_\lambda\|C_L(\lambda)\|\vee\|C_R(\lambda)\|\le\lambda_{\min}(\bar\rho)/2\),
whose probability tends to one,
\begin{equation}\sup_{\lambda}\Big|n\hat I_L(\lambda)-\lambda_n^{-2}\,\Phi\big(W_n(\lambda)\big)\Big|\ \le\ n\,c\,\sup_\lambda\|C_L(\lambda)\|^3\ =\ O_p(n^{-1/2})\ \to\ 0,\end{equation}
and \(|\lambda_n^{-2}-\lambda^{-2}|=O(n^{-1})\) uniformly on
\([\varepsilon,1-\varepsilon]\). Since \(\Phi\) is a fixed continuous
quadratic form, the map from the path of \(W_n\) to the pair of
processes above is continuous with respect to the supremum norm at every
continuous path, and the continuous mapping theorem yields the joint
limit; a further application of continuous mapping gives the convergence of the scan functionals.
\end{proof}

\begin{proposition}{P6-R}[the empirical-rank version]
Let the pairs be
i.i.d. with independent components and continuous margins, fix \(D\), the
frequency draw and the bandwidth, assume \(\rho_X\succ0\) and \(\rho_Y\succ0\) as in P6, and apply the deployed features of
Section 4.1 to the pseudo-observations
\(\hat U_t=\mathrm{rank}(X_t)/(m{+}1)\),
\(\hat V_t=\mathrm{rank}(Y_t)/(m{+}1)\) of the \(m\)-sample. Then the
conclusion of P6 holds with \(G\) replaced by the matrix Gaussian
\(\tilde G=G-\mathrm{Tr}_YG\otimes\rho_Y-\rho_X\otimes\mathrm{Tr}_XG\),
whose covariance carries the classical rank corrections, and
\(\Phi(\tilde G)=\Phi(G)\): the limit law, and with it the weights, are
those of P6, \(m\,\hat I_m\Rightarrow\sum_jw_jZ_j^2\).
\end{proposition}

\begin{proof}
Steps 1, 3 and 4 of P6 are statements about the sample and
the fixed reference \(\bar\rho\) --- the partial-trace identity, the
unit trace, the entropy expansion, and the non-negativity of \(\hat I\)
--- and hold for arbitrary inputs in \((0,1)\), in particular for
pseudo-observations; only Step 2, the law of \(\sqrt m\,C\), changes.
Each entry of \(\hat\rho_{XY}\) is \(\int H\,\mathrm d\hat C_m\) with
\(H(u,v)=K_X(u)_{ac}K_Y(v)_{bd}\), where \(\hat C_m\) is the empirical
measure of the pseudo-observations, whose distribution function is the
empirical copula. For the deployed map
\(K_X(u)=f(u)f(u)^{\top}/\lVert f(u)\rVert^2\) with
\(f(u)=\sqrt{2/D}\cos(uW+b)\), so \(H\) is a ratio of trigonometric
polynomials. It is real-analytic on \([0,1]^2\) provided
\(\min_{u\in[0,1]}\lVert f(u)\rVert>0\), which holds almost surely over
the draw of \((W,b)\) for \(D\ge2\) (two cosines with independent uniform phases have a common zero with probability zero, and \([0,1]\) is compact).

Integrating by parts turns
\(\sqrt m\int H\,\mathrm d(\hat C_m-C)\) into the integral of the
empirical copula process against the fixed finite signed measure
\(\mathrm dH\) --- a bounded linear functional with respect to the
supremum norm --- and the empirical copula process converges weakly at
the independence copula, which satisfies the non-restrictive smoothness
condition of [34] trivially; the Cramér--Wold device then gives
\(\sqrt m\,C\Rightarrow\tilde G\) jointly with its partial traces. At
the independence copula the marginal part of that limit,
\(-[v\,\mathbb B(1,u,1)+u\,\mathbb B(1,1,v)]\), with \(\mathbb B\) the Kiefer process that limits the sequential empirical process of the pairs [10], integrates against \(H\)
to \(-(\mathrm{Tr}_YG\otimes\rho_Y+\rho_X\otimes\mathrm{Tr}_XG)\), where
\(\int K_Y(v)\,\mathrm dv=\rho_Y\) and \(\mathrm{Tr}_YG\) is traceless
because \(\mathrm{tr}\,K_X\equiv1\). This gives the stated form of
\(\tilde G\). \emph{The form \(\Phi\) annihilates the correction.} Let
\(P=A\otimes\rho_Y+\rho_X\otimes B\) with \(A,B\) symmetric and traceless.
In the product eigenbasis, where \(\bar\rho\) has eigenvalues
\(\lambda_a\mu_b\) and
\(k(\lambda_a\mu_b,\lambda_{a'}\mu_b)=k(\lambda_a,\lambda_{a'})/\mu_b\),
one has
\(q_{\bar\rho}(E,A\otimes\rho_Y)=q_{\rho_X}(\mathrm{Tr}_YE,A)\) for every
symmetric \(E\), and \(\mathrm{Tr}_X(A\otimes\rho_Y)=(\mathrm{tr}A)\rho_Y=0\).
Hence the bilinear form of \(\Phi\) vanishes between \(E\) and
\(A\otimes\rho_Y\). The same holds for \(\rho_X\otimes B\), and the cross
term between the two is \(-\tfrac12\mathrm{tr}A\,\mathrm{tr}B=0\). So
\(\Phi(E+P)=\Phi(E)\) for every \(E\), and
\(\Phi(\tilde G)=\Phi(G)\). The degeneracy structure is
untouched because Step 1 is pathwise, and Step 3 needs only
\(\|C\|=O_p(m^{-1/2})\), which the convergence supplies.
\end{proof}

\begin{corollary}[process level, empirical ranks]
Under the
assumptions of P6-R, let the pseudo-observations be computed \emph{once}
from the whole window of \(n\) observations and reused for every left
and right segment, as the deployed scan does (Section 4.6). Then the
joint limit of the process-level Corollary holds unchanged, with the
same oracle matrix Brownian motion \(G\). The same is true when the ranks
are recomputed within each segment.
\end{corollary}

\begin{proof}
With global ranks the summands of a segment are no longer
i.i.d., so the substitution used for P6-R does not apply directly; we
linearise instead. Let \(U_t=F_X(X_t)\), \(V_t=F_Y(Y_t)\) and let \(W_n\)
be the oracle partial-sum process of the Corollary above. The
pseudo-observations satisfy \(\sup_t|\hat U_t-U_t|=O_p(n^{-1/2})\), and
\(\sqrt n(\hat U_t-U_t)=\alpha_n^X(U_t)+O(n^{-1/2})\), where
\(\alpha_n^X\) is the full-sample uniform empirical process. Since \(K_X\)
and \(K_Y\) are \(C^2\) (proof of P6-R), a first-order Taylor expansion
with uniformly bounded second-order remainder gives
\begin{equation}\begin{aligned}\sqrt n\,\big(\hat\rho_{XY,[0,\lfloor\lambda n\rfloor)}-\rho^{\circ}_{XY,[0,\lfloor\lambda n\rfloor)}\big)&=\lambda_n^{-1}\,\frac1n\sum_{t\le\lambda n}\Big[K_X'(U_t)\alpha_n^X(U_t)\otimes K_Y(V_t)\\
&\qquad+K_X(U_t)\otimes K_Y'(V_t)\alpha_n^Y(V_t)\Big]+O_p(n^{-1/2}),\end{aligned}\end{equation} uniformly in \(\lambda\), where \(\rho^{\circ}\) is the oracle segment moment. Approximate \(\alpha_n^X\) by a step function on a mesh of width \(\delta\); asymptotic equicontinuity makes the error small uniformly. For each step, the uniform law of large numbers for partial sums of bounded i.i.d. terms applies, and independence of \(U_t\) and \(V_t\) gives \(\mathbb E[K_Y(V)]=\rho_Y\). The first sum is therefore \(\lambda\int_0^1K_X'\alpha_n^X\,\mathrm du\otimes\rho_Y+o_p(1)\) uniformly, and integrating by parts (\(\alpha_n^X(0)=\alpha_n^X(1)=0\)) gives \(\int K_X'\alpha_n^X=-\mathrm{Tr}_YW_n(1)\). Hence, uniformly on \([\varepsilon,1-\varepsilon]\), \begin{equation}\begin{aligned}\sqrt n\,C_L(\lambda)&=\lambda_n^{-1}\big[W_n(\lambda)-\lambda_nP_n\big]+o_p(1),\\
\sqrt n\,C_R(\lambda)&=(1-\lambda_n)^{-1}\big[W_n(1)-W_n(\lambda)-(1-\lambda_n)P_n\big]+o_p(1),\end{aligned}\end{equation} with the \emph{common} correction \(P_n=\mathrm{Tr}_YW_n(1)\otimes\rho_Y+\rho_X\otimes\mathrm{Tr}_XW_n(1)\). This is the image under \(H\) of the marginal part \(-\lambda[v\,\mathbb B(1,u,1)+u\,\mathbb B(1,1,v)]\) of the limit of the sequential empirical process of the global pseudo-observations at the independence copula (cf. [10,11]). \(P_n\) has the product form of the proof of P6-R, with traceless factors, so \(\Phi(\sqrt n\,C_L(\lambda))=\lambda_n^{-2}\Phi(W_n(\lambda))+o_p(1)\) uniformly, and likewise on the right. Steps 1 and 3 hold per segment, because Step 1 is a sample identity and Step 3 a Taylor bound that needs only \(\sup_\lambda\|C_L\|\vee\|C_R\|=O_p(n^{-1/2})\), which the display supplies. The rest of the proof of the Corollary applies verbatim. Within-segment ranks give the same display, with each segment's own marginal fluctuation in place of \(\mathrm{Tr}_YW_n(1)\); that correction is again of product form and drops out.
\end{proof}

A finite-sample comparison of the global-rank, within-segment-rank and oracle statistics is in Section~\ref{sec:B3} [NV-8]. The boundary theory
(pointwise, at the process level and for the grid scan) holds, by the Corollary above, for the statistic exactly as computed, for a fixed frequency draw
and bandwidth. The deployed bandwidth is the median heuristic on a fixed 400-point subsample of the global pseudo-observations, a random quantity; the extension of the limit theory to this data-selected bandwidth is not established here, and Section~\ref{sec:B16} checks the deployed calibration empirically. The weights of the limit are continuous functions of the population feature moments, and the whole-series moments of the rank features converge to them. Provided the limit law of the scan is continuous at the nominal quantile and, for the studentised version, has positive variance at every candidate, critical values obtained from the limit with plug-in weights are consistent as \(n\) and the number of simulation draws grow; with the 2000 draws used here a Monte Carlo error remains. Section~\ref{sec:B16} compares them with the exact null law. The
permutation studentisation remains outside the theory; its null law
stabilises across \(n\) in the simulation of Section~\ref{sec:B3} [NV-6].

\subsection{Proposition P7 (consistency of the exchangeability diagnostic)}\label{sec:A13}

\begin{setting}
Let \((X_t,Y_t)\) be strictly stationary and ergodic
with continuous margins, and run the diagnostic of Section 4.3 with \(L\) lags, \(K\) permutation replicas and Bonferroni level \(\alpha\), all held fixed as \(n\to\infty\), with \(K\ge 5/\alpha-1\), so that the minimal attainable permutation
p-value \(1/(K{+}1)\) does not exceed the per-statistic threshold
\(\alpha/5\) (the paper uses \(K=199\), \(\alpha=0.05\)). Suppose the
serial structure is visible to one of the five statistics: for some lag
\(1\le k\le L\), the population autocorrelation \(\rho_k\) of
\(g(F_X(X_t))\) or \(g(F_Y(Y_t))\) is non-zero for \(g(u)=u\) (levels)
or \(g(u)=(u-\tfrac12)^2\) (scale), or a lagged cross-correlation of the
transformed margins is non-zero.
\end{setting}

\begin{claim}
The probability that the protocol selects the block calibration tends to one.
\end{claim}

\begin{proof}
Write \(z_t=g(\hat U_t)-n^{-1}\sum_{s}g(\hat U_s)\) for the centred transformed ranks of the margin in question, with \(g(u)=u\) for the level statistics and \(g(u)=(u-\tfrac12)^2\) for the scale statistics, and \(z_t^{\circ}=g(F_X(X_t))-E\,g\) for its oracle counterpart; the argument applies to either \(g\). \textbf{Step 1 (the observed statistic diverges).} The
ergodic theorem gives a.s. convergence of the sample autocovariances of
the bounded sequence \(z^{\circ}\) to their population values, and the
Glivenko--Cantelli property holds for strictly stationary ergodic
sequences, so \(\max_t|z_t-z_t^{\circ}|\to0\) a.s. (\(g\) is Lipschitz
on \([0,1]\)); boundedness then gives \(\hat r_k\to\rho_k\neq0\) a.s.
for the rank-based sample autocorrelation, whence the observed
portmanteau statistic satisfies
\(Q_{\mathrm{obs}}\ge n(n+2)\hat r_k^{\,2}/(n-k)\to\infty\) a.s.

\textbf{Step 2 (the replicas stay bounded).} Conditionally on the data,
a replica applies a uniformly random permutation to the pairs, and the
rank values form a fixed bounded multiset. For fixed bounded sequences,
the mean and variance of a lag-\(k\) serial (or lagged cross-)
covariance under a uniform random permutation are \(O(n^{-1})\) --- a
direct finite-population computation, classical since Wald and Wolfowitz
[33] --- so each replica's autocorrelations are \(O_p(n^{-1/2})\)
and each replica's portmanteau statistic is \(O_p(1)\); with \(L\) and
\(K\) fixed, the maximum over the \(K\) replicas is \(O_p(1)\).

\textbf{Step 3.} Hence the observed statistic exceeds all \(K\) replicas
with probability tending to one, its permutation p-value attains the
minimum \(1/(K{+}1)\le\alpha/5\), the Bonferroni rule rejects, and the
protocol selects the block calibration.
\end{proof}

\begin{remark}
Together with the Remark on testing exchangeability itself in \ref{sec:A5}, P7 makes the
diagnostic a bona fide test of exchangeability: exact level under the
exchangeable null, and power tending to one against every stationary
ergodic alternative whose serial structure is visible to one of the five
statistics at some lag \(k\le L\). Alternatives invisible to all five
--- structure only beyond lag \(L\), or orthogonal to levels, scales and
lagged cross-moments --- are not covered, in line with the impossibility
half of that Remark.
\end{remark}

%% file: appendixB.tex
\section{Additional experiments}

Sections~\ref{sec:B1}--\ref{sec:B5} support the synthetic results of Section 6: sensitivity to the design parameters (\ref{sec:B1}), the ablation of the functional (\ref{sec:B2}), finite-sample illustrations and simulation studies (\ref{sec:B3}), model-based and log-determinant baselines (\ref{sec:B4}), and variability over feature and calibration draws and over the break location (\ref{sec:B5}). Sections~\ref{sec:B6} and~\ref{sec:B7} vary what the main text fixes: the variable-block dimension of Section 4.1 and the kernel bandwidth. Section~\ref{sec:B8} gives the calibration and robustness checks for the applications of Sections 6.7 and 6.8, Section~\ref{sec:B9} the wall-clock timings for Sections 4.6 and 4.8, Section~\ref{sec:B10} the Gram form and the feature dimension, and Section~\ref{sec:B11} the non-Gaussian suite of Section 6.1 in full. Section~\ref{sec:B12} collects figures and tables cited from Sections 3.4, 4.3, 6.1--6.3, 6.5 and 6.8. The remaining sections report the matrix-information baseline BMCTC (\ref{sec:B13}), an exploratory resampling check of the whole two-stage procedure of Section 4.7 (\ref{sec:B14}), detection, localisation and permutation calibration in the single-break comparison, and a marginal-shape control for S2 (\ref{sec:B15}), the permutation-free calibration of Section 4.6 (\ref{sec:B16}), calibration under serial dependence (\ref{sec:B17}, Section 4.3), the combined statistic of Section 4.5 (\ref{sec:B18}), a simulation on empirical weather margins (\ref{sec:B19}), a representation-learning baseline (\ref{sec:B20}), a sign reversal of a correlation (\ref{sec:B21}) and the US financial panel of Section 6.8 under a fixed protocol (\ref{sec:B22}).

\subsection{Sensitivity}\label{sec:B1}

This section varies the series length, the candidate-window parameter
and the feature dimension \(D\) on scenario S2 at \(a=0.7\), with the
other settings of Section 6.1 held fixed. The sweep uses 200 replicates
and its own base seed, and estimates its thresholds from 100 null
replicates instead of 250, so its absolute levels are not comparable
with Table 1: at \(n=600\) it gives 0.475 where Table 1 gives 0.35 for
the same configuration.

Localised power (\(|\hat\tau-\tau|\le n/20\)) at \(n=300\), 600 and
1200 is 0.215, 0.475 and 0.910. Varying the candidate-window parameter
over \(\{40,60,100\}\) leaves power at 0.475: the global statistic of
Section 4.2 does not depend on the window, which only trims the
candidate grid, and each of the three grids contains the true break.
The sliding two-window form is not examined. At \(D=4\), 8 and 16 power
is 0.32, 0.475 and 0.41. \(D=8\) is used throughout, and no parameter is
varied per scenario; since \(D=8\) is also the best of the three values
on this sweep, the sweep does not validate that choice independently.

\subsection{Ablation: entropy versus Frobenius versus spectral}\label{sec:B2}

This ablation tests the claim of Sections 3.1 and 3.4 that, on the
alternatives targeted here, the entropy functional is more sensitive
than the Frobenius and spectral-only functionals of the same second
moment. It uses a separate scenario suite (200 replicates, \(n=600\),
break at 300, same 0.05 calibration) spanning mean shift, covariance
rotation at fixed spectrum, mixture shape, heavy tails, and dependence
without correlation, with all three functionals computed from the
\emph{same} raw-data random Fourier features (64 components) --- that
is, without the rank transform and tensor structure of Section 4.1, so
that the functional is the only quantity varied. Figure~\ref{fig:B1} shows the
comparison.

\begin{figure*}[!t]
\centering
\includegraphics[width=\textwidth,height=0.42\textheight,keepaspectratio]{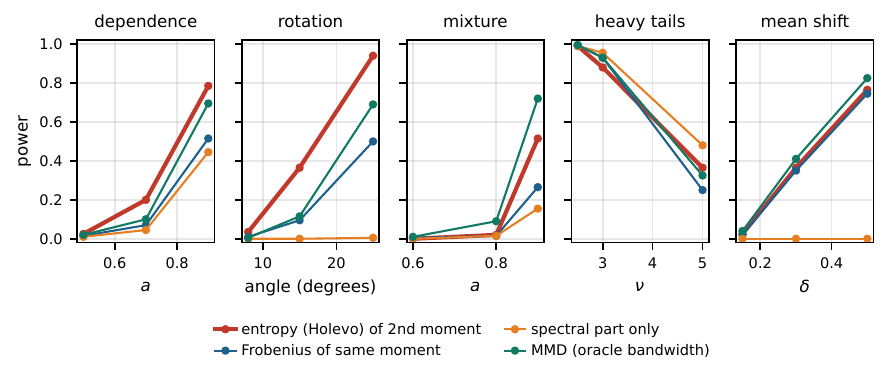}
\caption{Ablation across five alternative families --- dependence without correlation, covariance rotation, mixture shape, heavy tails and a mean shift, one panel each (\(x\)-axes: \(a\), rotation angle, \(a\), degrees of freedom \(\nu\), shift \(\delta\)) --- comparing entropy, Frobenius and spectral-only functionals of the same feature moment against MMD with an oracle bandwidth. Among the three functionals the ordering reverses only on heavy tails.}
\label{fig:B1}
\end{figure*}

On the two alternative types this paper targets, the entropy (Holevo)
functional is the strongest of the three: dependence without correlation
0.79 against 0.52 (Frobenius) and 0.45 (spectral-only), covariance
rotation 0.94 against 0.50 and 0.01. There the spectral-alignment excess
\(\Delta_{\mathrm{nc}}\) of Lemma 2 carries the signal, as Section 3.1
anticipates. On mixture shape the entropy functional also leads, 0.52
against 0.27 (Frobenius) and 0.16 (spectral) at the strongest level, and
at the weakest level all three are near zero. The exception
is the heavy-tailed family, where the spectral-only statistic is the
strongest of the three (at \(\nu=5\), 0.48 against 0.37 for the entropy and 0.25 for the Frobenius functional; at \(\nu=3\), 0.96 against 0.88 and 0.93), which indicates that the spectral-alignment excess does not capture a change
in tail weight. Across the fifteen cells the Frobenius functional
exceeds the entropy functional in three, none of them on a family this
paper targets. Two are ties to within one replicate. The third is heavy
tails at \(\nu=3\), 0.93 against 0.88: the paired McNemar test (both functionals are computed from the same replicates) gives 15 discordant replicates favouring Frobenius
against 5 favouring entropy and an uncorrected exact \(p\) of 0.041,
which does not survive correction for the three comparisons
(\(p=0.12\)). It is a reversal within the family already identified as the
exception, not an established difference. Against MMD with an oracle bandwidth (best of five), the
entropy statistic is higher on the dependence and rotation alternatives
(0.79 against 0.70; 0.94 against 0.69) and lower on mean shift (0.77 against 0.83 at the largest shift) and on mixture shape (0.52 against 0.72).

\subsection{Finite-sample illustrations and simulation studies}\label{sec:B3}

This section illustrates the boundary scaling of P6 and the finite-sample behaviour of the implemented permutation test, reports the localisation error that P3 concerns, investigates three properties that Supplement A leaves unproved, and compares global and within-segment ranks at the boundary. At \(\alpha=0.02/0.05/0.10/0.20\) the permutation p-value rejects at empirical rates 0.030/0.050/0.110/0.203 (300 null replicates, \(K=49\)), within
binomial error of nominal. The p-value form matters at finite \(K\): on the independent null of Section 4.3 under pair permutation, the rule that compares \(T_{\mathrm{obs}}\) with the empirical \((1-\alpha)\) quantile of the \(K\) replica maxima, which drops the \(+1\) of (9), rejects in 7.0\% of replicates at \(K=99\) against a nominal 0.05, where the p-value rule gives 5.0\%. Figure~\ref{fig:B2} collects these illustrations.

\begin{figure*}[!t]
\centering
\includegraphics[width=\textwidth,height=0.42\textheight,keepaspectratio]{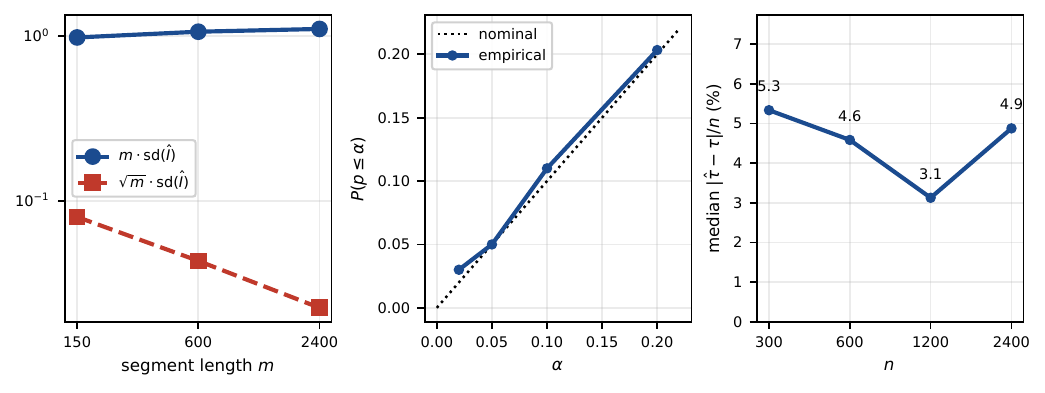}
\caption{Finite-sample illustrations. Left: scaling of the plug-in under independence, illustrating the rate $m$ of P6 (Sections~\ref{sec:A6} and \ref{sec:A12}). Middle: empirical rejection frequencies of the permutation test, whose finite-sample validity under the stated exchangeability assumptions is established by P1. Right: relative localisation error of a permutation-studentised estimator (S2 at \(a=0.9\), 120 replicates, 25 pair-permutation replicas); this panel does not verify P3, which concerns the unstudentised criterion.}
\label{fig:B2}
\end{figure*}

The left panel illustrates the \(m\) scaling predicted by P6: \(m\,\mathrm{sd}(\hat I)\) is 0.98, 1.07 and 1.11 at \(m=150\), 600 and 2400, whereas \(\sqrt m\,\mathrm{sd}(\hat I)\) falls by a factor of 3.6 over the same range; it is not a separate proof of the limiting law. The middle panel reports the empirical rejection frequencies above. The right panel shows the relative
localisation error, which is stable across \(n\) without a monotone
trend; the power scaling with \(n\) that P3 predicts is reported
numerically in Section~\ref{sec:B1} rather than plotted.

\emph{Identifiability in \(D\)} ([NV-5], for \ref{sec:A8}(iii)): the
population DOMI of four dependent families, estimated on samples of
20,000 pairs and corrected by the value on the same sample with \(Y\)
permuted, is bounded away from zero at every \(D\in\{4,8,16,32\}\), and
does not fall back towards it as \(D\) grows --- Gaussian \(r=0.5\):
0.030, 0.084, 0.088, 0.100; correlation-free \(|X|\)-dependence: 0.025,
0.044, 0.044, 0.063; Clayton at \(\tau=0.5\): 0.080, 0.207, 0.227,
0.255; sign-mixed dependence at zero correlation: 0.046, 0.082, 0.080,
0.116 --- while the independent pair returns \(0.0000\) at every \(D\)
to four decimals. No dimension in this range is blind to any of these
alternatives, which is what Section 3.5 requires of P4, and the
sign-mixed family shows the property holds where correlation carries no
signal at all.

\emph{The studentised scan} ([NV-6], open in Section 7):
under a constant-dependence null the law of the deployed studentised
maximum stabilises as \(n\) grows --- mean 1.869, 1.917, 1.959, 1.969
and 95th percentile 3.477, 3.580, 3.617, 3.622 at
\(n=300,600,1200,2400\) (10,000 simulated maxima each), with the
Kolmogorov distance between consecutive sample sizes falling across
0.022, 0.019, and 0.008. As an empirical control, the permutation test rejects at rates 0.035, 0.025, 0.030 and 0.040 at the nominal 0.05.

\emph{Segment-number consistency} ([NV-7], open in Section 7):
repeating the multiple-break design of Section 6.3 with the break
fractions held fixed and only \(n\) varying gives the correct number of breaks, \(K=3\), in 47.5\%, 99.5\%, and 100\% of replicates at
\(n=1800,3600,7200\) (adjusted Rand index 0.63, 0.98, 0.99; penalty and
grid resolution held fixed, null detection \(\le 1\%\) throughout). The
\(n=1800\) row is Table~\ref{tab:B7}'s configuration re-run on that fixed grid: it
reproduces the 47.5\% and returns 0.63 against the 0.62 reported there,
the difference being the grid. The failures at \(n=1800\) are
under-segmentation, not over-segmentation, and they disappear once the
segments are long enough.

\emph{Global against within-segment ranks at the boundary} ([NV-8], Section~\ref{sec:A12}): in 2000 null replicates at \(n=600\), \(D=8\) and the grid of Section 4.6, the empirical distribution functions of the global-rank, within-segment-rank and oracle statistics, pointwise at \(\lambda\in\{0.25,0.5\}\) and for the supremum of the scan, differ by at most 0.026 in Kolmogorov--Smirnov distance; their paired differences decrease as \(n\) grows. These are finite-sample illustrations of the common limit, not a test of equality of the finite-sample laws.

\subsection{Comparison with model-based and log-determinant alternatives}\label{sec:B4}

Three further baselines test whether existing correlation-monitoring
machinery suffices (200 replicates, calibration identical to Section 5):
a correlation-constancy CUSUM in the spirit of Wied et al.~[30]; a
two-stage DCC(1,1)-GARCH fit [29] followed by a CUSUM on the fitted
correlation path; and, as the closest non-quantum spectral relative, the
Bach--Jordan log-determinant mutual information computed by regularised
canonical correlations on the same rank features. Each baseline, like
the DOMI figures quoted beside it, is calibrated against the null that
matches each alternative. On the correlation
alternative S1 the correlation CUSUM is the strongest of the three baselines of this section
(power 0.75/0.91 at \(r=0.35\)/0.5, against 0.60/0.77 for DOMI in its random-feature form and 0.82/0.99 for its Gram form in the separate 500-replicate run of Section 6.1 and Figure 2). On the correlation-free
alternatives the ordering inverts: on S2 and S4 both correlation-based
methods attain power at or below 0.04 (structural insensitivity, since the tracked parameter does not change) while DOMI attains 0.09--0.87. Under the
marginal-only change M1 the DCC statistic is well behaved (0.02--0.05)
but the correlation CUSUM's false-alarm rate inflates to 0.09--0.20.

On the 2021--2022 stock--bond window both correlation methods reject under
pair-permutation calibration (\(p=0.01\) and 0.04), consistent with the
monotone-correlation character of the change; they place the break at
2022-05-25 and 2022-06-13, close to Spearman's 2022-05-31 and not at the
December 2021 candidate. These p-values rest on the pair-permutation
calibration, which Section 6.8 finds unreliable on this window.

The log-determinant statistic, used without the studentisation and bias
machinery of Sections 3.3 and 4.3, attains low power throughout:
0.00/0.00/0.15 on S1, where the correlation CUSUM attains
0.32/0.75/0.91, and 0.00--0.12 on S2 and 0.00--0.27 on S4, where DOMI
attains 0.09--0.87. Most of the loss is in localisation: at the strongest levels of S2
and S4 it exceeds its threshold in 90\% and 94\% of replicates but
places the break outside the tolerance, the edge-bias pattern that the
per-\(t\) studentisation of Section 4.3 is designed to remove. The
statistic differs from DOMI both in the functional and in the missing
studentisation, so the comparison does not separate the two. Under the
marginal-only change M1 its false-alarm rate rises to 0.06/0.17/0.23,
above the correlation CUSUM's at the largest shift. On the 2021--2022 stock--bond window it gives \(p=0.64\) under pair permutation, where DOMI, Spearman and the combined statistic give medians of 0.268, 0.006 and 0.016 over the permutation draws of Table~\ref{tab:B19}.

\subsection{Feature-draw variability and break location}\label{sec:B5}

\emph{The random-feature draw.} Every power in this paper, unless stated
otherwise, is computed from one draw of \(D\) frequencies per block. The
statistic is a function of that draw, so this section measures how far
the reported powers move when the draw changes. Across six independent
draws, on S2 at \(a=0.9\) with \(n=600\), the powers are 0.675, 0.810, 0.840,
0.865, 0.915, and 0.955 --- a spread far wider than the binomial
standard error of 0.035, with the six draws agreeing on the
reject/retain decision in only 62.5\% of replicates. On real data the
effect is of the same size: in a daily-resolution Seoul
temperature--wind test on a 2022--2023 window --- a separate check, not
the hourly stage two of Section 6.7 --- the block-permutation
p-value ranges from 0.37 to 0.86 across the same six draws. Averaging
the statistic over the six draws gives 0.915 and removes most of the
variation, at six times the cost. The comparative claims of this paper
are paired within a draw (DOMI and HSIC are computed from the \emph{same} features), so a feature draw moves both statistics; a power
quoted as a property of the method carries this additional uncertainty.

\emph{The calibration draw.} A power is also a function of which null
replicates set the threshold and the pointwise studentising moments.
Scoring one fixed set of alternative curves under three independent null
draws isolates that effect. On the correlation-free scenario S2 the DOMI
power runs over 0.352--0.444 at \(a=0.7\) and 0.872--0.954 at \(a=0.9\),
while HSIC moves the other way, over 0.062--0.030 and 0.658--0.500; the
DOMI-minus-HSIC margin therefore runs from 21 to 45 percentage points
across the three draws. On the disappearance alternative S4 at
\(r=0.85\) the spread is wider still. The DOMI power runs over
0.362--0.702 across the three draws, a wider range than the
feature-draw spread above, and the DOMI-minus-HSIC margin takes the
values 4.6, 15.6 and 40.4 points: positive under every draw, but well
below the reported 15.6 under one of them. Where the change is one of
correlation the margin nearly vanishes under some draws: 0.2, 3.6 and
17.4 points at S1 \(r=0.35\), and 2.8, 5.2 and 19.8 at \(r=0.5\). These
values are computed from per-replicate records, with the stored
matched-null files where the null is level-independent and fresh null
draws for S4.

\emph{The location of the break.} Every synthetic scenario places the
change at the midpoint, so every localised power reported above is a
centre-break power; Section~\ref{sec:B15} repeats the comparison of Table 1 with
the break at a quarter and three quarters of the series. Moving the break in S2 at \(a=0.9\) while holding everything else fixed
separates detection from localisation sharply. At \(\tau/n = 0.2\) the
test still rejects in 95\% of replicates but the estimate falls within
the tolerance in 0\% of them; at \(\tau/n = 0.33\), 100\% and 0\%; at
the midpoint, 100\% and 87\%; at \(0.67\), 100\% and 51\%; and at
\(0.8\), 79\% and 60\%. The statistic detects off-centre breaks but misplaces them when the dependent segment is the longer one (\(\tau/n<0.5\) here, where dependence begins at the break), and the failure is not repaired by any of the four
location rules we compared --- studentised, raw weighted, and the two
bias-centred variants. The segmentation route localises better, finding the break within the tolerance in 64\% of replicates
at the centre and in 23\% to 57\% off-centre across the same positions,
which is a further reason the applications use it for stage one.
Localisation consistency is not claimed for the deployed estimator
(Section 3.5); this is its practical consequence.

\subsection{Vector-valued blocks}\label{sec:B6}

The construction of Section 4.1 is written for two \emph{blocks}, not
two scalars: ranks are formed per component and one random Fourier map
is drawn per block, so the feature dimension \(D\) does not change with
the block dimension \(d\). Every other experiment in this paper uses
\(d=1\). This study measures the effect of the block form; at a fixed feature dimension it lowers both power and specificity.

Scenario MV-S2(\(d\)) applies scenario S2 coordinatewise:
\(X\in\mathbb{R}^d\) standard normal, \(Y=\varepsilon\) before the break
and \(Y_j=\sqrt{1-a^2}\,\varepsilon_j+a|X_j|\varepsilon'_j\) after it,
with \(a=0.7\). Every coordinate pair therefore carries the same correlation-free coupling, with the change of margin shape that S2 carries, and only the aggregation over coordinates
changes with \(d\). MV-M1(\(d\)) is the matching specificity control:
the blocks stay independent throughout and only the marginal scale of
\(X\) changes, by a factor of two. Calibration is that of Section 5 at
500 replicates and \(D=8\). Alongside the block test we run the pairwise
alternative, \(d\) separate scalar tests on the coordinate pairs, each
at the Bonferroni level \(0.05/d\). Table~\ref{tab:B1} reports the outcome. At \(d=1\) the design is S2 at \(a=0.7\), where Table 1 reports 0.35 against 0.46 here; the two runs use different replicates and null draws, and the calibration draw alone moves this power over 0.35--0.44 (Section~\ref{sec:B5}).

\begin{table*}[!t]
\centering
\caption{Vector-valued blocks. Localised power on the coordinatewise dependence change MV-S2 (a detection counts only if the break is localised to within 30 samples) and rejection rate on the marginal-only control MV-M1 (any detection), against the block dimension $d$ at a fixed feature dimension $D=8$ (500 replicates).}
\begin{tabular}{l cccc}
\toprule
Rate at a false-alarm rate of 0.05 & $d{=}1$ & $d{=}2$ & $d{=}3$ & $d{=}5$ \\
\midrule
MV-S2, DOMI on the blocks & 0.46 & 0.06 & 0.05 & 0.02 \\
MV-S2, DOMI pair by pair, Bonferroni & 0.46 & 0.45 & 0.64 & 0.63 \\
MV-S2, HSIC on the blocks & 0.06 & 0.00 & 0.00 & 0.01 \\
MV-M1, DOMI on the blocks: rejection rate (nominal 0.05) & 0.11 & 0.04 & 0.48 & 0.50 \\
MV-M1, DOMI pair by pair, Bonferroni: rejection rate (nominal 0.05) & 0.11 & 0.04 & 0.06 & 0.05 \\
\bottomrule
\end{tabular}
\label{tab:B1}
\end{table*}

Both power and specificity fail as \(d\) grows. The block test loses power
while the same statistic applied pair by pair does not, and the block test's specificity, near nominal at \(d=2\), is lost from \(d=3\) onward, where the pairwise test keeps it. At \(d=1\) this run rejects in 11\% of MV-M1 replicates; four further independent draws of the replicates and of the null calibration give 2.6--5.6\%, near the 0.05 of Table 1. The loss of power is not an artefact of estimation variance: at \(n=20{,}000\)
the difference between the two regimes' DOMI is 0.0169 at \(d=1\), 0.0045
at \(d=2\), 0.0026 at \(d=3\) and 0.0011 at \(d=5\), and raising the
feature dimension to \(D=16\) restores it to 0.0122 at \(d=2\) (one draw
of the data and of the frequencies; Section~\ref{sec:B5} measures how much the
latter moves such a number). Eight random frequencies spread over a
\(d\)-dimensional input resolve a coupling confined to a few
coordinates far more weakly, and the same thinning makes the plug-in
bias more sensitive to a segment's rank distribution --- which is what a
marginal change alters. That second half is measurable on its own. Under
MV-M1 the blocks are independent throughout and the population DOMI is zero on both sides of the break, so whatever the statistic sees there is a finite-sample effect of bias and sampling variability together; at the deployed segment length \(m=300\) its mean absolute value
is 0.0038, 0.0037, 0.0085 and 0.0093 at \(d=1,2,3,5\) (200 replicates).
The response to a purely marginal change thus more than doubles between
\(d=2\) and \(d=3\), which is where Table~\ref{tab:B1} shows the specificity
collapsing. Against the contrasts above, the ratio of signal to spurious
response falls from 4.4 at \(d=1\) to 1.2 at \(d=2\) and 0.31 at
\(d=3\). It falls below one between \(d=2\) and \(d=3\), where the
specificity fails; power has already collapsed at \(d=2\), where the
ratio is 1.2.
Routing the coupling through a fixed orthogonal mixing of the
coordinates, so that no single pair carries all of it, does not change
the ordering at \(D=8\) where the contrasts are resolvable: the block
contrast is 0.0047 against 0.0162 for the strongest coordinate pair at
\(d=2\), and 0.0026 against 0.0060 at \(d=3\); by \(d=5\) both have
fallen to the resolution floor (0.0009 and 0.0008).

The block form is available but is not the default.
The procedure deployed in Sections 6.7 and 6.8 tests variable pairs, and
a block of dimension \(d\) requires a feature dimension
that grows with it, at a cost of \(O(D^6)\) per segment evaluation. We
have not established the rate at which \(D\) must grow, and a coupling
that no coordinate pair carries at all --- the case in which a block
test would be the only option --- lies outside the designs used here.

\subsection{Kernel bandwidth, and a full-Gram HSIC}\label{sec:B7}

Two questions are left open by the design of Section 5. The
random-feature bandwidth is set once by the median heuristic and never
varied, so every power reported in this paper belongs to one operating
point rather than to the statistic at its best; and the HSIC comparator
is built from the same \(D=8\) random features as DOMI, which makes it a
clean ablation of the functional but not the standard method, whose Gram
matrix is full. One sweep answers both. Each statistic is recomputed at
Gaussian bandwidths of \(4,\,2,\,1,\,\tfrac12\) and \(\tfrac14\) times
the median heuristic under the calibration of Section 5 at 500
replicates; the full-Gram HSIC is the biased V-statistic on the same
global ranks, evaluated over the same candidate grid by two-dimensional
prefix sums. Table~\ref{tab:B2} reports the outcome.

\begin{table*}[!t]
\centering
\caption{Kernel bandwidth. Localised power at a false-alarm rate of 0.05 (500 replicates), and detection rate under the marginal-only control M1 at $s=2$, against the Gaussian bandwidth as a multiple of the median heuristic. Each statistic was swept over $\times 4,\,\times 2,\,\times 1,\,\times\tfrac12$ and $\times\tfrac14$; the table prints the deployed setting and, for each statistic, the settings that carry the comparison --- DOMI is best at $\times\tfrac12$ in five of the seven cells and at $\times\tfrac14$ on the copula-shape cell and, by 0.01, on S4 at $r{=}.85$; the full-Gram HSIC is best at $\times\tfrac14$ in all seven. The full grid is in the released outputs. As in Table 1, each column is calibrated against its own pre-change regime held throughout. The first row is the deployed setting and reproduces Table 1.}
\begin{tabular}{l cccccccc}
\toprule
Bandwidth & S1 $r{=}.35$ & S1 $r{=}.5$ & S2 $a{=}.7$ & S2 $a{=}.9$ & S3 $\tau{=}.5$ & S4 $r{=}.7$ & S4 $r{=}.85$ & M1 \\
\midrule
DOMI, $\times 1$ (deployed) & 0.60 & 0.77 & 0.35 & 0.87 & 0.01 & 0.09 & 0.55 & 0.05 \\
DOMI, $\times 1/2$ & 0.78 & 1.00 & 0.44 & 0.98 & 0.03 & 0.20 & 0.70 & 0.04 \\
DOMI, $\times 1/4$ & 0.06 & 0.32 & 0.10 & 0.76 & 0.19 & 0.16 & 0.71 & 0.07 \\
HSIC, full Gram, $\times 1$ & 0.47 & 0.60 & 0.03 & 0.51 & 0.01 & 0.01 & 0.26 & 0.06 \\
HSIC, full Gram, $\times 1/4$ & 0.74 & 0.99 & 0.25 & 0.96 & 0.03 & 0.12 & 0.64 & 0.09 \\
\bottomrule
\end{tabular}
\label{tab:B2}
\end{table*}

With each statistic at its best bandwidth on the shared grid, the comparison goes to DOMI
in every one of the seven cells, so the advantage in Table 1 is not an
artefact of holding HSIC to the feature dimension of DOMI: a full Gram
matrix with a tuned bandwidth does not close it. The margins divide into
two groups. On the copula-shape scenario and on three of the four
correlation-free cells the lead is decisive, 7 to 19 percentage points;
on the remaining three the two are separated by 0.6 to 3.6 points and are not distinguishable at 500
replicates. A tuned full-Gram HSIC exceeds the
\emph{untuned} DOMI of Table 1 in six cells of seven. Tuning one side and
not the other decides that comparison, not the functional.

The median heuristic is not DOMI's optimum, and the deployed
setting is not the best anywhere. Halving it raises power in
\emph{every} cell --- S1 at \(r=0.35\) from 0.60 to 0.78, S2 at
\(a=0.7\) from 0.35 to 0.44, S4 at \(r=0.7\) from 0.09 to 0.20, and even
the copula-shape scenario from 0.01 to 0.03 --- while the false-alarm
rate under M1 falls from 0.052 to 0.044. A quarter is better still on
the copula-shape scenario, and marginally on S4 at \(r=0.85\), and there the correlation alternatives
collapse, so no single value is best everywhere. The median heuristic is kept in the main text. It is the conventional
rule, so a power reported at it is comparable with the literature. The
dominance of \(\times\tfrac12\) is a selection over five values made
after seeing these seven cells. Every other method comparison in this paper is made
\emph{within} a bandwidth, and the operating point moves the levels far
more than the ranking. Against the full-Gram HSIC, DOMI leads in all
seven cells at \(\times 4\) and \(\times\tfrac12\), in six at
\(\times 1\) and in five at \(\times 2\), the exceptions being cells where Table 1 does not put it ahead of every statistic either. The ranking turns over at \(\times\tfrac14\), where DOMI leads in only three.

The copula-shape change is partly recoverable. At a quarter of
the median bandwidth DOMI attains 0.19 on S3 at \(\tau=0.5\), against
0.008 at the deployed setting and 0.032 for the tuned full-Gram HSIC: a
more than twentyfold gain over the deployed operating point and six times the tuned full-Gram HSIC, though below the 0.34 of the copula test (Table 1). Section 6.4 predicts the
direction: the obstacle there is estimation variance and not
non-identifiability, and a narrower kernel concentrates the feature
moment where the two copulas differ. The size of the recovery, however,
does not license calling S3 a limitation of the deployed bandwidth
alone. The narrow kernel also costs the competitor its
specificity. At \(\times\tfrac14\), where the full-Gram HSIC is most
powerful, its detection rate under the marginal-only control M1 is
0.092, nearly twice the nominal 0.05, against 0.044 for DOMI at its own
best bandwidth. The cost to DOMI is visible in the same
row: the correlation alternatives collapse, S1 at \(r=0.35\) falling
from 0.60 to 0.06, and it is at this bandwidth that the full-Gram HSIC overtakes DOMI in a majority of cells, four of seven. That reversal is
representational rather than functional. A fixed draw of \(D\) random
Fourier features approximates a Gaussian kernel with an error that grows
as the kernel narrows relative to the data scale, and the full Gram
matrix carries no such budget; raising \(D\) at the same bandwidth recovers part of what was lost, S1 at \(r=0.35\) going from 0.07 at \(D=8\) to
0.35 at \(D=12\) and 0.55 at \(D=16\) in a separate run of this sweep. The budget cuts both ways,
however: on the copula-shape scenario the same increase \emph{removes}
the recovery, 0.19 at \(D=8\) against 0.05 at \(D=12\) and \(D=16\), the
extra directions costing more in variance than they add in signal.
Bandwidth and feature dimension are therefore jointly constrained rather
than separately tunable, which is a further reason to fix both by rule.

\subsection{Calibration and robustness of the real-data analyses}\label{sec:B8}

This section gives the checks behind the calibration choices and
results of Sections 6.7 and 6.8.

\emph{Coverage.} Every analysis uses the hours at which both variables of a pair are present in the API Hub record; at most 148 of the 70,128 hours of 2018--2025 (0.21\%) are dropped for any station--pair (Suwon temperature--humidity). The two-week candidate grid counts observations, so an interval spans a few hours more than two weeks where hours are missing.

\emph{Effective sample sizes.} Over the full eight-year hourly record, which is what the segmentation of stage one sees, the effective sample sizes \(n(1-\hat\rho_1)/(1+\hat\rho_1)\) of Section 6.7 are 384--905 for temperature, 924--2265 for humidity, and 5073--12,418 for wind across the twelve stations.

\emph{Random-date rejection frequencies.} For each of the 36 station--pair analyses we drew 20
dates at random, at least a year from either end of the record, and ran
stage two as in Section 6.7 at each, with $K=999$. Under the
unrestricted six-week block permutation the dependence p-values fall at
or below 0.10, 0.05 and 0.01 at 18.5\%, 9.7\% and 1.25\% of the 720 dates, and the smaller of the two marginal p-values at 40.8\%, 23.2\% and 4.3\%. With blocks permuted only within calendar seasons the
dependence test gives 12.1\%, 5.4\% and 0.42\%, and the marginal test 23.1\%, 11.4\% and 2.5\% (each variable separately: 12.3\%, 5.8\% and 1.3\%).
The season-restricted dependence test is thus close to the nominal level at 0.05 and 0.01 and above it at 0.10, and the marginal test stays above 9.75\%, the rate at 0.05 for the smaller of two independent p-values. The windows overlap, may contain changes, and served to choose the scheme, so these frequencies are a diagnostic of calibration sensitivity, not verified null rejection rates. Different block permutations can yield equal test statistics without reproducing the observed arrangement; so that numerically equal statistics are not treated as strictly ordered, every permutation p-value in Sections 6.7 and 6.8 counts a replica whenever \(T_k\ge T_0-10^{-9}\max(1,|T_0|)\).

\emph{Tied values.} The anomaly series carry many tied values: averaged over consecutive two-year windows, 19--31\% of temperature, 43--65\% of wind-speed and 63--77\% of humidity observations share their value with another observation, across the twelve stations. The ranks break ties in time order, and stage two computes the features once and permutes them, which the Remark on randomness drawn independently of the data (Section~\ref{sec:A5}) does not cover. Repeated with ties broken at random, which it covers, stage two gives the same conclusions: the p-values of the five candidates at \(p\le0.05\) and of every split-sample window are unchanged, those of the other candidates move by at most 0.004, and the random-date frequencies at 0.10, 0.05 and 0.01 are unchanged, so the candidate counts and the classification are too. The marginal test uses the raw values and is unaffected up to Monte Carlo variation. Random tie-breaking addresses the tie-breaking issue only; the qualifications on the selected windows and on block exchangeability remain. On synthetic records with a constant Gaussian copula (\(n=600\), \(K=99\), 400 replicates) rounded to tie shares of 0, 28, 54 and 70\%, pair permutation rejects in 4.0--5.5\% of replicates whether ties are ranked in time order with the features fixed, broken at random, or re-ranked on every permuted record; at a tie share of 91\% all three reject in 7.25--7.75\% (standard error 1.1 percentage points), the two implementations the proof covers included. A re-check with 2000 new independent replicates gives 102, 103 and 106 rejections (5.10, 5.15 and 5.30\%; 95\% intervals 4.2--6.2, 4.2--6.2 and 4.4--6.4\%), so the excess is not reproduced.

\emph{Candidates.} Under the season-restricted stage two ($K=9999$) the
five candidates at $p\le0.05$ are Suwon humidity--wind ($p=0.0005$), Jeonju humidity--wind (0.0179), Daegu temperature--humidity (0.0237), Seoul temperature--wind (0.0335) and Jeonju temperature--humidity (0.0370);
every other candidate has $p\ge0.10$. Under the unrestricted scheme
the strongest candidate was Jeonju temperature--humidity
($p=0.001$). With the unrestricted scheme, 20 random feature draws
left Jeonju temperature--humidity and Suwon humidity--wind at $p\le0.012$
in every draw, while the three borderline candidates of that analysis
reached $p\le0.05$ in 65--75\% of the draws.

\emph{Split sample.} Selecting on alternate six-week blocks (24 hours trimmed at each block end) and testing on the others, and then the reverse, yields 22 and 11 candidates. Under the unrestricted scheme one
candidate in each direction reaches $p\le0.05$ (0.044 and 0.031) and
neither survives the multiplicity correction; under the
season-restricted scheme none reaches 0.05. With about two blocks per season
in a half-window, the season-restricted test has 6 to 48 distinct
permutations and 13 of the 33 test windows cannot reach $p\le0.05$, so
the split analysis is inconclusive at this record length.

\emph{Weekly financial returns.} On the weekly record a portmanteau on the returns gives \(p\) between 0.026 and 0.343 across the four series, and on \(|r_t|\) and \(r_t^2\) at most 0.0014 in every one, so weekly summation does not remove the volatility clustering (Section 6.8).

\emph{Diagnostic settings.} The diagnostic uses \(L=168\) lags (one week) on the hourly weather records, with the autocorrelation sums capped at 8760 lags, \(L=20\) on the weekly financial series and on the daily 2021--2022 window, and \(L=12\) on the monthly series; the synthetic designs of Section~\ref{sec:B17} use \(L=20\).

\emph{Financial breaks.} Stage two for the three breaks of Section 6.8
uses the diagnostic's recommended block length on the full record (44
weeks for the weekly pair, pair permutation for the monthly stock--bond
pair, twelve months for the monthly Korean pair) and the largest centred
window holding a whole number of blocks. The dependence p-values are
0.054 (pandemic onset, 2020-02-27), 0.988 (2021-08) and 0.395
(2018-01); across 20 feature draws they fall at or below 0.05 in 40\%,
0\% and 0\% of draws.

\subsection{Timing}\label{sec:B9}

The complexity of the procedure is given in Sections 4.6 and 4.8.
Table 2 of the main text compares the two computational forms of Section
4.8 on one replicate of scenario S2 at \(a=0.7\) for each series length,
timing one full difference curve over the candidate grid with no
permutation replicas. Each entry is a single-thread wall-clock time (BLAS threads pinned to one) on an otherwise idle Intel Core i9-12900 host; the per-candidate time is the median over 15 candidates
spread over the grid, and entries marked with an asterisk are that
median multiplied by the grid size. Across \(n\) from 600 to 19,200 the
cost per candidate of the random-feature form at \(D=8\) stays between 0.22 and 0.99 ms, whereas that of the Gram form rises from 21 ms at \(n=600\) to 17 s at \(n=4800\). A version of the Gram form that
subsamples at most 256 rows per segment costs between 7 and 11 ms per candidate for \(n\) from 600 to 19,200; Section~\ref{sec:B10} reports its effect on
localisation.

On a single core at \(D=8\), one full difference curve takes 0.14 s at \(n=600\), 0.32 s at \(n=1200\) and 1.5 s at \(n=4800\) (Table 2 of the main text), so the 100 difference curves alone of a \(K=99\) permutation test take about 14 s, 32 s and 146 s; the
end-to-end times below include the remaining steps.

The order-2 Gram variant has an entropy equal to the logarithm of a sum
of squared Gram entries, so two-dimensional prefix sums of the three
Gram matrices give every segment in constant time. One difference curve then takes 0.04 s at \(n=600\), 1.2 s at \(n=4800\) and 3.8 s at \(n=9600\), with a peak memory of 0.11, 0.60 and 2.2 GiB; three
\(n\times n\) arrays of doubles exceed 16 GiB beyond \(n\approx 26{,}750\).

End to end, on one core of the same host, Algorithm 1 with \(K=99\) takes 14 s at \(n=600\), 148 s at \(n=4800\) and 596 s at \(n=19{,}200\), with a peak memory of 112 to 870 MiB. Stage one of Section 6.7 on one station--pair series (\(n=70{,}069\), 209 two-week intervals) takes 136 s and 120 MiB, of which 130 s is the penalty calibration on 20 permuted copies. These timings belong to different workloads and implementations: Algorithm 1 evaluates every split of its grid (481 splits at \(n=600\)) through the general segment cache, the permutation timings of Section~\ref{sec:B16} use 161 split fractions and a block-sum scan, and stage two of Section 6.7, 27 windows each tested for dependence and for the scale of both margins at \(K=9999\), took 1736 s on six single-thread processes.

\emph{Test computation.} Table~\ref{tab:B16} times the pair-permutation test of Section 4.3 ($K=99$) for every statistic of Table~\ref{tab:B11} on scenario S2 at $a=0.7$ with the break at $n/2$, from ranks and features through the observed and 99 permuted curves, studentisation and the p-value; data generation, permutation-index generation and process start-up are not timed. Each test runs in a fresh single-thread process, four tests at a time on four pinned performance cores of an Intel Core i9-12900 host, while another job used about 3 CPU threads and the GPU (1-minute load 2.8--9.6 on 24 logical CPUs, the four tests included); the entry is the median over 10 data sets. The candidates are every split from $n/10$ to $9n/10$, as in the synthetic experiments, or 161 split fractions from 0.10 to 0.90, as for the long records of Section 6.7, where DOMI is computed from interval moment sums. A test that did not finish within 30 minutes or exceeded 16 GB of memory is marked, and the statistic was then not run at larger $n$.

On the fixed grid DOMI takes 5.3--10.4 s from $n=600$ to $n=76{,}800$, with a peak memory below 200 MB. Distance correlation takes 760 s and the global copula statistic 383 s at $n=9600$, each with 2.2 GB, and both exceed the memory limit at $n=38{,}400$; the copula test of [10] exceeds 30 minutes at $n=9600$ and the Gram form at $n=2400$. HSIC and Spearman take 12.4 and 1.8 s at $n=76{,}800$; on the correlation-free alternatives their localised power is lower (Table~\ref{tab:B11}).

\subsection{The Gram form and the feature dimension}\label{sec:B10}

Table~\ref{tab:B3} compares the computational forms of Section 4.8 on the
scenarios of the upper block of Table 1, under the calibration of Section 5 at 500
replicates, with S3 and S4 calibrated against their own pre-change
regimes. The Gram form uses the Gaussian kernel at the median-heuristic
bandwidth that the random features approximate, on every row of each
segment; the order-2 row replaces the von Neumann
entropy by the Rényi entropy of order 2 [14,15]. All entries are
studentised. The Gram form at \(\alpha=1\) is at least as powerful
as the random-feature form at \(D=8\) in every power column. At \(D=16\) the
random-feature form closes part of the difference on S2 at \(a=0.7\),
all of it at \(a=0.9\), and none of it on S4. The order-2 variant is weaker than the von Neumann one on S1,
0.48 against 0.82, and within 7 percentage points of it elsewhere, where it is higher on S2 at \(a=0.7\) (0.62) and on S4 at \(r=0.85\) (0.61); its detection rate under M1 is 0.05.

Capping the cost of the Gram form by subsampling rows is not a neutral
approximation. On four cells, with 200 null and 200 alternative
replicates each and the null split evenly between the studentising
moments and the threshold, the Gram form is computed both on every row
and on a fresh subsample of at most 256 rows per segment, on the same
replicates. On G2 at \(\tau=0.5\) and G4 at \(a=0.5\) both versions
reject in every replicate, but the subsampled version localises the
change in 33.0\% and 43.0\% of them, against 100.0\% and 99.5\% on full
segments; the random-feature form localises in 99.0\% of both. On S2 at
\(a=0.7\) the subsample lowers the localised power from 0.505 to 0.410,
and on G1 at \(\nu=2\) it raises it from 0.145 to 0.265. With only 100
null replicates for the threshold, these powers are not comparable with
Table 1; only the contrast within each cell is. The Gram form
in Table 1 and in Tables~\ref{tab:B3} and~\ref{tab:B4} is computed on full
segments.

\begin{table*}[!t]
\centering
\caption{The two computational forms and the feature dimension: localised power at a false-alarm rate of 0.05 (500 replicates, studentised variant) and detection rate under the marginal-only control M1 at $s=2$. Calibration as in Table 1. A dash marks a configuration not run.}
\begin{tabular}{l ccccccc}
\toprule
Form & S1 $r{=}.35$ & S2 $a{=}.7$ & S2 $a{=}.9$ & S3 $\tau{=}.5$ & S4 $r{=}.7$ & S4 $r{=}.85$ & M1 (FA) \\
\midrule
DOMI, random features, $D{=}8$ & 0.60 & 0.35 & 0.87 & 0.01 & 0.09 & 0.55 & 0.05 \\
DOMI, random features, $D{=}16$ & -- & 0.43 & 0.98 & 0.02 & 0.07 & 0.53 & -- \\
DOMI, Gram form ($\alpha{=}1$) & 0.82 & 0.58 & 0.98 & 0.01 & 0.12 & 0.55 & 0.04 \\
Gram form, $\alpha{=}2$ & 0.48 & 0.62 & 0.97 & 0.01 & 0.12 & 0.61 & 0.05 \\
\bottomrule
\end{tabular}
\label{tab:B3}
\end{table*}

\subsection{The non-Gaussian suite in full}\label{sec:B11}

Table~\ref{tab:B4} reports every statistic of the non-Gaussian suite of Section
5 in both of its variants, raw weighted global and studentised, at every
level. The lower block of Table 1 of the main text prints, for each statistic, the variant
selected by the reporting rule of Section 5. The studentised variant is selected for every statistic except Spearman and the copula statistic, for which the raw variant is the more powerful on average.

\begin{table*}[!t]
\centering
\caption{Non-Gaussian suite G1--G6 in full: localised power at a false-alarm rate of 0.05 (500 replicates; $|\hat\tau-\tau|\le 30$), each entry raw/studentised; \(\tau\) in G2--G3 is Kendall's \(\tau\), not the break location. CvM is the copula test of [10] with ranks recomputed within each segment; ``CvM, global'' computes it from global pseudo-observations. Margins are standard normal throughout; the pre-change regime is independence.}
\footnotesize
\setlength{\tabcolsep}{3pt}
\begin{tabular}{l cccccccc}
\toprule
Scenario & DOMI & Gram $\alpha{=}1$ & Gram $\alpha{=}2$ & HSIC & dCor & Spearman & CvM & CvM, global \\
\midrule
G1, $\nu=8$ & 0.00/0.01 & 0.00/0.01 & 0.00/0.01 & 0.00/0.00 & 0.00/0.00 & 0.00/0.00 & 0.01/0.00 & 0.00/0.00 \\
G1, $\nu=4$ & 0.00/0.03 & 0.00/0.05 & 0.00/0.08 & 0.00/0.00 & 0.00/0.00 & 0.00/0.00 & 0.01/0.00 & 0.00/0.00 \\
G1, $\nu=2$ & 0.00/0.33 & 0.00/0.44 & 0.01/0.49 & 0.00/0.03 & 0.00/0.01 & 0.00/0.00 & 0.01/0.01 & 0.01/0.01 \\
G2, $\tau=0.2$ & 0.01/0.72 & 0.00/0.79 & 0.03/0.54 & 0.03/0.45 & 0.27/0.48 & 0.51/0.38 & 0.60/0.36 & 0.05/0.07 \\
G2, $\tau=0.35$ & 0.29/0.97 & 0.21/1.00 & 0.36/0.73 & 0.35/0.65 & 0.62/0.68 & 0.88/0.68 & 0.91/0.69 & 0.22/0.24 \\
G2, $\tau=0.5$ & 0.49/1.00 & 0.48/1.00 & 0.61/0.89 & 0.51/0.76 & 0.82/0.84 & 0.99/0.82 & 0.99/0.86 & 0.50/0.54 \\
G3, $\tau=0.2$ & 0.01/0.72 & 0.01/0.86 & 0.06/0.63 & 0.10/0.65 & 0.29/0.56 & 0.55/0.52 & 0.63/0.57 & 0.04/0.04 \\
G3, $\tau=0.35$ & 0.30/0.97 & 0.27/1.00 & 0.38/0.89 & 0.38/0.88 & 0.64/0.78 & 0.92/0.89 & 0.94/0.91 & 0.22/0.20 \\
G3, $\tau=0.5$ & 0.47/0.99 & 0.45/1.00 & 0.60/0.98 & 0.46/0.96 & 0.82/0.87 & 0.98/0.98 & 0.99/1.00 & 0.49/0.46 \\
G4, $a=0.5$ & 0.07/0.97 & 0.02/0.98 & 0.05/0.83 & 0.08/0.77 & 0.01/0.52 & 0.01/0.01 & 0.36/0.11 & 0.00/0.00 \\
G4, $a=0.7$ & 0.31/1.00 & 0.25/1.00 & 0.33/0.97 & 0.30/0.91 & 0.35/0.86 & 0.01/0.01 & 0.84/0.58 & 0.01/0.01 \\
G4, $a=0.9$ & 0.58/1.00 & 0.56/1.00 & 0.63/1.00 & 0.50/1.00 & 0.68/0.99 & 0.01/0.01 & 0.99/0.87 & 0.05/0.06 \\
G5, $a=0.5$ & 0.00/0.94 & 0.00/0.98 & 0.03/0.90 & 0.01/0.77 & 0.01/0.49 & 0.00/0.00 & 0.22/0.12 & 0.00/0.01 \\
G5, $a=0.7$ & 0.15/1.00 & 0.20/1.00 & 0.32/0.99 & 0.17/0.97 & 0.22/0.92 & 0.01/0.01 & 0.80/0.67 & 0.01/0.02 \\
G5, $a=0.9$ & 0.41/1.00 & 0.49/1.00 & 0.53/1.00 & 0.36/0.99 & 0.56/0.99 & 0.01/0.01 & 0.97/0.97 & 0.02/0.04 \\
G6, $b=0.3$ & 0.00/0.49 & 0.00/0.65 & 0.00/0.36 & 0.00/0.16 & 0.00/0.18 & 0.00/0.00 & 0.10/0.02 & 0.00/0.01 \\
G6, $b=0.5$ & 0.10/0.86 & 0.09/0.91 & 0.14/0.78 & 0.04/0.60 & 0.05/0.73 & 0.01/0.00 & 0.59/0.27 & 0.01/0.02 \\
G6, $b=0.8$ & 0.38/0.94 & 0.37/0.97 & 0.46/0.94 & 0.32/0.74 & 0.56/0.95 & 0.01/0.00 & 0.95/0.71 & 0.03/0.04 \\
\bottomrule
\end{tabular}
\label{tab:B4}
\end{table*}

\subsection{Figures and tables supporting the main text}\label{sec:B12}

Figure~\ref{fig:B3} accompanies Section 3.4, Table~\ref{tab:B5} Section 4.3 and Section~\ref{sec:B17}, Table~\ref{tab:B6}
Sections 6.1 and 6.2, Table~\ref{tab:B7} and Figure~\ref{fig:B4} Section 6.3, Table~\ref{tab:B8}
Section 6.5, Figure~\ref{fig:B5} Section 6.8, and Table~\ref{tab:B9} Section~\ref{sec:B13}.

\begin{figure*}[!t]
\centering
\includegraphics[width=\textwidth,height=0.42\textheight,keepaspectratio]{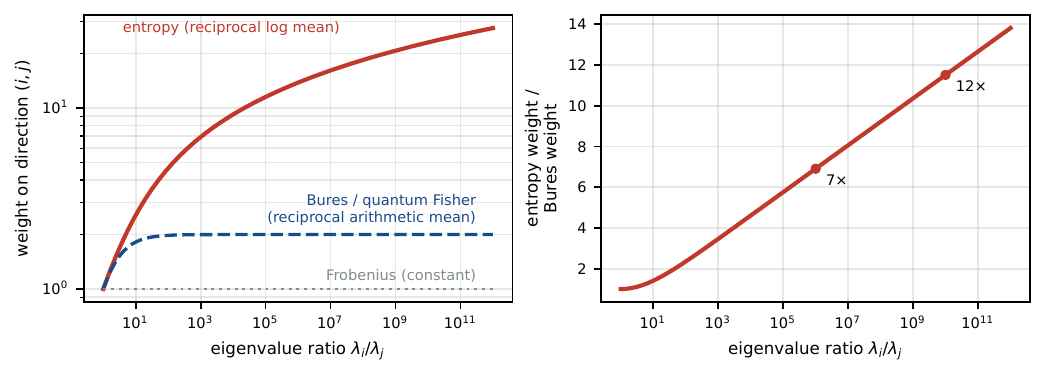}
\caption{The weight each functional places on the eigenvalue pair $(\lambda_i,\lambda_j)$ of a perturbed density operator, against the ratio $\lambda_i/\lambda_j$, the larger eigenvalue fixed at 1 (Section 3.4; the local expansion is derived in Section~\ref{sec:A9}). Left: the entropic kernel is the reciprocal logarithmic mean, the Bures or quantum Fisher kernel the reciprocal arithmetic mean, and the Frobenius kernel a constant; only the entropic one is unbounded in the ratio. Right: the entropic weight relative to the Bures weight, which grows like $\tfrac12\log(\lambda_i/\lambda_j)$ --- a 7-fold advantage at a ratio of $10^{6}$ and 12-fold at $10^{10}$, so the amplification grows only logarithmically in the ratio.}
\label{fig:B3}
\end{figure*}

\begin{table*}[!t]
\centering
\caption{Calibration under serial dependence: rejection rate with no change point present, at a nominal level of 0.05 (200 replicates, $D=8$, $K=99$). Pair permutation is exact only under exchangeability of the pairs; block permutation is exact under exchangeability of the blocks, which a serially dependent series satisfies only approximately. Studentisation is symmetric over all $K+1$ curves, as in P1 and P5. Under serial dependence pair permutation over-rejects, and block permutation at either block length is within about one standard error (1.5 percentage points) of nominal; under serial independence both schemes are valid.}
\begin{tabular}{l ccc}
\toprule
Null design & pair & block $b{=}20$ & block $b{=}50$ \\
\midrule
serially independent & 0.035 & 0.035 & 0.055 \\
AR(1) margins, $\phi = 0.6$, constant Gaussian copula & 0.245 & 0.050 & 0.035 \\
GARCH(1,1) volatility clustering & 0.085 & 0.050 & 0.045 \\
\bottomrule
\end{tabular}
\label{tab:B5}
\end{table*}

\begin{table*}[!t]
\centering
\caption{Joint-distribution references on the designs of the upper block of Table 1: localised power at a false-alarm rate of 0.05 (500 replicates) and detection rate under the marginal-only control M1 at $s=2$, calibrated as in Table 1. The three joint-distribution statistics are not specific to dependence and reject under M1 in every replicate; the copula statistic computed from global pseudo-observations rejects in 81\%.}
\begin{tabular}{l ccccccc}
\toprule
Method & S1 $r{=}.35$ & S2 $a{=}.7$ & S2 $a{=}.9$ & S3 $\tau{=}.5$ & S4 $r{=}.7$ & S4 $r{=}.85$ & M1 (FA) \\
\midrule
Joint-state Holevo (Section 4.2) & 0.38 & 0.09 & 0.81 & 0.09 & 0.22 & 0.87 & 1.00 \\
MMD (median bandwidth) & 0.09 & 0.04 & 0.34 & 0.01 & 0.02 & 0.13 & 1.00 \\
Gaussian likelihood ratio & 0.57 & 0.01 & 0.02 & 0.01 & 0.00 & 0.00 & 1.00 \\
Copula CvM, global pseudo-observations & 0.10 & 0.00 & 0.02 & 0.01 & 0.00 & 0.00 & 0.81 \\
\bottomrule
\end{tabular}
\label{tab:B6}
\end{table*}

\begin{table*}[!t]
\centering
\caption{Multiple-break recovery (scenario MB, 200 replicates). Penalties (for e.divisive, its significance level) are calibrated to a nominal 5\% of null replicates with any detection; the first column gives the achieved rate. $P(\hat K{=}3)$ is the probability of selecting the correct number of breaks; ``All three'' is the probability that the three sorted estimates lie within 90 of the three true breaks; Hausdorff distance (median) and adjusted Rand index (ARI, mean) at $a=0.9$. The last column is the rate of any detection on a marginal-only design (the scale of $X$ switching $1\to1.6\to1\to1.6$, no dependence change). ``+ DOMI re-test'' keeps a candidate only if the DOMI single-break test on a window around it rejects (pair permutation, $K=99$, Benjamini--Hochberg $q\le0.10$ within the replicate). Rank-Gaussian PELT and kernel binary segmentation (BinSeg-MMD), from the same replicates: null rate 0.02 and 0.05, $P(\hat K{=}3)$ 0.01 and 0.115 at $a=0.9$ (all three within 90: 0.000 and 0.050) and 0.01 and 0.015 at $a=0.8$; BMCTC with its calibrated split rule ($\alpha=1.01$): $P(\hat K{=}3)$ 0.19 and all three within 90: 0.115 at $a=0.9$ (Section~\ref{sec:B13}).}
\footnotesize
\begin{tabular}{l ccccccc}
\toprule
Method & Null & $P(\hat K{=}3)$, $a{=}.9$ & All three & Hausdorff & ARI & $P(\hat K{=}3)$, $a{=}.8$ & Marginal-only \\
\midrule
Holevo partitioning & 0.000 & 0.475 & 0.450 & 665 & 0.62 & 0.010 & 0.975 \\
KCP, raw values & 0.050 & 0.175 & 0.115 & 890 & 0.44 & 0.025 & 1.000 \\
KCP, ranks & 0.050 & 0.030 & 0.015 & 900 & 0.33 & 0.010 & 0.690 \\
e.divisive, raw values & 0.055 & 0.015 & 0.000 & 1800 & 0.07 & 0.005 & 0.865 \\
e.divisive, ranks & 0.060 & 0.000 & 0.000 & 1800 & 0.12 & 0.000 & 0.145 \\
Holevo partitioning + DOMI re-test & 0.000 & 0.470 & 0.450 & 665 & 0.62 & 0.010 & 0.130 \\
KCP, ranks + DOMI re-test & 0.005 & 0.015 & 0.015 & 900 & 0.32 & 0.000 & 0.090 \\
e.divisive, ranks + DOMI re-test & 0.005 & 0.005 & 0.000 & 1800 & 0.12 & 0.000 & 0.020 \\
\bottomrule
\end{tabular}
\label{tab:B7}
\end{table*}

\begin{figure*}[!t]
\centering
\includegraphics[width=\textwidth,height=0.42\textheight,keepaspectratio]{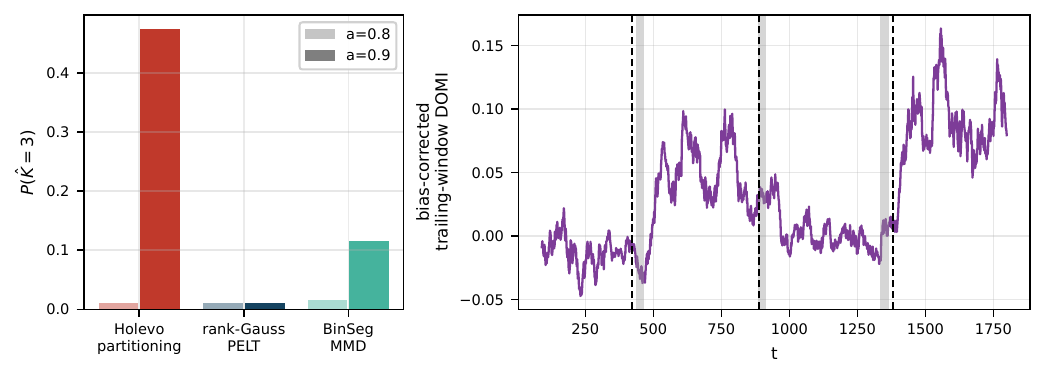}
\caption{Multiple breaks (scenario MB, three breaks). Left: the probability of selecting the correct number of breaks, $K=3$, for Holevo partitioning against rank-Gaussian PELT and kernel binary segmentation, at the two dependence levels (lighter bars $a=0.8$, darker bars $a=0.9$). Right: the first replicate at the stronger level, showing the 90-sample trailing-window DOMI, bias-corrected by one permuted copy and therefore possibly negative, with the true break locations in grey and the Holevo partitioning estimates at the calibrated penalty dashed.}
\label{fig:B4}
\end{figure*}

\begin{table}[!t]
\centering
\caption{The fully data-driven procedure: false-alarm rate and localised power under pair-permutation calibration ($K=99$ permutations, 100 replicates), with the threshold estimated from the observed series alone rather than from an oracle Monte Carlo null. M1 is a margins-only change, so its power entry counts localised false detections.}
\begin{tabular}{l cc}
\toprule
Pair-permutation test & false-alarm rate & power \\
\midrule
S1, $r{=}0.35$ & 0.06 & 0.41 \\
S2, $a{=}0.7$ & 0.03 & 0.20 \\
S2, $a{=}0.9$ & 0.03 & 0.67 \\
S4, $r{=}0.7$ & 0.06 & 0.53 \\
S4, $r{=}0.85$ & 0.05 & 0.66 \\
M1, $s{=}1.6$ & 0.03 & 0.00 \\
\bottomrule
\end{tabular}
\label{tab:B8}
\end{table}

\begin{figure}[!t]
\centering
\includegraphics[width=\columnwidth,height=0.42\textheight,keepaspectratio]{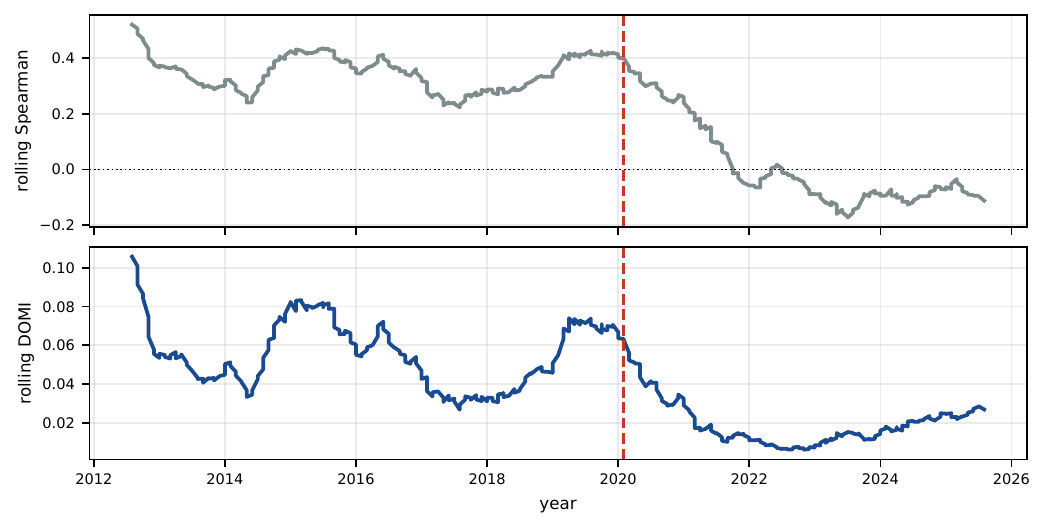}
\caption{S\&P 500 return--yield-change coupling, 2011--2026. Upper: the rolling Spearman correlation on two-year windows. Lower: the rolling DOMI on the same windows. The dashed vertical line in both panels is the single weekly Holevo partitioning break that the block-permutation calibration returns (2020-02, the pandemic onset); the breaks that the pair-permutation calibration reported are not shown, the diagnostic having rejected that calibration here.}
\label{fig:B5}
\end{figure}

\subsection{A matrix-information change-point baseline}\label{sec:B13}

The closest prior application of matrix-based information to change
detection is BMCTC, the moving-cut total-correlation method of [37]: the matrix-based R\'enyi total correlation on
a moving window, segmented by the Bernaola--Galv\'an procedure. We
reimplemented it from its description, with choices the description
does not fix stated in the released script (\(\alpha=1.01\) and \(\alpha=2\); Gaussian kernel with Silverman's width on values standardised within each window; window 60, step 1). Table~\ref{tab:B9} reports
it on the designs of Table 1.

With its own significance rule it
rejects in every null replicate and under M1, because the moving-window
sequence is strongly autocorrelated while the rule treats it as
independent, and on scenario MB it places about 38 breaks per series.

Under the calibration of Table 1 it is dependence-specific (a
false-alarm rate of 0.05--0.07 under M1) and less powerful than either
form of DOMI on S1, S2, G2, G4 and G6. It is more powerful on S4, at
0.80--0.94 against 0.55 at \(r=0.85\) and 0.23--0.57 against 0.09 and
0.12 at \(r=0.7\), and comparable on G1 at \(\nu=2\) (0.21--0.39 against 0.33 and 0.44). Across a $2\times2$ variation of BMCTC's input (raw values or whole-window ranks) and kernel rule (its own window standardisation with Silverman's width, or DOMI's median-heuristic width), rank input lowers its localised power on S4 at $r=0.7$ to 0.09--0.23, near that of DOMI, while it remains ahead at $r=0.85$ (0.67--0.83), and its false-alarm rate under M1 varies with both factors (0.04--0.99; Table~\ref{tab:B18}). The source of the remaining difference is not established. The signs in
that design are drawn independently at every time point, so the moving
window has no local regime to exploit; at a quarter of the median
bandwidth DOMI reaches 0.71 (Table~\ref{tab:B2}). With the calibrated split rule, BMCTC
returns exactly three breaks on scenario MB in 19\% of replicates (all three within 90: 11.5\%) at
\(a=0.9\), against 47.5\% for Holevo partitioning, and in 5.5\% (\(\alpha=1.01\)) and 7.5\% (\(\alpha=2\)) at \(a=0.8\).

\begin{table}[!t]
\centering
\caption{BMCTC (our reimplementation of [37]) on the designs of Table 1: localised power at a false-alarm rate of 0.05 under the calibration of Table 1 (500 replicates; raw and studentised at $\alpha=1.01$, raw at $\alpha=2$), and the detection rate with the published Bernaola--Galv\'an rule ($P_0=0.95$), which is not calibrated. The M1 row gives false-alarm rates.}
\footnotesize
\setlength{\tabcolsep}{3pt}
\begin{tabular}{l cccc}
\toprule
Setting & raw & stud. & $\alpha{=}2$ & BG rule \\
\midrule
S1 $r{=}.35$ & 0.11 & 0.17 & 0.08 & 1.00 \\
S2 $a{=}.7$ & 0.22 & 0.24 & 0.26 & 1.00 \\
S2 $a{=}.9$ & 0.67 & 0.71 & 0.69 & 1.00 \\
S3 $\tau{=}.5$ & 0.02 & 0.02 & 0.02 & 1.00 \\
S4 $r{=}.7$ & 0.33 & 0.23 & 0.57 & 1.00 \\
S4 $r{=}.85$ & 0.86 & 0.80 & 0.94 & 1.00 \\
G1 $\nu{=}2$ & 0.28 & 0.21 & 0.39 & 1.00 \\
G2 $\tau{=}.2$ & 0.13 & 0.10 & 0.13 & 1.00 \\
G4 $a{=}.5$ & 0.26 & 0.26 & 0.25 & 1.00 \\
G6 $b{=}.3$ & 0.06 & 0.05 & 0.04 & 1.00 \\
G6 $b{=}.5$ & 0.40 & 0.38 & 0.22 & 1.00 \\
M1 $s{=}2$ (FA) & 0.06 & 0.07 & 0.05 & 1.00 \\
\bottomrule
\end{tabular}
\label{tab:B9}
\end{table}

\subsection{An exploratory resampling check of the whole two-stage procedure}\label{sec:B14}

Stage-two p-values are nominal because stage one selects the windows
(Section 4.7). Permuting the whole procedure would address the selection: every permuted record passes through stage one and stage two unchanged, and the observed statistic is compared with the permuted
ones. We permute the twelve-week super-blocks of stage one within
calendar seasons. We hold the penalty and the bias correction at their observed values instead of recomputing them on every permuted record, which cuts one permuted run from over two minutes (136 s for stage one alone, Section~\ref{sec:B9}) to 1.7 s on one core (1.68--1.69 s over five runs on the host of Section~\ref{sec:B9}). Because these quantities come from the observed record, the observed and permuted statistics are not computed by the same rule, so this is an exploratory resampling check: its finite-sample validity and its adjustment for selection are not established.

On synthetic records ($n=600$, 200 replicates, 99 permutations each) the check rejects in 4.5\% of
replicates under a stationary null with AR(1) margins ($\phi=0.6$) and
constant dependence, and in 75\% against a coupling change of type S2 at
$a=0.9$, where the nominal stage two rejects in 53\%. Under a
marginal-only change (M1, $s=1.6$) it rejects in 57.5\%, against 6.5\%
for the nominal stage two: permuting blocks mixes the marginal regimes,
so stage one fires on 93\% of observed records and 8\% of permuted
ones. Its null is the absence of any change, and it is not specific to dependence.

On the weather records (199 permutations) the check built on the stage-two statistic gives an exploratory $p=0.105$ (0.18 with a Bonferroni combination). That statistic is the maximum over the candidates of each analysis, combined across analyses through the smallest p-value and calibrated by the same permutations. Stage one returns a candidate in 23 of the 36 analyses, against a mean of 0.91 and a maximum of 8 on the permuted records (exploratory tail fraction 0.005); this descriptive contrast uses the observed stage-one calibration and is not a level-controlled rejection of season-stratified block stationarity. By the M1 result neither establishes a dependence change.

\subsection{Detection, localisation and calibration of the single-break comparison}\label{sec:B15}

Table~\ref{tab:B10} gives the detection rates behind Table 1, on the same
replicates and in the same variants.

Table~\ref{tab:B11} repeats ten cells under the pair-permutation calibration of Section 4.3
($K=99$, 200 replicates),
in which the studentised statistics take their moments over the observed and replica curves together; Spearman and the copula statistics are raw, as in Table 1. On the three no-change designs every statistic rejects in 3.0--8.0\% of replicates. On the two marginal-only controls every statistic rejects in 3.0--7.5\% of replicates except the copula statistic computed from global pseudo-observations, which rejects in 83.5\% at M1 $s=2$.

Table~\ref{tab:B12} moves the break to $\tau=150$
and $\tau=450$ under the calibration of Table 1; the rows at $\tau=300$
reproduce Table 1. Localisation fails for every statistic when the
dependent segment is the longer one ($\tau=150$ in S2, G4 and G6, where
dependence begins at the break; $\tau=450$ in S4, where it ends), with
the studentised maximiser near the centre: on S2 at $\tau=150$ the
median estimate is 299 for DOMI and 294 for the Gram form.

\begin{table}[!t]
\centering
\caption{Detection rates without the localisation requirement for the settings of Table 1 (500 replicates, same variants and calibration). Bold marks the largest rate in a row at two decimals (ties bolded); the M1 row gives false-alarm rates.}
\footnotesize
\setlength{\tabcolsep}{3pt}
\begin{tabular}{l cccccc}
\toprule
Setting & DOMI & Gram & HSIC & dCor & Spear. & CvM \\
\midrule
S1, $r{=}0.35$ & 0.99 & 0.99 & 0.99 & \textbf{1.00} & 0.94 & 0.96 \\
S2, $a{=}0.7$ & 0.58 & \textbf{0.73} & 0.16 & 0.09 & 0.12 & 0.13 \\
S2, $a{=}0.9$ & \textbf{1.00} & \textbf{1.00} & 0.90 & 0.18 & 0.12 & 0.19 \\
S3, $\tau{=}0.5$ & 0.09 & 0.11 & 0.15 & 0.06 & 0.06 & \textbf{0.52} \\
S4, $r{=}0.7$ & 0.36 & \textbf{0.45} & 0.14 & 0.10 & 0.03 & 0.02 \\
S4, $r{=}0.85$ & 0.99 & \textbf{1.00} & 0.70 & 0.56 & 0.02 & 0.04 \\
M1, $s{=}2$ (FA) & 0.05 & 0.04 & 0.07 & 0.04 & 0.03 & 0.05 \\
\midrule
G1, $\nu{=}8$ & 0.04 & \textbf{0.06} & 0.04 & 0.04 & 0.05 & 0.03 \\
G1, $\nu{=}4$ & 0.12 & \textbf{0.20} & 0.04 & 0.04 & 0.05 & 0.05 \\
G1, $\nu{=}2$ & 0.70 & \textbf{0.86} & 0.20 & 0.09 & 0.08 & 0.08 \\
G2, $\tau{=}0.2$ & \textbf{0.98} & 0.97 & 0.94 & 0.97 & 0.81 & 0.89 \\
G2, $\tau{=}0.35$ & \textbf{1.00} & \textbf{1.00} & \textbf{1.00} & \textbf{1.00} & \textbf{1.00} & \textbf{1.00} \\
G2, $\tau{=}0.5$ & \textbf{1.00} & \textbf{1.00} & \textbf{1.00} & \textbf{1.00} & \textbf{1.00} & \textbf{1.00} \\
G3, $\tau{=}0.2$ & 0.97 & \textbf{0.98} & 0.97 & \textbf{0.98} & 0.85 & 0.90 \\
G3, $\tau{=}0.35$ & \textbf{1.00} & \textbf{1.00} & \textbf{1.00} & \textbf{1.00} & \textbf{1.00} & \textbf{1.00} \\
G3, $\tau{=}0.5$ & \textbf{1.00} & \textbf{1.00} & \textbf{1.00} & \textbf{1.00} & \textbf{1.00} & \textbf{1.00} \\
G4, $a{=}0.5$ & \textbf{1.00} & \textbf{1.00} & 0.97 & 0.74 & 0.06 & 0.54 \\
G4, $a{=}0.7$ & \textbf{1.00} & \textbf{1.00} & \textbf{1.00} & 0.99 & 0.08 & 0.98 \\
G4, $a{=}0.9$ & \textbf{1.00} & \textbf{1.00} & \textbf{1.00} & \textbf{1.00} & 0.09 & \textbf{1.00} \\
G5, $a{=}0.5$ & 0.99 & \textbf{1.00} & 0.91 & 0.61 & 0.05 & 0.39 \\
G5, $a{=}0.7$ & \textbf{1.00} & \textbf{1.00} & \textbf{1.00} & 0.98 & 0.07 & 0.96 \\
G5, $a{=}0.9$ & \textbf{1.00} & \textbf{1.00} & \textbf{1.00} & \textbf{1.00} & 0.06 & \textbf{1.00} \\
G6, $b{=}0.3$ & 0.85 & \textbf{0.92} & 0.42 & 0.36 & 0.06 & 0.19 \\
G6, $b{=}0.5$ & \textbf{1.00} & \textbf{1.00} & 0.99 & 0.95 & 0.07 & 0.74 \\
G6, $b{=}0.8$ & \textbf{1.00} & \textbf{1.00} & \textbf{1.00} & \textbf{1.00} & 0.07 & \textbf{1.00} \\
\bottomrule
\end{tabular}
\label{tab:B10}
\end{table}

\begin{table*}[!t]
\centering
\caption{Rejection rate / localised power under the pair-permutation calibration ($K=99$, 200 replicates, rejection at $p\le0.05$, tolerance 30). Kernel and distance statistics studentised, Spearman and the copula statistics raw, as in Table 1. The first three rows give rejection rates under no change; the M1 and M2 rows change only the margins. CvM is the copula test of [10] with ranks taken within each segment; CvM, global is the same statistic from global pseudo-observations.}
\footnotesize
\setlength{\tabcolsep}{4pt}
\begin{tabular}{l ccccccc}
\toprule
Cell & DOMI & Gram & HSIC & dCor & Spearman & CvM & CvM, global \\
\midrule
no change, S2 null & 0.065 & 0.060 & 0.060 & 0.060 & 0.050 & 0.040 & 0.030 \\
no change, Gaussian copula $\tau{=}0.5$ & 0.045 & 0.050 & 0.055 & 0.030 & 0.050 & 0.080 & 0.060 \\
no change, sign-mixed $r{=}0.85$ & 0.040 & 0.045 & 0.040 & 0.040 & 0.035 & 0.045 & 0.035 \\
M1 $s{=}2$ (margins only) & 0.035/0.005 & 0.030/0.015 & 0.030/0.000 & 0.065/0.000 & 0.055/0.010 & 0.045/0.015 & 0.835/0.690 \\
M2 $a{=}0.9$ (margins only) & 0.045/0.005 & 0.050/0.000 & 0.065/0.000 & 0.075/0.005 & 0.055/0.005 & 0.045/0.020 & 0.055/0.005 \\
S1 $r{=}0.35$ & 0.870/0.425 & 0.875/0.520 & 0.755/0.250 & 0.840/0.275 & 0.950/0.660 & 0.950/0.705 & 0.220/0.075 \\
S2 $a{=}0.7$ & 0.435/0.245 & 0.500/0.315 & 0.115/0.035 & 0.090/0.020 & 0.050/0.005 & 0.090/0.025 & 0.080/0.010 \\
S3 $\tau{=}0.5$ & 0.080/0.010 & 0.095/0.010 & 0.055/0.000 & 0.075/0.010 & 0.080/0.015 & 0.555/0.380 & 0.050/0.010 \\
S4 $r{=}0.85$ & 1.000/0.690 & 1.000/0.760 & 0.980/0.565 & 0.630/0.490 & 0.040/0.000 & 0.090/0.030 & 0.065/0.005 \\
G6 $b{=}0.5$ & 0.995/0.675 & 1.000/0.730 & 0.880/0.435 & 0.915/0.490 & 0.065/0.005 & 0.715/0.580 & 0.075/0.010 \\
\bottomrule
\end{tabular}
\label{tab:B11}
\end{table*}

\begin{table*}[!t]
\centering
\caption{Detection rate / localised power with the break at $\tau=150$, 300 and 450 ($n=600$, 500 replicates, calibration of Table 1, studentised forms).}
\footnotesize
\begin{tabular}{l c cccc}
\toprule
Cell & $\tau$ & DOMI & Gram & HSIC & dCor \\
\midrule
S2 $a{=}0.9$ & 150 & 0.98/0.00 & 1.00/0.00 & 0.72/0.00 & 0.07/0.00 \\
 & 300 & 1.00/0.87 & 1.00/0.98 & 0.90/0.66 & 0.18/0.07 \\
 & 450 & 0.94/0.69 & 0.97/0.76 & 0.51/0.28 & 0.17/0.05 \\
S4 $r{=}0.85$ & 150 & 0.89/0.85 & 0.95/0.93 & 0.38/0.32 & 0.27/0.15 \\
 & 300 & 0.99/0.55 & 1.00/0.55 & 0.70/0.39 & 0.56/0.24 \\
 & 450 & 0.93/0.01 & 0.96/0.00 & 0.37/0.00 & 0.18/0.00 \\
G4 $a{=}0.5$ & 150 & 0.97/0.00 & 0.97/0.00 & 0.90/0.01 & 0.36/0.00 \\
 & 300 & 1.00/0.97 & 1.00/0.98 & 0.97/0.77 & 0.74/0.52 \\
 & 450 & 0.89/0.62 & 0.90/0.62 & 0.69/0.39 & 0.25/0.08 \\
G6 $b{=}0.5$ & 150 & 0.99/0.00 & 1.00/0.00 & 0.91/0.00 & 0.54/0.00 \\
 & 300 & 1.00/0.86 & 1.00/0.91 & 0.99/0.60 & 0.95/0.73 \\
 & 450 & 0.98/0.81 & 0.99/0.86 & 0.72/0.49 & 0.43/0.22 \\
\bottomrule
\end{tabular}
\label{tab:B12}
\end{table*}

\emph{Marginal-shape control for S2.} In S2 the change also alters the shape of the \(Y\) margin: its variance stays 1 but its fourth moment rises to \(3+6a^4\). Scenario M2 reproduces that marginal change exactly while keeping \(X\) and \(Y\) independent, the multiplier \(|X'|\) being drawn independently of \(X\); it is run with the protocol of Table 1 (500 null and 500 alternative replicates per level). At \(a=0.5\), 0.7 and 0.9 the detection rate is 0.068, 0.074 and 0.056 for DOMI, 0.050, 0.084 and 0.054 for HSIC and 0.064, 0.090 and 0.072 for distance correlation, and the localised power of DOMI is at most 0.014, against 0.35 and 0.87 on S2 at \(a=0.7\) and 0.9. The joint-state Holevo statistic, which is not dependence-specific, detects 0.084, 0.050 and 0.288.

\subsection{Calibration by the boundary limit law}\label{sec:B16}

This section tests the permutation-free calibration of Section 4.6: critical values simulated from the limit law of P6-R are compared with the exact null quantile, and the level and power of the two limit-calibrated scans are measured. The limit theory of Section~\ref{sec:A12} is stated for a fixed bandwidth; the scans below use the deployed subsample bandwidth, so their agreement with the exact law is an empirical check, not a consequence of the theorem. The scan statistic here is $\max_\lambda \sqrt{\lambda(1-\lambda)}\,|n\hat I_L(\lambda)-n\hat I_R(\lambda)|$ on 161 split fractions from 0.10 to 0.90, with the deployed inputs (ranks over the whole series, $D=8$, median-heuristic bandwidth). Under independence the rank vectors are independent uniform permutations, so the exact null law depends only on $n$, $D$, the frequency draw and the grid; Table~\ref{tab:B13} gives its 95\% quantile from 2000 simulated series, the median over null replicates of the critical value obtained by simulating the limit functional, and the levels of the two calibrated scans (2000 replicates). The simulation discards directions of the moment matrices whose eigenvalue is below $10^{-9}$ of the largest and keeps the limit weights that carry all but $10^{-6}$ of their sum; at the tested cuts $10^{-12}$, $10^{-9}$ and $10^{-6}$ the bias constant differs by less than 0.01\% (one replicate at $n=600$; the critical values were not recomputed at the other cuts). The weights sum to 3.14 (median over replicates; the constant depends on the frequency draw and on the data-driven bandwidth), against a mean of $n\hat I$ of 3.10--3.22 over the four $n$.

The unstudentised scan places 69--72\% of its null maxima within 0.015 of the ends of the grid, where the short segment's upward bias is largest, and it localises poorly; the scan studentised by the moments of the limit process is slightly liberal below $n=38{,}400$ (levels up to 0.059), and the more powerful of the two: on an S2-type emergence at $a=0.7$ and $n=600$ it rejects in 64\% of 500 replicates, against 17\% for the unstudentised scan and 48\% for the permutation-calibrated scan with $K=99$.

Under the marginal-only change M1, which lies outside the null the calibration covers, both stay within one standard error of the nominal level or below it (0.040--0.054).

\begin{table*}[!t]
\centering
\caption{Calibration of the scan at the independence boundary by simulating its limit (P6-R): exact Monte Carlo 95\% quantile [95\% interval], median simulated critical value, level at 0.05 under independence (H0) and under the marginal-only change M1 for the unstudentised (raw) and limit-studentised (stud.) scans, and single-core wall-clock seconds per test against $K=99$ permutations.}
\footnotesize
\begin{tabular}{r c c cc cc cc}
\toprule
$n$ & Exact MC & Limit & H0 raw & H0 stud. & M1 raw & M1 stud. & Limit (s) & Perm. (s) \\
\midrule
600 & 16.63 [16.11, 17.05] & 16.64 & 0.052 & 0.058 & 0.046 & 0.054 & 0.69 & 5.78 \\
2{,}400 & 16.24 [16.01, 16.78] & 16.66 & 0.052 & 0.057 & 0.044 & 0.050 & 0.55 & 5.88 \\
9{,}600 & 16.87 [16.39, 17.19] & 16.66 & 0.051 & 0.059 & 0.046 & 0.044 & 0.77 & 7.48 \\
38{,}400 & 16.80 [16.49, 17.21] & 16.67 & 0.055 & 0.051 & 0.040 & 0.043 & 0.93 & 9.63 \\
\bottomrule
\end{tabular}
\label{tab:B13}
\end{table*}

\subsection{Calibration under serial dependence: block length, power and the adaptive protocol}\label{sec:B17}

Section 4.3 calibrates serially dependent series by block permutation, with an exchangeability diagnostic that chooses between pair and block permutation and sets the block length. The integrated autocorrelation time of a sequence is estimated as \(1+2\sum_{k=1}^{k^*}\hat r_k\), where \(\hat r_k\) is the lag-\(k\) sample autocorrelation and \(k^*\) the last lag before the first non-positive one (the initial positive sequence), with at most \(n/4\) lags. The block length is set from the margins' own sequences \(z_t\) and \(z_t^2\), so it does not account for cross-lagged dependence that extends beyond one block. The designs below use \(n=600\) and \(D=8\) at a nominal level of 0.05.

\emph{Block length.} Table~\ref{tab:B5} gives the level as a function of the block length. On the serially dependent designs block permutation restores it at both block lengths (0.035--0.050), where pair permutation rejects in 24.5\% of the AR(1) and 8.5\% of the GARCH replicates. Section~\ref{sec:A11} gives the trade-off: longer blocks carry less dependence across a boundary but leave a coarser permutation group, and the candidate grid does not depend on \(b\).

\emph{Power.} The power design has AR(1) margins with \(\phi=0.6\) and independent innovations, so the pair is independent before the break; a correlation-free dependence of type S2 appears in the innovations at the midpoint. This null differs from the constant Gaussian-copula AR(1) design of Table~\ref{tab:B5}: pair permutation rejects under it in 19\% of replicates, not 24.5\%. Block permutation with \(b=20\) or 50 rejects in 7.0\% of replicates at \(a=0.7\) and in 19--22\% at \(a=0.9\), against its own null rate of 6.5--8.5\%; at \(a=0.7\) the test has no power. With \(b=20\) the localised powers (break within 60 samples) are 2\% and 11\%, against 20\% and 67\% for independent series under pair permutation (Section 6.5) at the stricter tolerance of 30 samples. Pair permutation rejects in 25\% of replicates at \(a=0.7\), against 19\% under the null. A likely cause of the loss is a shortage of information, not a failure of calibration. An AR(1) with \(\phi=0.6\) has an effective sample size for the mean of \(n(1-\phi)/(1+\phi)=150\) at \(n=600\), or 75 points per segment, below the smallest segment of the sweep of Section~\ref{sec:B1} (150 points at \(n=300\)), where localised power has already fallen to 0.215.

\emph{Level of the adaptive protocol.} The protocol of Section 4.3 was run as Algorithm 1 with \(K=199\) (against \(K=99\) in Table~\ref{tab:B5}), studentising over all \(K+1\) curves. It rejects in 4.3\% of 720 serially independent replicates, and in 5.0\% and 5.5\% of 200 replicates with AR(1) margins under constant Gaussian-copula dependence and with GARCH margins. In the same run, fixed blocks with \(b=20\) reject in 6.0\% of the AR(1) constant-copula replicates, against 5.0\% in Table~\ref{tab:B5}. With AR(1) margins and independent innovations the protocol rejects in 9.0\% (standard error 2.0\%), against 6.5\% for either fixed block length in the same run, so choosing the block length from the data may add some inflation there.

\emph{Selection and block length.} On the three null designs (200 replicates, \(L=20\)) the protocol selects block permutation in 1\% of serially independent replicates, in 100\% of AR(1) replicates (the five statistics reject in 97.5--100\% of replicates individually), and in 96.5\% of GARCH replicates (through the scale statistics, which reject in 75--78\% of replicates individually), with median recommended block lengths \(\hat b=22\) under AR(1) and \(\hat b=21\) under GARCH. The rule \(b=\lceil 5\hat\tau_{\mathrm{int}}\rceil\) places them at the short end of the range of Table~\ref{tab:B5} by construction: it returns \(b=20\) for an AR(1) with \(\phi=0.6\), whose integrated autocorrelation time is 4. On the AR(1) constant-copula design the level at \(b=20\) is 5.0\% in Table~\ref{tab:B5} and 6.0\% in the protocol run, and the protocol itself, at a median \(\hat b=22\), rejects in 5.0\% (standard errors 1.5 to 1.7 percentage points), so on the constant-copula designs the recommended block length shows no inflation; under volatility clustering the level at \(b=20\) is 5.0\% in Table~\ref{tab:B5} and 4.5\% in the protocol run.

\subsection{The combined directional--omnibus statistic}\label{sec:B18}

Section 4.5 combines the directional and omnibus statistics as \(T=\max(T_{\mathrm{DOMI}},T_{\mathrm{Spearman}})\), both studentised maxima computed from the same pair-permuted replicas. This section measures what the combination gains on a monotone alternative and loses on correlation-free ones. Each configuration uses 100 null and 100 alternative replicates of length \(n=600\) and \(K=49\) pair permutations. The rates below are rejection rates, with no localisation requirement, so they exceed the localised powers of Table 1 and Table~\ref{tab:B8} for the same configurations.

On S1 at \(r=0.35\) the combination rejects in 83\% of replicates, against 87\% for Spearman and 77\% for DOMI. On S2 it equals DOMI at \(a=0.9\) (0.97) and falls below it at \(a=0.7\), 0.23 against 0.37 (Spearman 0.02). False-alarm rates under the null of each design are 0.01--0.03. On the daily 2021--2022 stock--bond window of Section 6.8 the single-break test was run with 20 independent permutation draws of \(K=999\) under each calibration (Table~\ref{tab:B19}). Each draw is a permutation test in its own right; the \(p\)-values differ between draws because the studentising moments and the null maxima depend on the drawn copies, and with 25 blocks at \(b=20\) and 10 at \(b=50\) the spread is visible. Under block permutation DOMI and the combined statistic do not reject in any draw. Spearman is borderline at \(b=20\) (median 0.057, range 0.035--0.078) and does not reject at \(b=50\); HSIC gives 0.009--0.033 at \(b=20\) and 0.050--0.105 at \(b=50\). The copula statistic from global ranks gives at most 0.003 under every calibration, but it also responds to changes confined to the margins (0.835 under M1, Table~\ref{tab:B11}), so these values do not by themselves indicate a change in dependence.

\subsection{Simulation on empirical weather margins}\label{sec:B19}

This design combines a known dependence change with the marginal distributions of the weather records of Section 6.7. The margins are the hourly observations of station 105 (Gangneung) in 2018, temperature for $X$ and relative humidity for $Y$. The empirical quantile functions of the winter (December--February) and summer (June--August) observations, and of December and of February alone, define the marginal regimes; the recorded values, quantised to 0.1\,$^\circ$C and 1\%, are first spread uniformly within their resolution. The dependence is the latent construction of S4 at $r=0.85$ (sign-mixed, correlation zero) or independence, carried to the margins by the normal distribution function and the empirical quantile function. Each series has $n=1200$ with the change at 600. The scan uses the pair-permutation calibration and the statistics of Table~\ref{tab:B11} ($K=99$, 200 replicates, candidates 120--1080, localisation tolerance 60). In the rounded variant the observations are rounded back to the sensor resolution, which produces ties as in the record. Conditions C4--C5b, a marginal change within one season, were added after C0--C3b had been evaluated.

Under no change the rejection rates are 0.035--0.080 (Table~\ref{tab:B14}). With winter margins throughout, DOMI and HSIC detect the start or the end of the correlation-free dependence in every replicate and localise it in 56--68\% of replicates, and distance correlation detects it in 92--94.5\% and localises it in 67.5--72\%; Spearman and the global copula statistic detect it in at most 9\% of replicates and the copula test of [10] in 44.5--49\%. Rounding changes no rejection rate by more than 0.025 and no localised power by more than 0.08.

A change of the margins from winter to summer, with the dependence unchanged (C1), is detected in every replicate by every statistic. The ranks do not absorb a marginal change of this size: the copula test of [10], which ranks within each segment, also rejects in every replicate, and only the global copula statistic places the break at the marginal change; with the dependence change at the same point (C3a, C3b) no other statistic localises the break. A change of the margins from December to February (C4) raises the false-alarm rate of DOMI to 0.110 and that of HSIC to 0.140, while Spearman and the copula test of [10] stay at 0.055 and 0.035; with a dependence change at the same point (C5a, C5b) DOMI detects it in every replicate and localises it in 66--68\%, against 63.5--67\% without the marginal change. The weather analyses of Section 6.7 permute blocks within seasons and use anomalies against an hourly climatology, which remove the annual cycle of the level but not of the variance.

\begin{table*}[!t]
\centering
\caption{Simulation on empirical weather margins (Section~\ref{sec:B19}): rejection rate / localised power under pair-permutation calibration ($K=99$, 200 replicates, rejection at $p\le0.05$, tolerance 60). Kernel and distance statistics studentised, Spearman and the copula statistics raw; CvM and CvM, global as in Table~\ref{tab:B11}. C0 rows give rejection rates under no change. Continuous: margins interpolated between order statistics, no ties; rounded: observations rounded to 0.1\,$^\circ$C and 1\%. C4--C5b were run in the continuous variant only.}
\footnotesize
\setlength{\tabcolsep}{4pt}
\begin{tabular}{l l cccccc}
\toprule
Variant & Condition & DOMI & HSIC & dCor & Spearman & CvM & CvM, global \\
\midrule
continuous & C0 no change & 0.070 & 0.070 & 0.080 & 0.075 & 0.055 & 0.035 \\
 & C1 margins, winter to summer & 1.000/0.000 & 1.000/0.000 & 1.000/0.000 & 1.000/0.000 & 1.000/0.000 & 1.000/1.000 \\
 & C2a dependence ends & 1.000/0.635 & 1.000/0.615 & 0.925/0.720 & 0.090/0.020 & 0.450/0.325 & 0.080/0.010 \\
 & C2b dependence begins & 1.000/0.670 & 1.000/0.600 & 0.945/0.675 & 0.075/0.005 & 0.475/0.315 & 0.065/0.010 \\
 & C3a ends, winter to summer & 1.000/0.000 & 1.000/0.000 & 1.000/0.000 & 1.000/0.000 & 1.000/0.000 & 1.000/1.000 \\
 & C3b begins, winter to summer & 1.000/0.000 & 1.000/0.000 & 1.000/0.000 & 1.000/0.000 & 1.000/0.000 & 1.000/1.000 \\
 & C4 margins, December to February & 0.110/0.010 & 0.140/0.005 & 0.105/0.010 & 0.055/0.015 & 0.035/0.000 & 1.000/0.965 \\
 & C5a ends, December to February & 1.000/0.680 & 1.000/0.600 & 0.870/0.630 & 0.060/0.005 & 0.510/0.335 & 1.000/0.940 \\
 & C5b begins, December to February & 1.000/0.660 & 1.000/0.590 & 0.940/0.700 & 0.045/0.000 & 0.435/0.325 & 1.000/0.970 \\
\midrule
rounded & C0 no change & 0.070 & 0.075 & 0.080 & 0.075 & 0.060 & 0.045 \\
 & C1 margins, winter to summer & 1.000/0.000 & 1.000/0.000 & 1.000/0.000 & 1.000/0.000 & 1.000/0.000 & 1.000/1.000 \\
 & C2a dependence ends & 1.000/0.615 & 1.000/0.560 & 0.920/0.705 & 0.090/0.020 & 0.445/0.315 & 0.055/0.000 \\
 & C2b dependence begins & 1.000/0.680 & 1.000/0.680 & 0.935/0.680 & 0.075/0.005 & 0.490/0.330 & 0.055/0.005 \\
 & C3a ends, winter to summer & 1.000/0.000 & 1.000/0.000 & 1.000/0.000 & 1.000/0.000 & 1.000/0.000 & 1.000/1.000 \\
 & C3b begins, winter to summer & 1.000/0.000 & 1.000/0.000 & 1.000/0.000 & 1.000/0.000 & 1.000/0.000 & 1.000/1.000 \\
\bottomrule
\end{tabular}
\label{tab:B14}
\end{table*}

\subsection{A representation-learning baseline}\label{sec:B20}

MC-TIRE [38] learns features of a multi-channel series with parallel autoencoders in the time and the frequency domain, separates them into a cross-channel part (A\&S) and a channel-specific part (B), and marks change points at peaks of the smoothed distance between the features of adjacent windows. We ran the authors' code unchanged, with the settings of their multi-channel experiments (window 40, discrete Fourier transform on 30 points, rank 1, at most 200 epochs with early stopping) and a fixed seed per series; the channels are $X$ and $Y$. We keep the authors' network and settings but replace their peak-merging step, which has no false-alarm control, by a single-break scan of the A\&S curve and of the sum of the A\&S and B curves, with the matched-null Monte Carlo threshold of BMCTC in Section~\ref{sec:B13}, on the same null and alternative replicates (Table~\ref{tab:B15}). The authors' native peak-merging procedure is therefore not evaluated here.

On every dependence change of the table MC-TIRE localises the break in at most 2.2\% of replicates and detects a change in at most 10.8\%, against localised powers of 0.35--0.87 for DOMI on S2 and 0.55 on S4 at $r=0.85$ (Table 1). Under M1 its A\&S curve rejects in 28.6--31.6\% of replicates. The comparison uses serially independent single-break series of length 600; the method was designed for longer records with many change points.

\begin{table}[!t]
\centering
\caption{MC-TIRE [38] on the designs of Table 1: detection rate / localised power at a false-alarm rate of 0.05 under the calibration of Table 1 (500 replicates), for the cross-channel curve (A\&S) and the sum of the cross-channel and channel-specific curves, raw and studentised. The M1 row gives false-alarm rates.}
\footnotesize
\setlength{\tabcolsep}{3pt}
\begin{tabular}{l cccc}
\toprule
 & \multicolumn{2}{c}{A\&S} & \multicolumn{2}{c}{A\&S + B} \\
Setting & raw & stud. & raw & stud. \\
\midrule
S1 $r{=}.35$ & 0.038/0.022 & 0.036/0.022 & 0.026/0.002 & 0.032/0.000 \\
S2 $a{=}.7$ & 0.108/0.010 & 0.104/0.010 & 0.018/0.000 & 0.014/0.000 \\
S2 $a{=}.9$ & 0.104/0.004 & 0.104/0.004 & 0.008/0.002 & 0.008/0.000 \\
S3 $\tau{=}.5$ & 0.084/0.006 & 0.088/0.006 & 0.094/0.012 & 0.088/0.008 \\
S4 $r{=}.7$ & 0.036/0.000 & 0.032/0.000 & 0.080/0.002 & 0.076/0.002 \\
S4 $r{=}.85$ & 0.018/0.000 & 0.014/0.000 & 0.032/0.002 & 0.032/0.002 \\
G1 $\nu{=}2$ & 0.106/0.006 & 0.108/0.006 & 0.090/0.010 & 0.088/0.014 \\
G2 $\tau{=}.2$ & 0.042/0.018 & 0.032/0.016 & 0.032/0.004 & 0.042/0.010 \\
G4 $a{=}.5$ & 0.084/0.010 & 0.070/0.012 & 0.020/0.002 & 0.014/0.002 \\
G6 $b{=}.3$ & 0.070/0.014 & 0.070/0.014 & 0.072/0.012 & 0.060/0.010 \\
G6 $b{=}.5$ & 0.106/0.012 & 0.106/0.008 & 0.070/0.006 & 0.070/0.006 \\
M1 $s{=}2$ (FA) & 0.286 & 0.316 & 0.082 & 0.086 \\
\bottomrule
\end{tabular}
\label{tab:B15}
\end{table}

\begin{table*}[!t]
\centering
\caption{Test-computation time: median wall-clock seconds of the pair-permutation test ($K=99$) from ranks through the p-value, excluding data and permutation-index generation, over 10 data sets of S2 at $a=0.7$, one thread per test, four at a time on pinned performance cores of an Intel Core i9-12900 alongside another job (Section~\ref{sec:B9}). Dense: every split from $n/10$ to $9n/10$; fixed: 161 split fractions. $>$30 min: time limit reached; mem: above 16 GB; --: not run after a failure at a smaller $n$. CvM and CvM, global as in Table~\ref{tab:B11}.}
\footnotesize
\setlength{\tabcolsep}{4pt}
\begin{tabular}{l ccc ccccc}
\toprule
 & \multicolumn{3}{c}{dense grid, $n$} & \multicolumn{5}{c}{fixed grid, $n$} \\
Statistic & 600 & 2400 & 9600 & 600 & 2400 & 9600 & 38{,}400 & 76{,}800 \\
\midrule
DOMI & 16.8 & 78.1 & 341 & 5.38 & 5.32 & 5.89 & 7.59 & 10.4 \\
Gram & 232 & $>$30 min & -- & 254 & $>$30 min & -- & -- & -- \\
HSIC & 1.23 & 3.76 & 14.5 & 0.754 & 1.03 & 1.80 & 5.88 & 12.4 \\
dCor & 4.46 & 50.8 & 755 & 4.63 & 54.4 & 760 & mem & -- \\
Spearman & 1.24 & 5.18 & 20.8 & 0.464 & 0.647 & 0.861 & 0.922 & 1.84 \\
CvM & 8.97 & 440 & $>$30 min & 7.03 & 404 & $>$30 min & -- & -- \\
CvM, global & 1.45 & 19.6 & 397 & 1.28 & 18.8 & 383 & mem & -- \\
\bottomrule
\end{tabular}
\label{tab:B16}
\end{table*}

\subsection{Sign reversal of a correlation}\label{sec:B21}

A Gaussian copula with correlation $r$ and one with $-r$ have the same mutual information, and the Gram form of DOMI takes the same value on both, since reversing the sign of one variable reflects its ranks (Lemma 4(ii)). This design changes the correlation from $+r$ to $-r$ at $\tau=300$ ($n=600$, standard normal margins, $r=0.35$ and 0.7) under the calibration of Table 1, with $+r$ held throughout as the null (500 null and 500 alternative replicates; Table~\ref{tab:B17}).

The two forms of DOMI do not localise the break. At $r=0.7$ both detect a change in every replicate and localise it in none; the studentised maximiser has a median of 394 for the random-feature form and 397 for the Gram form (interquartile ranges 169--451 and 167--448, all replicates). At $r=0.35$ they detect in 28 and 31\% of replicates and localise in 0.2\%. The segments at the break carry equal DOMI, so the difference statistic vanishes there; a straddling segment mixes the two signs and, as the marginals are fixed, by convexity of the relative entropy carries no more, so the curve peaks away from the break. HSIC and distance correlation do not localise it either. Spearman, the copula test of [10] and the correlation CUSUM localise the break in 98--100\% of replicates. The combined statistic of Section 4.5 localises it in 97.4 and 100\% of replicates, close to Spearman, while retaining detection of the correlation-free change S2, where Spearman has little power: its rejection rates are 0.23 at \(a=0.7\) and 0.97 at \(a=0.9\), against 0.37 and 0.97 for DOMI (Section~\ref{sec:B18}).

\begin{table}[!t]
\centering
\caption{Sign reversal of a Gaussian correlation from $+r$ to $-r$ at $\tau=300$ (Section~\ref{sec:B21}): detection rate / localised power under the calibration of Table 1 (500 null replicates with $+r$ throughout, 500 alternative replicates). Kernel and distance statistics studentised; Spearman, copula and CUSUM statistics raw; the combined statistic studentised. CvM and CvM, global as in Table~\ref{tab:B11}.}
\footnotesize
\setlength{\tabcolsep}{4pt}
\begin{tabular}{l cc}
\toprule
Statistic & $r=0.35$ & $r=0.7$ \\
\midrule
DOMI & 0.280/0.002 & 1.000/0.000 \\
DOMI, Gram form & 0.312/0.002 & 1.000/0.000 \\
HSIC & 0.172/0.000 & 0.990/0.000 \\
dCor & 0.928/0.000 & 1.000/0.000 \\
Spearman & 1.000/0.980 & 1.000/1.000 \\
CvM & 1.000/0.980 & 1.000/1.000 \\
CvM, global & 0.672/0.398 & 1.000/0.938 \\
correlation CUSUM & 1.000/0.984 & 1.000/0.998 \\
combined (Section 4.5) & 1.000/0.974 & 1.000/1.000 \\
\bottomrule
\end{tabular}
\label{tab:B17}
\end{table}

\subsection{A US financial panel under a fixed protocol}\label{sec:B22}

\emph{Protocol.} Before the gold and dollar-index series were obtained, we fixed in a dated file the data, pairs, primary statistics, calibration, multiplicity rule and decision rules of this analysis, and recorded its SHA-256 value (beginning \texttt{cd5f49d2220eef05}) at 08:42 KST on 2026-09-30. The protocol is not deposited in a public registry. The stage-two details were fixed in an addendum (SHA-256 beginning \texttt{29d9e704c6ad2714}, 13:52 KST the same day) after the stage-one p-values of the gold and Treasury pairs had been seen and before the dollar-index result, the sensitivity draws or any stage-two result were opened; a log records what was seen when, and the protocol, its addendum and the log are available from the authors on request. The Treasury pair is not a blind test: an exploratory whole-record scan of it (\(K=199\)) preceded the protocol and gave a DOMI p-value of 0.02 with the maximiser at 2013-05-08 and a Spearman p-value of 0.005 at 2020-12-23. Oil was excluded before any download.

\emph{Data and procedure.} Daily closing levels from Yahoo Finance on the trading days common to both series of a pair, 2011-08-23 to 2026-08-20: the S\&P 500 (\textasciicircum{}GSPC), the SPDR Gold Shares ETF (GLD) and the US Dollar Index (\mbox{DX-Y.NYB}) as log returns, and the 10-year Treasury yield (\textasciicircum{}TNX) as first differences. For each pair the exchangeability diagnostic of Section 4.3 on the whole record rejected pair exchangeability and set the block length; one whole-record single-break test then computed DOMI (random-feature form, \(D=8\)), Spearman and the combined statistic of Section 4.5 (the larger of the two studentised components, calibrated as such, so its p-value can exceed that of either component) from the same \(K=9999\) block permutations, with HSIC recorded descriptively. The primary statistic was DOMI for gold, whose average correlation with stocks is near zero with a sign that varies by regime, and the combined statistic for the Treasury yield and the dollar, whose expected changes are of correlation. The three primary p-values were corrected by Benjamini--Hochberg at \(q=0.10\). Each surviving pair entered stage two (Section 4.7) at the maximiser of its primary statistic, on the largest centred window holding a whole number of blocks, with \(K=999\) block permutations for the dependence and marginal-scale tests, Benjamini--Hochberg at \(q=0.10\) over the pairs for the dependence test, a marginal change declared at \(\min(p_X,p_Y)\le0.05\), and 20 further feature draws as sensitivity. The protocol called a change confirmed when its primary test survived the correction and stage two classified it as a change of the dependence, and correlation-free when, in addition, the Spearman test on the same record did not reject at 0.05; for the reason below we do not use that term for the gold pair.

\emph{Results} (Table~\ref{tab:B20}). All three primary tests reject and survive the correction. For gold, DOMI rejects and so does HSIC, while the Spearman scan does not; over 20 further permutation draws the DOMI p-value stays within 0.004--0.015. In stage two all three breaks are classified as dependence changes with no detected scale change. The Treasury break rests on 16 blocks, with \(p\le0.05\) in all 20 feature draws. The gold and dollar breaks lie near the start of the record, so their centred windows hold four and six blocks; both reach \(q=0.098\), inside the protocol's 0.10 but not the 0.05 used for the weather candidates, and the gold p-value exceeds 0.05 in all 20 feature draws. With four blocks the marginal-scale test has little power, so the absence of a detected scale change for gold is weak evidence: within the window the annualised volatility of gold returns is 0.14 before the break and 0.22 after it.

\emph{Reading the gold result.} Within the stage-two window the rank correlation of stock and gold returns falls from 0.31 before the break to \(-0.04\) after it, and the bias-corrected DOMI (Section 3.3) from 0.040 to 0.003, so the change is not free of correlation. The non-rejection of the whole-record Spearman scan is consistent with the stock--gold correlation changing sign from year to year (Table~\ref{tab:B21}; from \(-0.28\) to 0.39 across calendar years, 0.05 over the whole record), under which permuted records often show differences in correlation as large as the observed one. DOMI measures the strength of the dependence irrespective of its sign; in 2011 and 2012 it exceeds every calendar year from 2013 to 2025. The DOMI maximisers of all three pairs fall in March to May 2013; we report this as a descriptive observation and draw no causal conclusion.

\begin{table*}[!t]
\centering
\caption{BMCTC under a $2\times2$ variation of its input and kernel rule (Section~\ref{sec:B13}): localised power on S4 and detection rate under M1 (false alarm), raw / studentised, under the calibration of Table~\ref{tab:B9} (500 replicates, same replicates and seeds). Kernel rule BMCTC: each 60-point window standardised, Gaussian width $1.06\cdot60^{-1/5}$; DOMI: no window standardisation, the median-heuristic width of DOMI's random features. The first two rows are BMCTC itself (Table~\ref{tab:B9}).}
\footnotesize
\setlength{\tabcolsep}{4pt}
\begin{tabular}{llc cc cc cc}
\toprule
 & & & \multicolumn{2}{c}{S4 $r{=}0.7$} & \multicolumn{2}{c}{S4 $r{=}0.85$} & \multicolumn{2}{c}{M1 $s{=}2$ (FA)} \\
Input & Kernel rule & $\alpha$ & raw & stud. & raw & stud. & raw & stud. \\
\midrule
raw & BMCTC & 1.01 & 0.33 & 0.23 & 0.86 & 0.80 & 0.06 & 0.07 \\
raw & BMCTC & 2 & 0.57 & 0.51 & 0.94 & 0.93 & 0.05 & 0.04 \\
raw & DOMI & 1.01 & 0.25 & 0.17 & 0.69 & 0.63 & 0.99 & 0.99 \\
raw & DOMI & 2 & 0.62 & 0.46 & 0.90 & 0.88 & 0.57 & 0.49 \\
ranks & BMCTC & 1.01 & 0.14 & 0.09 & 0.80 & 0.67 & 0.44 & 0.44 \\
ranks & BMCTC & 2 & 0.15 & 0.10 & 0.77 & 0.69 & 0.13 & 0.10 \\
ranks & DOMI & 1.01 & 0.21 & 0.10 & 0.83 & 0.69 & 0.15 & 0.14 \\
ranks & DOMI & 2 & 0.23 & 0.16 & 0.82 & 0.73 & 0.18 & 0.19 \\
\bottomrule
\end{tabular}
\label{tab:B18}
\end{table*}

\begin{table*}[!t]
\centering
\caption{Single-break $p$-values on the daily 2021--2022 stock--bond window of Section 6.8 ($n=503$; Section~\ref{sec:B18}): median (range) over 20 independent permutation draws of $K=999$ each, all on the same observed window, so the spread is Monte Carlo variation and not variation over independent data; a median of 20 draws is the midpoint of the two central values, hence the occasional fourth decimal. Block permutation at $b=20$ and $b=50$ leaves 25 and 10 blocks. Pair permutation, which the diagnostic of Section 4.3 rules out for this window, is shown for comparison. CvM, global also responds to changes confined to the margins (Table~\ref{tab:B11}).}
\footnotesize
\begin{tabular}{lccc}
\toprule
Statistic & pair & block, $b=20$ & block, $b=50$ \\
\midrule
DOMI & 0.268 (0.223--0.303) & 0.231 (0.202--0.272) & 0.3605 (0.298--0.408) \\
combined (Section 4.5) & 0.016 (0.008--0.022) & 0.1475 (0.095--0.183) & 0.4335 (0.358--0.493) \\
Spearman & 0.006 (0.003--0.010) & 0.057 (0.035--0.078) & 0.2965 (0.241--0.340) \\
HSIC & 0.026 (0.017--0.049) & 0.0195 (0.009--0.033) & 0.078 (0.050--0.105) \\
CvM, global & 0.001 (0.001--0.002) & 0.001 (0.001--0.001) & 0.002 (0.001--0.003) \\
\bottomrule
\end{tabular}
\label{tab:B19}
\end{table*}

\begin{table*}[!t]
\centering
\caption{US financial panel (Section~\ref{sec:B22}; Section 6.8), daily, 2011-08-23 to 2026-08-20. Upper block: whole-record single-break tests under block permutation at the diagnostic's block length \(b\) in trading days (number of blocks), \(K=9999\): p-value of the primary statistic, its range over 20 further permutation draws of \(K=999\), Benjamini--Hochberg \(q\) over the three pairs, p-values of DOMI, Spearman and HSIC from the same replicas, and the maximiser of the primary statistic. Lower block: stage two (Section 4.7) on the largest centred window holding whole blocks: number of blocks, p-value of the dependence test (\(K=999\)) and its Benjamini--Hochberg \(q\) over the three pairs, p-values of the marginal-scale test for the S\&P 500 (\(p_X\)) and for the second series (\(p_Y\)), range of the dependence p-value over 20 feature draws and the number of draws with \(p\le0.05\), and the class of Algorithm 2 (dependence: a dependence change with no detected scale change; Figure~1).}
\footnotesize
\setlength{\tabcolsep}{3.5pt}
\begin{tabular}{llccccccc c}
\toprule
Pair & Primary & \(b\) (blocks) & \(p\) & Range & \(q\) & DOMI & Spearman & HSIC & Break \\
\midrule
S\&P 500--gold & DOMI & 134 (28) & 0.0080 & 0.004--0.015 & 0.0120 & 0.0080 & 0.2289 & 0.0225 & 2013-03-12 \\
S\&P 500--10-year yield & combined & 173 (21) & 0.0204 & 0.007--0.024 & 0.0204 & 0.0226 & 0.0001 & 0.0243 & 2020-12-07 \\
S\&P 500--dollar index & combined & 133 (28) & 0.0053 & 0.002--0.015 & 0.0120 & 0.0053 & 0.0227 & 0.0096 & 2013-04-01 \\
\bottomrule
\end{tabular}

\vspace{4pt}
\begin{tabular}{llcccccccl}
\toprule
Pair & Window & Blocks & \(p_{\mathrm{dep}}\) & \(q\) & \(p_X\) & \(p_Y\) & Range & \(p\le0.05\) & Class \\
\midrule
S\&P 500--gold & 2012-02-14 to 2014-04-02 & 4 & 0.075 & 0.098 & 0.168 & 0.867 & 0.075--0.150 & 0 of 20 & dependence \\
S\&P 500--10-year yield & 2015-06-09 to 2026-06-11 & 16 & 0.004 & 0.012 & 0.553 & 0.077 & 0.002--0.007 & 20 of 20 & dependence \\
S\&P 500--dollar index & 2011-08-25 to 2014-10-27 & 6 & 0.098 & 0.098 & 0.104 & 0.207 & 0.019--0.098 & 5 of 20 & dependence \\
\bottomrule
\end{tabular}
\label{tab:B20}
\end{table*}

\begin{table}[!t]
\centering
\caption{Calendar-year Spearman correlation (\(\rho_S\)) and bias-corrected DOMI (Section 3.3; \(D=8\), ranks within the year) of daily returns for the pairs of Table~\ref{tab:B20} (Section~\ref{sec:B22}). 2011 starts on 2011-08-23 and 2026 ends on 2026-08-20. The bias-corrected estimate can be slightly negative. Descriptive; not used for any decision.}
\footnotesize
\setlength{\tabcolsep}{3.5pt}
\begin{tabular}{l cc cc cc}
\toprule
 & \multicolumn{2}{c}{S\&P 500--gold} & \multicolumn{2}{c}{S\&P 500--yield} & \multicolumn{2}{c}{S\&P 500--dollar} \\
Year & \(\rho_S\) & DOMI & \(\rho_S\) & DOMI & \(\rho_S\) & DOMI \\
\midrule
2011 & 0.23 & 0.044 & 0.71 & 0.192 & $-$0.77 & 0.235 \\
2012 & 0.35 & 0.063 & 0.58 & 0.123 & $-$0.50 & 0.087 \\
2013 & 0.03 & 0.006 & 0.16 & 0.007 & $-$0.10 & 0.004 \\
2014 & $-$0.18 & 0.011 & 0.45 & 0.092 & 0.12 & $-$0.002 \\
2015 & $-$0.08 & 0.009 & 0.39 & 0.054 & 0.15 & 0.007 \\
2016 & $-$0.28 & 0.031 & 0.34 & 0.047 & $-$0.01 & $-$0.003 \\
2017 & $-$0.24 & 0.015 & 0.34 & 0.047 & 0.21 & 0.013 \\
2018 & 0.06 & 0.006 & 0.25 & 0.024 & $-$0.11 & 0.018 \\
2019 & $-$0.19 & 0.009 & 0.44 & 0.061 & 0.05 & $-$0.001 \\
2020 & 0.07 & 0.005 & 0.41 & 0.065 & $-$0.08 & 0.013 \\
2021 & 0.17 & 0.003 & 0.10 & 0.004 & $-$0.26 & 0.023 \\
2022 & 0.16 & 0.006 & $-$0.11 & 0.007 & $-$0.46 & 0.061 \\
2023 & 0.03 & 0.001 & $-$0.04 & 0.002 & $-$0.24 & 0.012 \\
2024 & 0.26 & 0.027 & $-$0.15 & 0.019 & $-$0.18 & 0.006 \\
2025 & $-$0.06 & 0.011 & 0.01 & 0.012 & $-$0.03 & 0.008 \\
2026 & 0.39 & 0.061 & $-$0.37 & 0.056 & $-$0.42 & 0.061 \\
\bottomrule
\end{tabular}
\label{tab:B21}
\end{table}